\documentclass[authoryear,12pt]{article}

\usepackage[T1]{fontenc}
\usepackage[utf8]{inputenc}
\usepackage[english]{babel}

\usepackage[resetlabels]{multibib}
\newcites{oa}{References}
\usepackage[round]{natbib}

\usepackage{mathpazo}
\usepackage{setspace}
\usepackage[top=3cm, bottom=3cm, left=3cm, right=3cm]{geometry}
\usepackage{fancyhdr}
\usepackage{ragged2e}
\usepackage{changepage}
\usepackage{multicol}
\usepackage{rotating}
\usepackage{pdflscape}
\usepackage{afterpage}

\usepackage{graphicx}
\usepackage{pdfpages}
\usepackage{adjustbox}
\usepackage{pgfplots}
\pgfplotsset{compat=1.18}
\usepackage[tikz]{bclogo}

\usepackage{comment}
\usepackage{paralist}
\usepackage{enumitem}
\usepackage{listings}
\usepackage{nth}
\usepackage{lettrine}
\usepackage{textcomp}
\usepackage{pifont}
\usepackage{fancybox}

\usepackage[english,ruled,vlined,resetcount]{algorithm2e}

\SetKwComment{Comment}{/* }{ */}
\SetAlgoCaptionSeparator{}

\usepackage{amsmath}
\usepackage{amsthm}
\usepackage{mathtools}
\usepackage{amsfonts}
\usepackage{mathrsfs}
\usepackage{upgreek}
\usepackage{mathabx}
\usepackage{wasysym}
\usepackage{dsfont}
\usepackage{bbm}

\usepackage{caption}
\usepackage{multirow}
\usepackage{booktabs}
\usepackage{supertabular}
\usepackage{longtable}
\usepackage{tabularx}
\usepackage{threeparttable}
\usepackage{array}
\usepackage{colortbl}
\usepackage{dcolumn}
\usepackage{siunitx}

\newcolumntype{P}[1]{>{\centering\arraybackslash}p{#1}}
\newcolumntype{d}[1]{D..{#1}}

\usepackage{footnote}

\usepackage[dvipsnames]{xcolor}

\definecolor{deepblue}{RGB}{0,0,110}
\definecolor{deepred}{RGB}{110,0,0}
\definecolor{colari}{rgb}{0.7,0,0.7}
\definecolor{coland}{rgb}{0,0.7,0.4}

\usepackage{xurl}
\usepackage[
    pdftitle={Asymmetries in Peer Effects},
    pdfauthor={Aristide Houndetoungan, Mathieu Lambotte},
    colorlinks=true,
    allcolors=blue
]{hyperref}

\usepackage[open,openlevel=2,atend,numbered]{bookmark}
\usepackage[nameinlink]{cleveref}

\hypersetup{
    linkcolor=deepblue,
    citecolor=deepblue,
    urlcolor=deepblue
}

\usepackage{titling}
\usepackage{xpatch}
\usepackage{titlesec}
\usepackage{etoolbox}

\titleformat{\section}
    {\large\bfseries}
    {\thesection}
    {1em}
    {}

\titleformat{\subsection}
    {\large\bfseries}
    {\thesubsection}
    {1em}
    {}

\titleformat{\paragraph}
    {\normalfont\normalsize\bfseries}
    {\theparagraph}
    {0.7em}
    {}

\titlespacing*{\paragraph}
    {0pt}
    {1ex plus 1ex minus .2ex}
    {1pt}

\makeatletter

\patchcmd{\@afterheading}
    {\clubpenalty \@M}
    {\clubpenalty 0}
    {}{}

\patchcmd{\equation}
    {\@beginparpenalty\predisplaypenalty}
    {}
    {}{}

\newtheorem{proposition}{Proposition}
\newtheorem{lemma}{Lemma}

\newtheorem{assumption}{Assumption}

\newcommand{\neutralize}[1]{%
    \expandafter\let\csname c@#1\endcsname\count@
}

\let\endtitlepage\relax

\newenvironment{mytitlepage}
    {\begin{titlepage}\def\@thanks{}}
    {\end{titlepage}}

\xpatchcmd{\titlepage}
    {\setcounter{page}\@ne}
    {}
    {}{}

\xpatchcmd{\endtitlepage}
    {\setcounter{page}\@ne}
    {}
    {}{}

\newcommand{\ostar}{\mathbin{\mathpalette\make@circled\star}}

\newcommand{\make@circled}[2]{%
    \ooalign{%
        $\m@th#1\smallbigcirc{#1}$\cr
        \hidewidth$\m@th#1#2$\hidewidth\cr
    }%
}

\newcommand{\smallbigcirc}[1]{%
    \vcenter{%
        \hbox{%
            \scalebox{0.77778}{$\m@th#1\bigcirc$}%
        }%
    }%
}

\makeatother

\DeclareMathOperator{\plim}{plim}

\usepackage{authblk}
\newcommand{\TITLE}{Asymmetries in Peer Effects}
\title{\vspace{-1cm} \TITLE\footnote[1]{\fontsize{10pt}{10pt} We are grateful to Vincent Boucher for helpful comments and discussions.
We acknowledge financial support from the Social Sciences and Humanities Research Council (SSHRC) under Grant CH150174.\\
This research uses data from Add Health, a program directed by Kathleen Mullan Harris and designed by J. Richard Udry, Peter S. Bearman, and Kathleen Mullan Harris. Special acknowledgment is given to Ronald R. Rindfuss and Barbara Entwisle for assistance in the original design. Information on how to obtain Add Health data files is available on the Add Health website.\\\ \vspace{-.2cm}\\
Email addresses: \href{mailto:ahoundetoungan@ecn.ulaval.ca}{ahoundetoungan@ecn.ulaval.ca} (A. Houndetoungan),\\ \href{mailto:mathieu.lambotte@univ-rennes.fr}{mathieu.lambotte@univ-rennes.fr} (M. Lambotte)\\\ \vspace{-.2cm}\\
An R package, including all replication codes, is available at: \url{https://github.com/MathieuLambotte/AsyPeer}.} \vspace{-.5cm}}

\author[a]{Aristide Houndetoungan}
\author[b]{Mathieu Lambotte}
\affil[a]{\normalsize\emph{Laval University, Quebec, Canada}\vspace{-2pt}}
\affil[b]{\normalsize\emph{Rennes University, CNRS, CREM – UMR6211, F-35000 Rennes France}\vspace{-2pt}}

\date{\small August 2026}
\begin{document}
\setlength{\abovedisplayskip}{6pt}  % default is ~12pt
\setlength{\belowdisplayskip}{6pt}

\begin{mytitlepage}
\maketitle

\vspace{-1cm}

\begin{abstract}
\vspace{-0.5cm}
\singlespacing
\small \noindent 
Individuals are often influenced by their peers because deviating from prevailing behavior entails social costs. However, existing peer effects models typically assume that individuals respond similarly to peers who perform better or worse than they do. This paper introduces a novel structural model of asymmetric peer effects in which conformity incentives depend on whether individuals perform below or above each of their peers. We establish that the model admits a unique equilibrium and show that its parameters can be identified and estimated through simple moment conditions. Applying our method to several student outcomes, we uncover strong evidence of asymmetries in peer effects. We then demonstrate that these asymmetries are highly policy-relevant by studying targeted interventions under budget constraints. Ignoring asymmetries leads to inefficient treatment allocation and substantial welfare losses, reducing welfare to levels comparable to those in a benchmark without social interactions.

\vspace{0.5cm}


\end{abstract}
\end{mytitlepage}

\clearpage
\section{Introduction}
Interactions with peers and friends shape many economic and social behaviors. These social influences, commonly referred to as peer effects, have long been central to economics because they help explain how behaviors spread and inform policy design \citep[see][]{Manski1993, de2017econometrics, bramoulle2020peer, zenou2025peer}. These influences often operate through social comparison, generating concerns about status, competition, social approval, or stigmatization \citep[see][]{bernheim1994theory,akerlof1997,luttmer2005neighbors,benabou2006incentives,akerlof2000economics}.

A central tenet of the social comparison literature is that the direction of comparison matters: individuals often react differently to peers who perform above them than to peers who perform below them \citep{festinger1954theory, wills1981downward, card2012inequality}. For instance, a student's study effort may depend on whether her friends study more or less than she does \citep{feld2017understanding}. Yet existing models of peer effects largely ignore this distinction by treating peers who outperform an individual and those who underperform her as exerting the same influence. The standard linear-in-means model, for example, treats deviations above and below peers symmetrically \citep{blume_linear_2015}. As a result, this model may fail to capture how individuals actually respond to their peers and may therefore provide misleading guidance for policy design. Recent contributions that introduce heterogeneity in peer effects through the distribution of peers' outcomes \citep[e.g.,][]{boucher2023, herstad2026identification} likewise do not accommodate asymmetries in peer effects.

In this paper, we bridge the gap between social comparison theory and the peer effects literature by introducing a new model in which individuals respond differently to peers who perform above them than to those who perform below them. We rationalize this asymmetry through a game of conformity in which upward and downward comparisons affect preferences differently. We establish the existence and uniqueness of equilibrium and derive an econometric framework to identify and estimate both asymmetric effects. We show that the standard linear-in-means peer effects model estimates a weighted average of the underlying asymmetric effects, with weights that may be negative. As a result, the estimated symmetric peer effect may fall outside the range defined by the two asymmetric effects. Applying the proposed model to several student outcomes, we find pervasive evidence of asymmetric peer influence: for some outcomes, peers who perform above the individual exert strong effects, while those who perform below have little or no influence; for other outcomes, the reverse occurs. We show that these asymmetries are highly policy-relevant by studying a treatment allocation problem under limited resources. We find that ignoring asymmetry in peer effects leads to an inefficient treatment allocation, resulting in substantial losses in spillovers and welfare.

In our model, individual preferences are characterized by a utility function in which the influence of each friend depends on the friend's \textit{status} as higher- or lower-performing relative to the individual's performance. We impose no restrictions on which type of peer exerts a stronger influence. Instead, we allow the direction of influence to be determined by the data. This flexibility, however, introduces several complications in the game. Unlike in many peer effects models, the outcome of individual $i$, denoted by $y_i$, cannot be written explicitly as a function of peer outcomes, because peer statuses depend not only on the distribution of peer outcomes, but also on $y_i$ itself. As a result, the best-response function is implicit. Moreover, this function is not everywhere differentiable, as peer statuses switch discretely when $y_i$ crosses the outcome of one of her peers. %Consequently, standard approaches for solving peer effect models do not apply directly to our setting. PLS, REMOVE MY CHANGES IF YOU DON'T LIKE IT
To show the uniqueness of the Nash equilibrium, we exploit the game's topological structure to establish that the best-response mapping remains a contraction under reasonable conditions.

We develop an econometric framework to identify both peer effect parameters. Our identification conditions extend those of the standard model \citep{Bramoulle2009} to accommodate asymmetry in peer effects. In practice, estimating our model requires exogenous predictions of peer statuses as instruments. We rely on flexible machine learning methods to construct such instruments. To establish consistency and asymptotic normality despite the use of generated instruments, our inference builds on recent developments in double/debiased machine learning \citep{Chernozhukov2018Double}. We conduct a Monte Carlo study showing that our approach performs well in finite samples. The simulations also illustrate, through concrete examples, how the standard symmetric model may estimate peer effects that fall outside the range of the asymmetric effects.

Using data from the National Longitudinal Study of Adolescent to Adult Health (Add Health), we show that many student activities exhibit asymmetric peer effects. For three of the four outcomes we study---namely, smoking, fighting, and optimism---we find significant asymmetry, although its direction differs across outcomes. For smoking, both types of peers exert strong effects, but peers who smoke more than the individual exert a larger influence. For fighting, individuals tend to conform strongly to more aggressive peers, while less aggressive peers exert little or no influence. For optimism, individuals tend to conform to more pessimistic peers, while more optimistic peers exert little or no influence. For drinking, by contrast, we find evidence that peer effects are approximately symmetric.

%Our empirical results provide new insights into the empirical literature. For socially undesirable or risky behaviors (smoking, fighting, and drinking) students tend to be more responsive to friends who exhibit higher levels of the behavior, although the asymmetry is not statistically significant for drinking. This pattern is consistent with peer-norm conformity or reference-group influence, whereby adolescents adjust their behavior toward salient peers to gain approval or avoid exclusion. For optimism, however, the stronger influence of more pessimistic friends is consistent with downward-comparison mechanisms, whereby exposure to less optimistic peers shapes self-assessments, expectations, or affect more strongly than exposure to more optimistic peers.

We further examine the policy implications of these empirical findings by studying targeting policies in which, because of limited resources, only a subset of students is treated, with the objective of reducing smoking, fighting, or drinking, or increasing optimism \citep{ballester2006,galeotti2020}. We find that ignoring asymmetry leads to inefficient treatment allocation, generating substantial losses in spillovers. Indeed, even when the symmetric model identifies a strong peer effect, using this model to select students for treatment may nevertheless generate little or no spillover effects, similar to those obtained in a benchmark without social interactions. This occurs particularly for fighting and optimism, where peer effects exhibit large asymmetries. By ignoring these asymmetries, the symmetric model often selects students who exert little influence under the true asymmetric model. The implemented policy therefore fails to fully leverage interactions among individuals. We show that this inefficiency generates substantial welfare losses.

This paper primarily contributes to the literature on how social interactions shape individual behavior \citep[see][]{de2017econometrics, bramoulle2020peer, zenou2025peer}. Recent developments highlight the importance of heterogeneity in peer effects, either through the distribution of peer outcomes \citep[see][]{carrell2010sex, boucher2023, herstad2026identification} or through predefined groups of peers \citep{comola2025heterogeneous, houndetoungan2026count}. We contribute to this literature by developing a new structural model that identifies distinct conformity effects from higher- and lower-performing friends. Our model differs from existing approaches because heterogeneity depends not only on the distribution of peer outcomes, but also on the individual's own outcome. This allows us to empirically test behavioral mechanisms that have long been emphasized in the theory of social comparison. Some papers also document heterogeneous peer effects by interacting individuals' performance types, such as high- or low-performing status, with those of their peers \citep[see, e.g.,][]{luttmer2005neighbors,feld2017understanding,carrell2013natural}. However, these analyses are reduced-form and are not derived from an underlying behavioral model. One exception is \cite{lambotte2025}, who estimates asymmetric peer effects in a binary-outcome game. In that case, however, the binary nature of the outcome makes the distinction between higher- and lower-performing peers equivalent to estimating separate effects for each of the two actions an individual can take.

We also contribute to the empirical literature on peer effects \citep{sacerdote2014experimental} by documenting asymmetric peer influence across several outcomes. Finally, by showing that ignoring asymmetry leads to suboptimal targeting decisions and substantial welfare losses, we contribute to the literature on optimal treatment allocation \citep{kitagawa2018should, galeotti2020, viviano2025policy}.

The remainder of the paper is organized as follows. Section \ref{sec:micro} presents the microfoundation of the structural model. Section \ref{sec:metrics} presents our identification approach and estimation method, and documents the bias of the standard symmetric model. Section \ref{sec:MCsimu} reports the results of the Monte Carlo study. Section \ref{sec:empirics} provides an empirical application using the Add Health data. Section \ref{sec:conclu} concludes.

\section{Microfoundations}\label{sec:micro}
\noindent We consider a complete information game with $n$ individuals who interact through a network, where $n \geq 2$. The network is characterized by an adjacency matrix $\mathbf{G} = [g_{ij}]_{i,j=1}^n$, where $g_{ij} \geq 0$ measures the strength of the link from $i$ to $j$. We allow for flexible network specifications: the network can be weighted or unweighted, and directed or undirected. For expositional simplicity, however, we focus on row-normalized unweighted networks; that is, $g_{ij} = \frac{1}{n_i}$ if $j$ is a friend (peer) of $i$, and $g_{ij} = 0$ otherwise, where $n_i$ denotes the number of friends of $i$. Additionally,  $g_{ii} = 0$ for all $i$, implying that individuals do not interact with themselves.

Individuals choose their actions (e.g., academic effort), measured by a continuous variable $y_i \in \mathbb{R}$, taking into account their friends' actions. Preferences are represented by the following utility function:
\begin{equation}\label{eq:utility}
 U_i(y_i, \mathbf y_{-i})=\alpha_i y_i - \dfrac{y_i^2}{2} - \mathcal{S}_i(y_i, \mathbf y_{-i}),
\end{equation}
where $\mathbf{y}_{-i} = (y_1, \dots, y_{i-1}, y_{i+1}, \dots, y_n)$ is a vector of other players' actions. 

The utility function \eqref{eq:utility} is additively separable into private and social components. The private component is given by $\alpha_i y_i - \frac{y_i^2}{2}$, where $\alpha_i$ denotes idiosyncratic productivity and $\frac{y_i^2}{2}$ represents the private cost of choosing $y_i$. The term $\mathcal{S}_i(y_i, \mathbf{y}_{-i})$ represents a social distance (or social cost) and captures the effect of differences between individual $i$'s effort and their friends' efforts \citep{akerlof1997}. We introduce asymmetry into the social distance function by allowing social pressure from friends to depend on whether they exert higher or lower levels of effort than $i$.

\subsection{Asymmetric Preferences for Conformity}
Most papers rely on a symmetric social distance to capture conformity peer effects \citep[see, e.g.,][]{blume_linear_2015, ushchev2020}. This social distance is defined as
\begin{equation}\label{eq:sycost}
\mathcal{S}^{sym}_i(y_i, \mathbf y_{-i}) = \dfrac{\beta}{2}\sum_{j \ne i} g_{ij}(y_i - y_j)^2.
\end{equation}

Under this specification, individuals have conformist preferences when $\beta$ is positive and anti-conformist preferences when $\beta$ is negative. This social cost implies that deviations above or below the effort exerted by a friend $j$ yield the same social penalty or benefit. 
However, the social penalty may depend on whether a peer performs better or worse than the individual. To account for this scenario, we introduce an asymmetric social cost:
\begin{equation}\label{eq:asycost}
    \mathcal{S}_i(y_i, \mathbf y_{-i}) = \dfrac{\beta^l}{2}\sum_{j \ne i: y_j \leq y_i} g_{ij}(y_i - y_j)^2 + \dfrac{\beta^h}{2}\sum_{j\ne i: y_j > y_i} g_{ij}(y_i - y_j)^2.
\end{equation}

The asymmetric social distance function allows for flexible mechanisms of peer effects. Individuals may conform more strongly to friends who exert higher levels of the behavior than they do ($\beta^h > \beta^l \geq 0$), or instead to friends who exert lower levels of the behavior ($\beta^l > \beta^h \geq 0$). Which form of asymmetry arises depends on the mechanism underlying the outcome, such as role-model effects, \textit{keeping up with the Joneses} dynamics, or norm formation. The model therefore does not impose \textit{a priori} which type of peer is more influential; instead, it allows the direction of influence to be determined empirically.

Furthermore, our model can also capture situations in which $\beta^h$ and $\beta^l$ have different signs. For example, in certain competitive environments, upward comparison can be discouraging while low-performing peers may boost confidence. This scenario occurs when $\beta^h < 0$ and $\beta^l > 0$.

The asymmetric specification nests the symmetric distance function when $\beta = \beta^h = \beta^l$. We can therefore empirically assess the validity of the symmetric specification by testing whether $\beta^h = \beta^l$. Under the asymmetric social distance, the influence exerted by a friend depends on both an individual's own action and those of her friends. This contrasts with recent models that introduce heterogeneity solely through the distribution of friends' outcomes \cite[e.g.,][]{boucher2023,herstad2026identification}.

\subsection{Equilibrium}
Individuals' optimal choices (effort levels) are determined by maximizing their utility functions. To ensure that this optimization problem has a unique finite solution, we impose a lower bound on $\beta^l$ and $\beta^h$. These conditions guarantee that the utility function is strictly concave.
\begin{assumption}\label{assumption:minbeta}
The parameters $\beta^l$ and $\beta^h$ satisfy $\beta^l > -\dfrac{1}{2}$ and $\beta^h > -\dfrac{1}{2}$.
\end{assumption}
As we show below, the total marginal effect of peers' outcomes on $y_i$ lies between $\frac{\beta^l}{1+\beta^l}$ and $\frac{\beta^h}{1+\beta^h}$. Thus, Assumption \ref{assumption:minbeta} implies that the magnitude of the total marginal peer effect is less than one, which is a standard stability condition in the literature \citep[see][]{blume_linear_2015}.

Under Assumption \ref{assumption:minbeta}, we establish several results in Appendix~\ref{proof:uniqueness}. First, we show in Lemma~\ref{lemma:differentiable} that $U_i(y_i, \mathbf y_{-i})$ is continuously differentiable in $y_i$ despite discrete changes in peer status in Equation \eqref{eq:asycost}.\footnote{We refer to a peer's \textit{status} as high- or low-performing depending on whether the peer's outcome exceeds or falls below that of the individual, respectively.} We also demonstrate in Lemma \ref{lemma:concave} that $U_i(y_i, \mathbf y_{-i})$ is strictly concave in $y_i$. Finally, in Lemma~\ref{lemma:BRFcontinuous}, we show that $i$'s strategy that maximizes their utility function given the choices of other individuals (i.e., their best-response function) is given by the following implicit function:

\begin{equation}\label{eq:yi1}
    y_i = \dfrac{\alpha_i + \beta^l \bar{y}_i^l + \beta^h \bar{y}_i^h}{1 + \beta^l g_i^l + \beta^h g_i^h}.
\end{equation}
In Equation~\eqref{eq:yi1}, $\bar{y}_i^l = \sum_{j \ne i} \mathbbm{1}\{y_j \leq y_i\} g_{ij} y_j$ and $\bar{y}_i^h = \sum_{j \ne i} \mathbbm{1}\{y_j > y_i\} g_{ij} y_j$, where $\mathbbm{1}\{\cdot\}$ denotes the indicator function. Moreover, $g_i^l = \sum_{j \ne i} \mathbbm{1}\{y_j \leq y_i\} g_{ij}$ and $g_i^h = \sum_{j \ne i} \mathbbm{1}\{y_j > y_i\} g_{ij}$ are the shares of low- and high-performing friends, respectively. The outcome $y_i$ is given in an implicit form because $\bar{y}_i^l$ and $g_i^l$ depend on $y_i$ itself. 

For each $i$, there are $2^{n_i}$ possible configurations of peer statuses, where $n_i$ is the number of peers of $i$. Since $\mathbbm{1}\{y_j \leq y_i\}$ is constant within a given configuration, the best-response function is linear within each configuration. Therefore, the best-response function is piecewise linear across configurations. 

The marginal effect of friends' outcomes on $y_i$ is $\frac{\beta^l}{1 + \beta^l}$ when all friends are low-performing, and $\frac{\beta^h}{1 + \beta^h}$ when they are all high-performing. For intermediate configurations, the marginal effect can take up to $2^{n_i}$ distinct values between $\frac{\beta^l}{1 + \beta^l}$ and $\frac{\beta^h}{1 + \beta^h}$. This introduces substantial heterogeneity, unlike in the standard symmetric model where the marginal effect is constant and equal to $\frac{\beta}{1 + \beta}$. However, this heterogeneity comes at the cost that the best-response function is not globally differentiable.

By exploiting the identities $\bar y_i = \bar{y}_i^l + \bar{y}_i^h$ and $g_i^l + g_i^h = 1$, which hold for individuals with at least one friend (non-isolated individuals), we can substitute $\bar{y}_i^l = \bar y_i - \bar{y}_i^h$ and $g_i^l = 1 - g_i^h$ into Equation~\eqref{eq:yi1}.\footnote{We assume that all individuals are non-isolated only for expositional convenience. Isolated individuals have an equilibrium outcome $y_i = \alpha_i$ and can be accommodated within our framework.} It follows that Equation~\eqref{eq:yi1} can be rewritten as:
\begin{equation}\label{eq:yi2}
    y_i = \dfrac{\alpha_i +  \beta^l\bar y_i + (\beta^h - \beta^l)\check y_i}{1 + \beta^l},
\end{equation}
where $\check{y}_{i} = \sum_{j \ne i} g_{ij}\mathbbm{1}\{y_{j} > y_{i}\}(y_{j} - y_{i})$. As we will show later in the econometric analysis, expressing the optimal outcome as in Equation~\eqref{eq:yi2} is important for deriving a reduced-form equation that is linear in parameters.

A strategy profile $\mathbf{y} = (y_i, \dots, y_n)^{\prime}$ is a Nash equilibrium if $y_i$ verifies Equation~\eqref{eq:yi2} for all $i$. We establish the following result. 
\begin{proposition}\label{propo:equilibrium}
Under Assumption \ref{assumption:minbeta}, the game described by the utility function \eqref{eq:utility} has a unique Nash equilibrium $\mathbf y^{\ast} = (y_1^{\ast}, \dots, y_n^{\ast})$, such that $y_i^{\ast}$ verifies Equations \eqref{eq:yi2}.% if individual $i$ is non-isolated and $y_i^{\ast} = \alpha_i$ otherwise. 
\end{proposition}

The proof of Proposition~\ref{propo:equilibrium} is provided in Appendix~\ref{proof:uniqueness:theo}. We establish that the best-response mapping is a contraction; that is, for any strategy profiles $\mathbf{y}$ and $\tilde{\mathbf{y}}$ and any $i$, the uniform Lipschitz condition, $\lvert b_i(\mathbf{y}_{-i}) - b_i(\tilde{\mathbf{y}}_{-i})\rvert \leq \mu \lVert \mathbf{y} - \tilde{\mathbf{y}} \rVert_{\infty}$, holds for some constant $\mu < 1$, where $b_i$ denotes the best-response function of individual $i$.\footnote{For any $\mathbf a = (a_1, \dots, a_n)^{\prime} \in \mathbb R^n$, the infinity norm of $\mathbf a$ is defined as $\displaystyle \lVert \mathbf a \rVert_{\infty} = \max_{i} \lvert a_i \rvert$.} However, since statuses may differ across $\mathbf{y}$ and $\tilde{\mathbf{y}}$, achieving this result is not straightforward. We first partition the strategy space into cells (polyhedra). Within each cell, the statuses of friends remain fixed, allowing us to establish the uniform Lipschitz condition. We then extend this argument to the case in which $\mathbf{y}$ and $\tilde{\mathbf{y}}$ do not belong to the same cell. 

\subsection{Treatment Assignment and Spillover Effects} \label{microfoundation::socialmultiplier}
In the presence of social interactions, interventions that affect the productivity of treated individuals (e.g., targeted subsidy programs) can also affect nonrecipients, thereby generating social multiplier or spillover effects. In this section, we study the effects of productivity shocks in the asymmetric model. In particular, using a treatment allocation problem in which the treatment takes the form of a productivity shock, we show that ignoring asymmetries in peer effects can lead to inefficient treatment allocation. For simplicity, we assume that productivity shocks do not affect the network structure.

In the standard symmetric model, treating all individuals by uniformly increasing their productivity yields the same uniform increase in outcomes and thus generates no social multiplier effects \citep{boucherfortin2015,ushchev2020}. We first show that this result extends to the asymmetric specification.
\begin{proposition}\label{propo:socialmultiplier}
Assume a uniform shock $\bar \alpha$ to any individual's productivity $\alpha_i$, such that $\tilde{\alpha}_i = \alpha_i + \bar \alpha$ is the new productivity level. Let $y_i$ and $\tilde{y}_i$ be the equilibrium outcomes corresponding to $\alpha_i$ and $\tilde{\alpha}_i$, respectively. We have $\tilde{y}_i = y_i + \bar \alpha$ for all $i$.
\end{proposition}

\noindent The proof of Proposition~\ref{propo:socialmultiplier} is provided in Appendix~\ref{append:proof:socialmultiplier}. The result is analogous to that obtained in the symmetric case, because treating all individuals does not alter peer statuses. Consequently, the asymmetric structure of peer effects does not play a role.

However, in many settings, only a subset of individuals can be treated due to budget constraints. This is the case for targeted policies, where interventions benefit only a limited number of individuals \cite[see][]{banerjee2019using, beaman2021can}. Such policies can improve cost-effectiveness when treatment is optimally allocated.

Assume that a social planner aims to maximize aggregate outcomes $\sum_{j=1}^n y_j$ by assigning treatment to a fixed number of individuals, where a budget constraint exogenously determines this number \citep[e.g.,][]{galeotti2020}. As we show below, in the standard symmetric model, the optimal treatment-assignment rule depends only on the network structure and the peer-effect parameter, but not on the initial productivity distribution  \citep[see][]{boucherfortin2015}.

Let $\mathbf I_n$ denote the identity matrix of order $n$ and $\mathbf 1_n$ an $n$-vector of ones. Define $\hat{\boldsymbol \alpha} = (\hat\alpha_1,\dots,\hat\alpha_n)'$, where $\hat\alpha_i = \alpha_i$ if $i$ is isolated and $\hat\alpha_i = \alpha_i/(1+\beta)$ otherwise. The equilibrium outcome of the symmetric model is given by $\mathbf y = \left(\mathbf I_n - \frac{\beta}{1+\beta}\mathbf G\right)^{-1}\hat{\boldsymbol\alpha}$ \citep[see][]{blume_linear_2015}. Consequently,
$$\sum_{j=1}^n y_j = \hat{\boldsymbol\alpha}^\prime \left(\mathbf I_n - \frac{\beta}{1+\beta}\mathbf G'\right)^{-1} \mathbf 1_n.$$
Let $\mathbf e_i \in \mathbb R^n$ denote a vector of zeros except for its $i$-th entry, which is equal to $1$ if individual $i$ is isolated and $1/(1+\beta)$ otherwise.\footnote{The parameter $\beta$ is the single peer effect parameter in the symmetric model (see Equation~\eqref{eq:sycost}).} The marginal effect of increasing $\alpha_i$ by one unit on aggregate outcomes is $\frac{\partial}{\partial \alpha_i}\sum_{j=1}^n y_j = \mathbf e_i' \left(\mathbf I_n - \frac{\beta}{1+\beta}\mathbf G'\right)^{-1} \mathbf 1_n$. Stacking these marginal effects for all individuals yields $\boldsymbol{\Delta} = \mathcal E \left(\mathbf I_n - \frac{\beta}{1+\beta}\mathbf G'\right)^{-1} \mathbf 1_n$, where $\mathcal E = (\mathbf e_1,\dots,\mathbf e_n)'$. Treatment should therefore be assigned according to the ranking induced by $\boldsymbol{\Delta}$, which depends only on $\mathbf G$ and $\beta$.

In contrast, in the asymmetric model, the optimal treatment assignment depends not only on the network structure and peer effects, but also on the entire productivity distribution. For instance, if conformity to low-performing peers is stronger (high $\beta^l$ and low $\beta^h$), targeting an individual who is generally a high-performing friend may be inefficient, even though this individual is central in the network under the symmetric model. In such cases, individuals’ relative positions in the outcome distribution play a key role in determining optimal assignment.

In Figure~\ref{fig:reflection}, we illustrate different targeting schemes using a 2-star network, where $i$ is the central node. Before the intervention (panel a), the aggregated outcome is $3.66$. Since $j$ and $k$ have no friends but are both connected to $i$, increasing the productivity of $j$ or $k$ yields the same aggregate increase in the symmetric model, regardless of $\beta$. Thus, a social planner who employs the symmetric model can assign the treatment to either $j$ or $k$, expecting the same shift in the outcome distribution. However, panels b and c show that targeting $j$ rather than $k$ is not optimal. Specifically, targeting $j$ increases the aggregate outcome by 1.12, i.e., a spillover effect of 12\%, whereas targeting $k$ almost doubles the spillover effect to 22\%.\footnote{The spillover effect is measured as the additional increase in the aggregate outcome relative to the increase in the aggregate outcome in a benchmark model without social interactions, expressed as a percentage. In this benchmark model, a one-unit increase in $\alpha_i$ raises the aggregate outcome by exactly one unit.} Panel d further shows that targeting $i$ is inefficient, as there are no spillover effects on $j$~and~$k$.

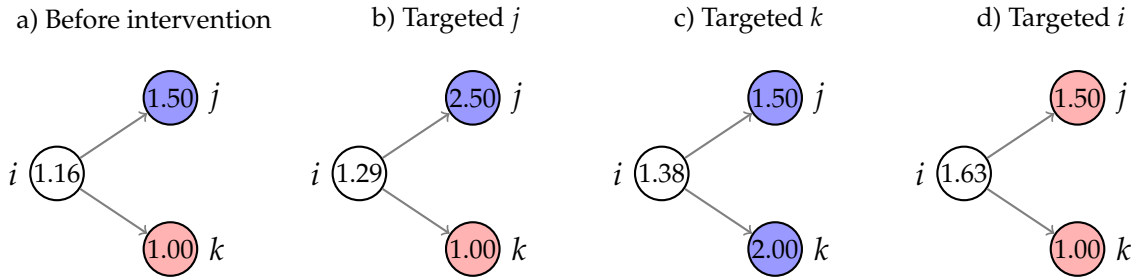
\begin{figure}[!htbp]
    \centering
    \begin{tikzpicture}[scale=0.5]
            \tikzstyle{node}=[circle,draw,fill=white,thick, inner sep=0.5pt]
            \tikzstyle{nodeh}=[circle,draw,fill=blue!40, thick, inner sep=0.5pt]
            \tikzstyle{nodel}=[circle,draw,fill=red!30, thick, inner sep=0.5pt]
            \tikzstyle{link}=[->, color=black!50, thick]
            \tikzstyle{targeted}=[circle,thick]

            \node[above, align=center] at (2.3, 3.4) {\footnotesize a) Before intervention};
            \node[node, label=180:$i$] (i1) at (0,0) {\footnotesize 1.16};
            \node[nodeh, label=0:$j$] (j1) at (3,2) {\footnotesize 1.50};
            \node[nodel, label=0:$k$] (k1) at (3,-2) {\footnotesize 1.00};

            \node[above, align=center, xshift=4cm] at (2.3, 3.4) {\footnotesize b) Targeted $j$};
            \node[node, label=180:$i$, xshift=4cm] (i2) at (0,0) {\footnotesize 1.29};
            \node[nodeh, targeted, label=0:$j$, xshift=4cm] (j2) at (3,2) {\footnotesize 2.50};
            \node[nodel, label=0:$k$, xshift=4cm] (k2) at (3,-2) {\footnotesize 1.00};

            \node[above, align=center, xshift=8cm] at (2.3, 3.4) {\footnotesize c) Targeted $k$};
            \node[node, label=180:$i$, xshift=8cm] (i3) at (0,0) {\footnotesize 1.38};
            \node[nodeh, label=0:$j$, xshift=8cm] (j3) at (3,2) {\footnotesize 1.50};
            \node[nodeh, targeted, label=0:$k$, xshift=8cm] (k3) at (3,-2) {\footnotesize 2.00};

            \node[above, align=center, xshift=12cm] at (2.3, 3.4) {\footnotesize d) Targeted $i$};
            \node[node, targeted, label=180:$i$, xshift=12cm] (i4) at (0,0) {\footnotesize 1.63};
            \node[nodel, label=0:$j$, xshift=12cm] (j4) at (3,2) {\footnotesize 1.50};
            \node[nodel, label=0:$k$, xshift=12cm] (k4) at (3,-2) {\footnotesize 1.00};
            
            \draw[link] (i1)--(j1);
            \draw[link] (i1)--(k1);

            \draw[link] (i2)--(j2);
            \draw[link] (i2)--(k2);

            \draw[link] (i3)--(j3);
            \draw[link] (i3)--(k3);

            \draw[link] (i4)--(j4);
            \draw[link] (i4)--(k4);

        \end{tikzpicture}
    \caption{Targeted policy}
    \label{fig:reflection}
    \justifying
    \footnotesize{\noindent Note: This figure illustrates a targeted intervention in a directed network of three nodes. The values indicated inside each node represent the individuals' performance (outcomes). Red nodes denote low-performing friends, and blue nodes denote high-performing friends. $\beta^l = 1.5$, $\beta^h = 0.5$, and before the intervention (panel a), $(\alpha_i, \alpha_j, \alpha_k) = (1.2, 1.5, 1)$. The intervention consists of increasing $\alpha$ of the targeted individual by one unit.}
\end{figure}

This result highlights a fundamental difference between the symmetric and asymmetric peer effect models. In the asymmetric model, the marginal social return to treatment depends on individuals' relative positions in the outcome distribution. The symmetric model ignores these positions and may lead to inefficient allocation. 

However, unlike the symmetric model, the asymmetric peer effect model does not admit a closed-form characterization of the optimal treatment-allocation ranking. Determining the optimal allocation requires an exhaustive search over all possible treatment sets, the number of which grows rapidly with both network size and the number of individuals to be treated. To make treatment assignment computationally feasible when peer effects are asymmetric, we develop algorithms in Section~\ref{application:policy} that approximate the optimal allocation. We use these algorithms in the empirical application to assign treatment in large networks.

\section{Econometric Framework} \label{sec:metrics}

In this section, we present the econometric model and study its identification and estimation. Since interactions induce dependence across individuals and may complicate identification and inference, we assume that the population is partitioned into $M$ nonoverlapping and independent networks. This assumption is commonly imposed and appropriate for many data sets, as samples typically consist of independent schools, villages, or markets \citep{Bramoulle2009, blume_linear_2015, boucher2023}. We denote by $n_m$ the number of individuals in network $m$. We assume that $n_m$ is bounded, while $M$ goes to infinity asymptotically. 

Throughout the paper, the subscript $m$ refers to variables defined for network $m$; for example, $\mathbf{G}_m$ denotes the interaction matrix in network $m$. A double subscript $m,i$ (e.g., $y_{m,i}$) denotes a variable associated with individual $i$ in network $m$, while a triple subscript $m,ij$ (e.g., $g_{m,ij}$) denotes a pairwise variable.

\subsection{Econometric Model}
Individual $i$ is characterized by a vector of observable characteristics $\mathbf{x}_{m,i} \in \mathbb{R}^{d_x}$, which includes control variables such as age, sex, etc. As in \cite{blume_linear_2015}, we model productivity $\alpha_{m,i}$ as:
\begin{equation}\label{eq:alpha}
    \alpha_{m,i} = c + \mathbf{x}_{m,i}^{\prime}\boldsymbol{\gamma}_1 + \bar{\mathbf{x}}_{m,i}^{\prime}\boldsymbol{\gamma}_2 + \varepsilon_{m,i},
\end{equation}
where $\bar{\mathbf{x}}_{m,i} = \sum_{j \ne i} g_{m,ij}\mathbf{x}_{m,j}$ is a vector of contextual variables defined as the average of exogenous characteristics among friends, and $\varepsilon_{m,i}$ is an unobserved preference shock. The parameters $\boldsymbol{\gamma}_1$ and $\boldsymbol{\gamma}_2$ capture the effects of individual characteristics and contextual variables, respectively, while $c$ is an intercept term. For ease of exposition, we exclude network fixed effects from Equation~\eqref{eq:alpha}. As we show below, although the model allows for a nonlinear relationship between individual and peer outcomes, we derive from Equation~\eqref{eq:yi2} a reduced-form specification that is linear in parameters. Therefore, any network fixed effects in Equation~\eqref{eq:alpha} can be partialled out by demeaning the reduced form within each network.

From Equation~\eqref{eq:yi2}, the reduced-form equation for non-isolated individuals can be written as:
\begin{equation}\label{eq:yred}
    y_{m,i} = \tilde c + \theta_1 \bar y_{m,i} + \theta_2 \check{y}_{m,i} + \mathbf{x}_{m,i}^{\prime}\boldsymbol{\theta}_3 + \bar{\mathbf{x}}_{m,i}^{\prime}\boldsymbol{\theta}_4 + \tilde \varepsilon_{m,i},
\end{equation}
where $\tilde c = \frac{c}{1 + \beta^l}$, $\theta_1 = \frac{\beta^l}{1 + \beta^l}$, $\theta_2 = \frac{\beta^h - \beta^l}{1 + \beta^l}$, $\boldsymbol{\theta}_3 = \frac{\boldsymbol{\gamma}_1}{1 + \beta^l}$, $\boldsymbol{\theta}_4 = \frac{\boldsymbol{\gamma}_2}{1 + \beta^l}$, and $\tilde \varepsilon_{m,i} = \frac{\varepsilon_{m,i}}{1 + \beta^l}$. For isolated individuals, the reduced form simplifies to:
\begin{equation}\label{eq:yrediso}
y_{m,i} = c + \mathbf{x}_{m,i}^{\prime}\boldsymbol{\gamma}_1 + \varepsilon_{m,i}.
\end{equation}

Equations \eqref{eq:yred} and \eqref{eq:yrediso} can be combined into a single linear equation by introducing a dummy variable that indicates whether individual $i$ is isolated. However, for ease of exposition, we assume that there are no isolated individuals and focus on Equation~\eqref{eq:yred} in the remainder of this section.

Equation~\eqref{eq:yred} includes two endogenous regressors: the conventional average outcome among peers, $\bar{y}_{m,i}$, and a new endogenous variable, $\check{y}_{m,i}$. This makes the endogeneity problem more complex than in the standard model, which includes only $\bar{y}_{m,i}$. Additionally, since $\check{y}_{m,i}$ is directly linked to the dependent variable $y_{m,i}$, failing to address this endogeneity can lead to substantial bias.

\subsection{Identification} 
Identifying peer effects is challenging because of the reflection problem, which arises when individuals and their peers influence one another simultaneously \citep{Manski1993}. In the standard linear model, \citet{Bramoulle2009} show that flexible restrictions on the network structure can be imposed to alleviate this problem. Unfortunately, such restrictions are generally unavailable when individual and peer behaviors are related nonlinearly. Most studies addressing this nonlinearity rely on standard rank conditions for identification, typically imposed within a generalized method of moments (GMM) framework \citep[see, e.g.,][]{brock2007identification,lee_binary_2014,boucher2023}. However, these conditions are difficult to verify empirically because they depend on unobserved variables.

In this paper, we establish identification under relatively mild conditions that differ from the standard GMM-based rank conditions. In addition to the restriction on the network structure proposed by \citet{Bramoulle2009}, our identification conditions require a novel exclusion restriction. We show that this restriction is likely to be satisfied because $\check{y}_{m,i}$ is nonlinearly related to peer outcomes.\footnote{Recall that $\check{y}_{m,i} = \sum_{j \ne i} g_{m,ij}\, \mathbbm{1}\{y_{m,j} > y_{m,i}\} (y_{m,j} - y_{m,i})$.} 

For notational ease, we denote by $\mathbb E_m$ the conditional expectation given $\mathbf X_m$ and $\mathbf G_m$. We also define $$\mathbf A_m = \left[\mathbb E_m(\check{\mathbf y}_m), \, \mathbf G_m \mathbb E_m(\check{\mathbf y}_m),\, \mathbf 1_{n_m},\, \mathbf X_m,\, \mathbf G_m \mathbf X_m, \,\mathbf G_m^2 \mathbf X_m\right],$$
where $\check{\mathbf y}_m = (\check y_{m,1}, \, \dots, \, \check y_{m,n_m})^{\prime}$, $\mathbf X_m = (\mathbf x_{m,1}, \, \dots, \, \mathbf x_{m,n_m})^{\prime}$.
\begin{assumption}\label{assumption:data}~\hfill~
\begin{enumerate}[label=(\roman*), ref=\theassumption(\roman*), align=left, leftmargin=7pt, itemsep=0pt, topsep=0pt]
    \item The sequence $(n_m, \mathbf y_m, \mathbf X_m, \mathbf G_m, \boldsymbol \varepsilon_m)_{m\ge1}$ is i.i.d. and satisfies equation \eqref{eq:yred} for any $m$. \label{assumption:data:process}
    \item For all $m$ and $i$, $\mathbb{E}[\varepsilon_{m,i} \mid \mathbf X_m, \mathbf G_m] = 0$. \label{assumption:data:exogeneity}
\end{enumerate}
\end{assumption}
\begin{assumption}\label{assumption:ident}~\hfill~
\begin{enumerate}[label=(\roman*), ref=\theassumption(\roman*), align=left, leftmargin=7pt, itemsep=0pt, topsep=0pt]
    \item The reduced form parameters satisfy $\theta_1\boldsymbol{\theta}_3 + \boldsymbol{\theta}_4 \ne \mathbf 0$. \label{assumption:ident:nonzero}
    \item The matrix $\frac{1}{M}\sum_{m = 1}^M \mathbf A_m^{\prime}\mathbf A_m$ is of full rank. \label{assumption:ident:fullrank}
\end{enumerate}
\end{assumption}

Assumption~\ref{assumption:data} imposes standard conditions on the data-generating process. Assumption~\ref{assumption:data:process} requires the data sequence to be i.i.d. across $m$, while allowing for flexible dependence structures within networks. We include $n_m$ in the sequence because network size may vary across observations. Specifically, the assumption implies that $(\mathbf y_m, \mathbf X_m, \mathbf G_m, \boldsymbol \varepsilon_m)$ is i.i.d. conditional on $n_m$, and that $n_m$ itself is i.i.d. This assumption allows us to establish the equivalence between empirical averages across networks and their population counterparts using standard laws of large numbers. Assumption~\ref{assumption:data:exogeneity} requires $\mathbf X_m$ and $\mathbf G_m$ to be exogenous. We maintain this assumption to abstract from endogenous network formation.

Assumption~\ref{assumption:ident} introduces our identification restrictions. The restriction in Assumption~\ref{assumption:ident:nonzero} holds in many settings when $\mathbf{x}_{m,i}$ includes multiple variables. This restriction is also required for the standard linear model \cite[see][]{Bramoulle2009}. Assumption~\ref{assumption:ident:fullrank} is our new rank condition. Since $\mathbf{A}_m$ does not include the standard average peer outcome, but instead $\check{\mathbf{y}}_m$, which is nonlinearly related to the outcome, this restriction is easier to motivate than classical GMM rank conditions.

Specifically, Assumption~\ref{assumption:ident:fullrank} can be decomposed into two parts. First, it requires $\mathbf 1_{n_m}$, $\mathbf X_m$, $\mathbf G_m \mathbf X_m$, and $\mathbf G_m^2 \mathbf X_m$ to be linearly independent, which can be verified from the data because these variables are all observed. This requirement is related to the identification of the standard model. It holds when many individuals have friends of friends who are not direct friends \citep[see][]{Bramoulle2009}. Second, Assumption~\ref{assumption:ident:fullrank} also requires that no nonzero linear combination of $\mathbb E_m(\check{\mathbf y}_m)$ and $\mathbf G_m \mathbb E_m(\check{\mathbf y}_m)$ lies in the linear span of $\mathbf 1_{n_m}$, $\mathbf X_m$, $\mathbf G_m \mathbf X_m$, and $\mathbf G_m^2 \mathbf X_m$. While it is difficult to provide low-level conditions ensuring this restriction, the presence of the indicator function in $\check{\mathbf y}_m$ makes it unlikely that $\mathbb E_m(\check{\mathbf y}_m)$ and $\mathbf G_m \mathbb E_m(\check{\mathbf y}_m)$ lie in this linear span.\footnote{Additionally, $\mathbb E_m(\check{\mathbf y}_m)$ and $\mathbf G_m \mathbb E_m(\check{\mathbf y}_m)$ must not be collinear. This condition also generally holds when some friends of friends are not direct friends, because $\mathbf G_m \mathbb E_m(\check{\mathbf y}_m)$ involves friends of friends, whereas $\mathbb E_m(\check{\mathbf y}_m)$ depends only on direct friends.}

We establish the following result.
\begin{proposition}\label{prop:ident}
    Under Assumptions \ref{assumption:minbeta}--\ref{assumption:ident}, $\beta^l$, $\beta^h$, $c$, $\boldsymbol\gamma_1$, and $\boldsymbol\gamma_2$ are point identified. 
\end{proposition}

The proof of Proposition \ref{prop:ident} is provided in Appendix \ref{append:proof:Ident}. 

Let $\mathbf{V}_{m} = [\mathbf G_m \mathbf y_m, \, \check{\mathbf y}_{m}, \, \mathbf 1_{n_m}, \,\mathbf{X}_{m}, \,\mathbf G_m \mathbf{X}_{m}]$ denote the matrix of regressors in~\eqref{eq:yred}, and let $\mathbf{V}_{m}^e = \mathbb{E}_m(\mathbf{V}_{m})$ denote its conditional expectation given $\mathbf{X}_m$ and $\mathbf{G}_m$. As shown by \cite{chamberlain1987}, $\mathbf{V}_{m}^{e}$ corresponds to the optimal instrument for $\mathbf{V}_m$. However, $\mathbf{V}_{m}^{e}$ is infeasible because $\mathbb{E}_m(\mathbf y_{m})$ and $\mathbb{E}_m(\check{\mathbf y}_{m})$ are unobserved. In the next section, we discuss how to construct valid instruments for $\mathbf G_m \mathbf y_m $ and $\check{\mathbf y}_{m}$.

\subsection{Estimation} \label{sec:estim}
We propose an instrumental variable (IV) approach to estimate the model parameters. Since $\mathbb{E}_m(\bar{y}_{m,i})$ and $\mathbb{E}_m(\check{y}_{m,i})$ are unobserved and cannot be used as instruments, we instead construct their predictors that serve as instruments. Importantly, these predictors do not need to be consistent estimators. Any exogenous variables that are sufficiently informative about $\mathbb{E}_m(\bar{y}_{m,i})$ and $\mathbb{E}_m(\check{y}_{m,i})$ can serve as valid instruments. Throughout, we use the notation $\hat{\mathbb{E}}_m$ to denote a predicted counterpart of $\mathbb{E}_m$.

From Equation~\eqref{eq:yred}, for any network $m$, $\mathbb{E}_m(\mathbf{y}_m)$ can be written as:\footnote{See the derivation of $\mathbb{E}_m(\mathbf{y}_m)$ in Equation~\eqref{eq:AppendEys} in Appendix~\ref{append:proof:Ident}. Since $\lvert \theta_1 \rvert < 1$ (Assumption~\ref{assumption:minbeta}), we can express $(\mathbf{I}_{n_m} - \theta_1 \mathbf{G}_m)^{-1}$ using the Neumann series expansion $\sum_{k = 0}^{\infty} \theta_1^k \mathbf{G}_m^k$. Moreover, because there are no isolated individuals, $\mathbf{G}_m \mathbf{1}_s = \mathbf{1}_s$, which implies $\sum_{k = 0}^{\infty} \theta_1^k \mathbf{G}_m^k \mathbf{1}_s = \frac{1}{1 - \theta_1} \mathbf{1}_s$.}
\begin{equation}\label{eq:Eys}
  \mathbb E_m (\mathbf y_m) = \dfrac{\tilde c}{1 - \theta_1} \mathbf 1_{n_m}
  + \mathbf X_m\boldsymbol\theta_3
  + \sum_{k = 0}^{\infty} \theta_1^k \theta_2 \mathbf G_m^k \mathbb E_m (\check{\mathbf y}_m)
  + \sum_{k = 0}^{\infty} \theta_1^k \mathbf G_m^{k + 1} \mathbf X_m(\theta_1 \boldsymbol\theta_3 + \boldsymbol\theta_4).
\end{equation}
This representation suggests that $\check{\mathbf W}_m = [\mathbf X_m,~ \mathbf G_m \mathbf X_m,~ \mathbf G_m^2 \mathbf X_m]$ provides a natural set of predetermined predictors of $\mathbf E_m(\mathbf y_m)$.\footnote{Higher powers of $\mathbf G_m$ can be included in $\check{\mathbf W}_m$ for greater accuracy.} Let $\check{\mathbf w}_{m,i}^{\prime}$ denote the $i$-th row of $\check{\mathbf W}_m$.

We begin by predicting $\mathbb{E}_m(\check{y}_{m,i})$. Let $\Delta^y_{m,ji} = y_{m,j} - y_{m,i}$ and let $\mathbb{P}_m$ denote the probability measure conditional on $\mathbf{X}_m$ and $\mathbf{G}_m$. We rely on the following decomposition:
\begin{equation}\label{eq:Eycheck}
    \mathbb{E}_m(\check{y}_{m,i}) 
    = \sum_{j \ne i} g_{m,ij}\, 
    \mathbb{P}_m\{\Delta^y_{m,ji} > 0\} \, 
    \mathbb{E}_m(\Delta^y_{m,ji} \mid \Delta^y_{m,ji} > 0),
\end{equation}
where $\mathbb{P}_m\{\Delta^y_{m,ji} > 0\}$ can be interpreted as an extensive margin and $\mathbb{E}_m(\Delta^y_{m,ji} \mid \Delta^y_{m,ji} > 0)$ is the corresponding intensive margin. We adopt this decomposition because the two margins are easier to predict than $\mathbb{E}_m(\check{y}_{m,i})$ directly. Moreover, because both margins are defined for each pair of connected students, we can exploit a large number of dyadic observations to estimate them, which improves prediction accuracy. To construct $\hat{\mathbb{E}}_m(\check{y}_{m,i})$, we replace each margin with its predicted counterpart. We use flexible supervised machine learning methods to predict each margin. Throughout the paper, we employ random forests, although any supervised learning method could be used. For both prediction tasks, $\check{\mathbf w}_{m,j} - \check{\mathbf w}_{m,i}$ serves as the predictor vector, whereas the target variable is $\mathbbm{1}\{\Delta^y_{m,ji} > 0\}$ for the extensive margin and $\Delta^y_{m,ji}$ for the intensive margin.

To ensure that the resulting predictions are exogenous, we use a cross-fitting procedure across networks \citep{Chernozhukov2018Double}. We randomly split the networks into $L \geq 2$ folds, $\mathcal{F}_1, \dots, \mathcal{F}_L$. For any fold $\mathcal{F}_l$, the random forest model is trained using observations from $\mathcal{F}_{-l}$.\footnote{We use $\mathcal{F}_{-l}$ to denote the subsample of all observations excluding those in $\mathcal{F}_l$.} Once the model is trained, only predetermined predictors from $\mathcal{F}_l$ are used to generate predictions for $\mathcal{F}_l$. Thus, no endogenous variables from $\mathcal{F}_l$ enter either the training or prediction steps for that fold. Since the networks are independent and randomly split, this procedure ensures that the predictions for any network $m$ are independent of $\boldsymbol{\varepsilon}_m$.\footnote{Furthermore, for the intensive margin, we train the random forest on the subsample satisfying $y_{m',j} > y_{m',i}$, where $m' \in \mathcal{F}_{-l}$. This allows the prediction to be interpreted as a conditional expectation given $y_{m,j} > y_{m,i}$, consistent with the definition of the intensive margin.} For completeness, we provide a detailed description of the prediction procedure in Supplemental Appendix~\ref{append:instrument}.

For the prediction of $\mathbb{E}_m(\bar{y}_{m,i})$, we define
$$
\bar{\mathbf{W}}_m = [\mathbf{G}_m \mathbf{X}_m,\, \mathbf{G}_m^2 \mathbf{X}_m,\, \mathbf{G}_m \hat{\mathbb{E}}_m(\check{\mathbf{y}}_m),\, \mathbf{G}_m^2 \hat{\mathbb{E}}_m(\check{\mathbf{y}}_m)].
$$
We also denote by $\bar{\mathbf{w}}_{m,i}^{\prime}$ the $i$-th row of $\bar{\mathbf{W}}_m$. Premultiplying Equation~\eqref{eq:Eys} by $\mathbf{G}_m$ implies that $\bar{\mathbf{w}}_{m,i}^{\prime}$ is a natural set of predetermined predictors for $\mathbb{E}_m(\bar{y}_{m,i})$. Notably, $\bar{\mathbf{W}}_m$ includes the standard instruments $\mathbf{G}_m \mathbf{X}_m$ and $\mathbf{G}_m^2 \mathbf{X}_m$ used in the linear model \citep[see][]{Bramoulle2009}, as well as additional predictors to accommodate asymmetries in peer effects. We predict $\mathbb{E}_m(\bar{y}_{m,i})$ using random forests, where $\bar{\mathbf{w}}_{m,i}$ serves as the predictor and the target variable is $\bar{y}_{m,i}$. We also use cross-fitting across networks so that the resulting predictions are independent of $\boldsymbol{\varepsilon}_m$.

Equipped with the predictions $\hat{\mathbb{E}}_m(\check{y}_{m,i})$ and $\hat{\mathbb{E}}_m(\bar{y}_{m,i})$, we define the matrix of instruments as
$\hat{\mathbf{Z}}_m = \big[ \hat{\mathbb{E}}_m(\check{\mathbf{y}}_{m}), \, \hat{\mathbb{E}}_m(\bar{\mathbf{y}}_{m}), \, \mathbf 1_{n_m}, \, \mathbf{X}_m, \, \mathbf{G}_m \mathbf{X}_m \big]$.
The model parameters are then estimated by solving the moment condition
$$
\mathbb E \left( \hat{\mathbf{Z}}_m^{\prime} \boldsymbol{\varepsilon}_m \right)= \mathbf{0},
$$
where $\boldsymbol{\varepsilon}_m = (\varepsilon_{m,1}, \, \dots, \, \varepsilon_{m,n_m})^{\prime}$.

However, since $\hat{\mathbb{E}}_m(\check{\mathbf{y}}_{m})$ and $\hat{\mathbb{E}}_m(\bar{\mathbf{y}}_{m})$ are generated instruments, they act as nuisance parameters, and their uncertainty must be taken into account in inference. We show that the moment function is orthogonal to these nuisance parameters, and therefore their influence can be ignored. This is a particularly strong form of orthogonality, as it holds at any order \citep[see][]{mackey2018orthogonal, bonhomme2026higher}, thereby allowing the use of highly flexible machine learning methods without requiring convergence rates. As a result, our inference relies on weak regularity conditions stated in Assumption~\ref{assumption:iv} below.

Let $\check{f}_m$ and $\bar{f}_m$ be two nonstochastic functions mapping $(\mathbf X_m, \mathbf G_m)$ into $\mathbb R^{n_m}$. Let $$\mathbf Z_m = [\check f_m(\mathbf X_m, \mathbf G_m),\, \bar f_m(\mathbf X_m, \mathbf G_m),\, \mathbf 1_{n_m},\, \mathbf X_m,\, \mathbf G_m \mathbf X_m]$$ denote the asymptotic equivalent of $\hat{\mathbf Z}_m$.
\begin{assumption}\label{assumption:iv}~\hfill~
\begin{enumerate}[label=(\roman*), ref=\theassumption(\roman*), align=left, leftmargin=7pt, itemsep=0pt, topsep=0pt]
    \item $\displaystyle \max_{m} \lVert \hat{\mathbb E}_m(\check{\mathbf y}_{m}) - \check f_m(\mathbf X_m, \mathbf G_m)\rVert = o_p(1)$ and $\displaystyle \max_{m} \lVert \hat{\mathbb E}_m(\bar{\mathbf y}_{m}) - \bar f_m(\mathbf X_m, \mathbf G_m)\rVert = o_p(1)$. \label{assumption:iv:ml}
    \item The matrix $\frac{1}{M} \sum_{m = 1}^M \mathbf Z_m^{\prime}\mathbf V_m$ is asymptotically nonsingular and nonstochastic. \label{assumption:iv:fullrank} 
    \item For some $\nu > 0$, $\displaystyle \max_m \lVert \mathbb E_m(\boldsymbol \varepsilon_m \boldsymbol \varepsilon_m^{\prime}) \rVert^{2 + \nu} < \infty$. \label{assumption:iv:variance} 
\end{enumerate}
\end{assumption}

Assumption~\ref{assumption:iv:ml} requires that $\hat{\mathbb{E}}_m(\check{\mathbf{y}}_{m})$ and $\hat{\mathbb{E}}_m(\bar{\mathbf{y}}_{m})$ converge to nonstochastic quantities, $\check{f}_m(\mathbf{X}_m, \mathbf{G}_m)$ and $\bar{f}_m(\mathbf{X}_m, \mathbf{G}_m)$, respectively, conditional on $\mathbf{X}_m$ and $\mathbf{G}_m$. Notably, no rate of convergence is required, and $\check{f}_m(\mathbf{X}_m, \mathbf{G}_m)$ and $\bar{f}_m(\mathbf{X}_m, \mathbf{G}_m)$ need \textit{not} coincide with the estimands $\mathbb{E}_m(\bar{y}_{m,i})$ or $\mathbb{E}_m(\check{y}_{m,i})$, so the random forest predictions may be asymptotically biased. Nevertheless, the predictions must be sufficiently correlated with the endogenous variables $\bar{y}_{m,i}$ and $\check{y}_{m,i}$ to avoid weak instrument problems. This requirement is reflected in the nonsingularity condition in Assumption~\ref{assumption:iv:fullrank}. In practice, we assess instrument strength using standard weak instrument tests. We also examine whether $\frac{1}{M} \sum_{m=1}^M \hat{\mathbf{Z}}_m^{\prime} \mathbf{V}_m$ is nonsingular using rank tests \citep[e.g.,][]{kleibergen200}. Finally, Assumption~\ref{assumption:iv:variance} imposes a standard finite-moment condition required for the central limit theorem.

Let $\boldsymbol\theta = (\theta_1, \, \theta_2, \,  \tilde c, \, \boldsymbol{\theta}_3^{\prime}, \, \boldsymbol{\theta}_4^{\prime})^{\prime}$ be the vector of reduced-form parameters and $\boldsymbol{\hat \theta}$ be its estimator. Let also $\boldsymbol{\theta}_0$ denote the true value of $\boldsymbol{\theta}$. We establish the following result.
\begin{proposition}\label{prop:consistency}
     Under Assumptions~\ref{assumption:minbeta}--\ref{assumption:iv}, $\boldsymbol{\hat \theta}$ converges in probability to $\boldsymbol{\theta}_0$, and $\sqrt{M}(\boldsymbol{\hat \theta} - \boldsymbol{\theta}_0) \overset{d}{\to} N(\mathbf 0, \, \boldsymbol\Sigma)$, where $\boldsymbol\Sigma$ is given in Appendix~\ref{append:prop:consistency}.
\end{proposition}
\noindent The proof of Proposition~\ref{prop:consistency} is provided in Appendix~\ref{append:prop:consistency}. This result also extends to the structural parameters. The estimators are consistent and asymptotically normally distributed. The asymptotic variance can be obtained using the Delta method.

\subsection{Bias of the Standard Symmetric Model}\label{sec:biasSymmetric}
Given that the standard symmetric model is widely used in empirical studies, we investigate what its single peer effect parameter, $\beta$, measures when the data exhibits asymmetric peer effects. One might expect the estimate of $\beta$ to be a simple weighted average of the estimates of $\beta^l$ and $\beta^h$, which would imply that $\beta$ lies asymptotically between $\beta^l$ and $\beta^h$. However, we show that the weights in this average can be negative, implying that $\beta$ may lie outside the range defined by $\beta^l$ and $\beta^h$. This issue is common in settings where a single parameter aggregates heterogeneous effects. A well-known example is the heterogeneous treatment effects problem studied by \cite{de2020two}.

In the symmetric model \citep{blume_linear_2015}, the outcome is modeled as:
\begin{equation}\label{eq:LIM}
  \mathbf y_{m} = \delta_1 \bar{\mathbf y}_{m}
  + \mathbf{R}_{m} \boldsymbol{\delta}_2
  + \boldsymbol\eta_{m},
\end{equation}
where $\delta_1 = \frac{\beta}{1+\beta}$, $\mathbf R_m = (\mathbf 1_{n_m}, \, \mathbf X_m, \, \mathbf G_m \mathbf X_m)$ denotes the matrix of control variables, and $\boldsymbol\eta_{m}$ is an error term. The parameters $\beta$ and $\boldsymbol{\delta}_2$ can be estimated by GMM, where $\bar{\mathbf y}_{m}$ is instrumented by $\hat{\mathbb E}_m(\bar{\mathbf y}_{m})$, or by any other valid instrument, such as $\mathbf G_m^2 \mathbf X_m$ \citep[see][]{Bramoulle2009}. The resulting moment condition is
$$
\mathbb E\big(\tilde{\mathbf{Z}}_m^{\prime} (\mathbf y_{m} - \delta_1 \bar{\mathbf y}_{m} - \mathbf{R}_{m} \boldsymbol{\delta}_2)\big) = \mathbf{0},
$$
where $\tilde{\mathbf{Z}}_m = \big[\hat{\mathbb{E}}_m(\bar{\mathbf{y}}_{m}), \, \mathbf R_m \big]$ is the matrix of instruments.

For the rest of this section, we add a superscript $h$ to $\check y_{m,i}$; that is, we write $\check y_{m,i}^h = \bar y_{m,i}^h - g_{m,i}^h y_{m,i}$. We also define the analogous variable $\check y_{m,i}^l = \bar y_{m,i}^l - g_{m,i}^l y_{m,i}$, where $\bar{y}_{m,i}^l = \sum_{j \ne i} \mathbbm{1}\{y_j \leq y_i\} g_{ij} y_j$ and $g_{m,i}^l = \sum_{j \ne i} \mathbbm{1}\{y_j \leq y_i\} g_{ij}$. Let $\check{\mathbf y}_m^h = (\check y_{m,1}^h, \, \dots, \, \check y_{m,n_m}^h)^{\prime}$ and $\check{\mathbf y}_m^l = (\check y_{m,1}^l, \, \dots, \, \check y_{m,n_m}^l)^{\prime}$. 

We consider the following auxiliary moment conditions:
\allowdisplaybreaks
\begin{align}\label{eq:auxiliaryIV}
\begin{split}
    \mathbb E \left( \tilde{\mathbf{Z}}_m^{\prime} \big( \check{\mathbf y}_m^h - \delta^h_1 \bar{\mathbf y}_m - \mathbf R_m\boldsymbol{\delta}^h_2\big) \right) &= \mathbf{0},\\
    \mathbb E \left( \tilde{\mathbf{Z}}_m^{\prime} \big( \check{\mathbf y}_m^l - \delta^l_1 \bar{\mathbf y}_m - \mathbf R_m\boldsymbol{\delta}^l_2 \big)\right) &= \mathbf{0},
\end{split}
\end{align}
for some unknown parameters $\delta^h_1$, $\boldsymbol{\delta}^h_2$, $\delta^l_1$, and $\boldsymbol{\delta}^l_2$. Specifically, $\delta^h_1$ is the coefficient from the IV regression of $\check{\mathbf y}_m^h$ on $\bar{\mathbf y}_m$, partialling out $\mathbf R_m$. Similarly, $\delta^l_1$ is the coefficient from the IV regression of $\check{\mathbf y}_m^l$ on $\bar{\mathbf y}_m$, partialling out $\mathbf R_m$.

\begin{proposition}\label{prop:biasSymmetric}
Under Assumptions~\ref{assumption:minbeta}--\ref{assumption:iv}, we have:
$$
\frac{\beta}{1+\beta} = \frac{\beta^l}{1+\beta^l} + \delta^h_1 \dfrac{\beta^h - \beta^l}{1+\beta^l}
\quad \text{and} \quad
\frac{\beta}{1+\beta} = \frac{\beta^h}{1+\beta^h} + \delta^l_1 \dfrac{\beta^l - \beta^h}{1+\beta^h}.
$$
\end{proposition}
The proof of Proposition~\ref{prop:biasSymmetric} is provided in Appendix~\ref{append:prop:biasSymmetric}.

A direct implication of Proposition~\ref{prop:biasSymmetric} is that the necessary and sufficient conditions for $\beta$ to lie in the range defined by $\beta^l$ and $\beta^h$ are that $\delta^l_1 \geq 0$ and $\delta^h_1 \geq 0$.\footnote{Since $f(\beta) = \frac{\beta}{1+\beta}$ is strictly increasing in $\beta$, saying that $\beta$ lies in the range defined by $\beta^l$ and $\beta^h$ is equivalent to saying that $f(\beta)$ lies in the range defined by $f(\beta^l)$ and $f(\beta^h)$ } For example, if $\beta^l < \beta^h$, then $\frac{\beta^l}{1+\beta^l} \leq \frac{\beta}{1+\beta} \leq \frac{\beta^h}{1+\beta^h}$ if and only if $\delta^l_1 \geq 0$ and $\delta^h_1 \geq 0$. A similar argument applies when $\beta^l > \beta^h$. 

If $\delta^l_1$ or $\delta^h_1$ is negative, then $\beta$ falls outside the range defined by $\beta^l$ and $\beta^h$, meaning that $\beta$ is a weighted average of $\beta^l$ and $\beta^h$ with negative weights.\footnote{Proposition~\ref{prop:biasSymmetric} implies that $\dfrac{\beta}{1 + \beta} = \rho \dfrac{\beta^l}{1 + \beta^l} + (1 - \rho)\dfrac{\beta^h}{1 + \beta^h}$, where $\rho = \dfrac{\delta^l_1(1 + \beta^l)}{\delta^l_1(1 + \beta^l) + \delta^h_1(1 + \beta^h)}$. Since $f(\beta)$ is strictly increasing in $\beta$, this means that $\beta$ can also be written as an average of $\beta^l$ and $\beta^h$.} To see when this can happen, note, for example, that a negative $\delta^l_1$ implies that an increase in $\bar{y}_{m,i}$ has a larger effect on $g_{m,i}^l y_{m,i}$ than on $\bar y_{m,i}^l$. This suggests that this increase is mainly driven by $\bar y_{m,i}^h$, while $\bar y_{m,i}^l$ remains nearly unchanged. This situation arises in certain networks where higher-performing peers are predominant, and there are few or no (direct or indirect) connections from these peers to lower-performing ones. This is, for example, the case in certain segregated networks where friendships are not always reciprocated and often run from lower-status individuals to higher-status individuals \citep[e.g.,][]{ball2013friendship}. We consider such network structures in our simulation study and show that $\beta$ lies outside the range defined by $\beta^l$ and $\beta^h$.

\section{Monte Carlo Simulations}\label{sec:MCsimu}

We conduct a simulation study to assess the finite-sample performance of our estimation method. We consider a sample of $M = 50$ networks, each comprising $n_m = 50$ individuals. Each individual $i$ in network $m$ is characterized by two exogenous variables, $x_{1,m,i}$ and $x_{2,m,i}$. Productivity is defined as in Equation~\eqref{eq:alpha}, except that we introduce an unobserved network-level shock $c_m$:
\begin{equation*}
    \alpha_{m,i} = c_m + \gamma_{1,1} x_{1,m,i} + \gamma_{1,2} x_{2,m,i} + \gamma_{2,1} \bar{x}_{1,m,i} + \gamma_{2,2} \bar{x}_{2,m,i} + \varepsilon_{m,i},
\end{equation*}
where $c_m \sim \mathcal{N}(10, \,1)$ and $\varepsilon_{m,i} \sim \mathcal{N}(0, \,1)$. We simulate $x_{1,m,i}$ to be correlated with $c_m$, so that $c_m$ acts as a network fixed effect.

We consider both fully random and segregated network formation. In the case of random networks, the degree (number of friends) is randomly drawn from the empirical distribution observed in the Add Health data used in the empirical application. Individuals can have up to 10 friends, with an average of 3.47. Given the degree, links are formed uniformly at random within the network. We simulate $x_{1,m,i}$ and $x_{2,m,i}$ from $\text{Uniform}(c_m - 2, \,c_m + 2)$ and $\text{Poisson}(2)$, respectively, and set $\gamma_{1,1} = 1.4$, $\gamma_{1,2} = -0.8$, $\gamma_{2,1} = 0.7$, and $\gamma_{2,2} = -0.5$.

Under the segregated network, to easily handle peers’ relative status, we employ a model without contextual effects for both the data-generating process (DGP) and the estimation specification. We thus impose that $\gamma_{2,1} = 0$ and $\gamma_{2,2} = 0$, and keep the other parameters at their values in the random network case. To introduce segregation, we also change the way we simulate $x_{1,m,i}$. We randomly split the 50 individuals in each network $m$ into three groups, $\mathcal{G}_{m,1}$, $\mathcal{G}_{m,2}$, and $\mathcal{G}_{m,3}$, comprising 10, 10, and 30 individuals, respectively. For individuals in $\mathcal{G}_{m,1}$ and $\mathcal{G}_{m,2}$, we simulate $x_{1,m,i} \sim \text{Uniform}(c_m - 2, \, c_m + 2)$, while for those in $\mathcal{G}_{m,3}$, we simulate $x_{1,m,i} \sim \text{Uniform}(c_m + 10, \, c_m + 20)$.

Individuals in $\mathcal{G}_{m,1}$ form reciprocal links with those in $\mathcal{G}_{m,2}$. Individuals in $\mathcal{G}_{m,2}$ further form non-reciprocal links with those in $\mathcal{G}_{m,3}$, whereas individuals in $\mathcal{G}_{m,3}$ have no friends. For simplicity, the degrees from $\mathcal{G}_{m,1}$ to $\mathcal{G}_{m,2}$ and from $\mathcal{G}_{m,2}$ to $\mathcal{G}_{m,3}$ are independent and randomly drawn from the empirical Add Health distribution. Generating networks in this way ensures that higher-performing peers---who include most friends from $\mathcal{G}_{m,3}$---are predominant and have fewer connections with lower-performing peers. We therefore expect $\delta_1^l$ in Proposition~\ref{prop:biasSymmetric} to be negative for $\beta$ to fall outside the range defined by $\beta^l$ and $\beta^h$.

For each type of network, we consider four DGPs with different values of $\beta^l$ and $\beta^h$. In DGP 1, $\beta^l = \beta^h = 1$, i.e., peer effects are symmetric. In DGP 2, $\beta^l = 0.4$ and $\beta^h = 2.6$, while in DGP 3, $\beta^l = 2.6$ and $\beta^h = 0.4$. In DGP 4, we set $\beta^l = -0.4$ and $\beta^h = 2.6$. We allow for anti-conformity toward low-performing peers in DGP 4 because this pattern arises in our empirical application, although it is not statistically significant. For each DGP, we estimate both the asymmetric and the symmetric models. Table~\ref{tab:MonteCarlo} summarizes the simulation results.

\begin{table}[htbp]
\centering
\setlength{\tabcolsep}{4pt}
\renewcommand{\arraystretch}{1.18}
\begin{threeparttable}
\caption{Simulation Results}
\label{tab:MonteCarlo}
\footnotesize
\begin{tabular}{ld{4}ld{4}ld{4}ld{4}l}
\toprule
                   & \multicolumn{4}{c}{Random Network}                                       & \multicolumn{4}{c}{Segregated Network}                               \\
Model              & \multicolumn{2}{c}{Asymmetry}       & \multicolumn{2}{c}{Symmetry}       & \multicolumn{2}{c}{Asymmetry}     & \multicolumn{2}{c}{Symmetry}     \\
\midrule
\multicolumn{9}{c}{DGP 1: $\beta^l = 1.00$,   $\beta^h = 1.00$}  \\[1ex]
$\beta^l$      & 1.005         & (0.045)       &              &               & 1.009         & (0.091)       &              &               \\
$\beta^h$      & 0.998         & (0.134)       &              &               & 1.001         & (0.019)       &              &               \\
$\beta$        &               &               & 1.003        & (0.035)       &               &               & 0.999        & (0.014)       \\
$\gamma_{1,1}$ & 1.401         & (0.031)       & 1.401        & (0.029)       & 1.400         & (0.005)       & 1.400        & (0.005)       \\
$\gamma_{1,2}$ & -0.801        & (0.023)       & -0.801       & (0.021)       & -0.801        & (0.017)       & -0.800       & (0.017)       \\
$\gamma_{2,1}$ & 0.699         & (0.036)       & 0.699        & (0.035)       &               &               &              &               \\
$\gamma_{2,2}$ & -0.500        & (0.028)       & -0.500       & (0.028)       &               &               &              &               \\
\midrule
\multicolumn{9}{c}{DGP 2: $\beta^l = 0.40$,   $\beta^h = 2.60$}  \\[1ex]
$\beta^l$      & 0.403         & (0.030)       &              &               & 0.411         & (0.086)       &              &               \\
$\beta^h$      & 2.600         & (0.170)       &              &               & 2.603         & (0.045)       &              &               \\
$\beta$        &               &               & 0.757        & (0.036)       &               &               & 3.683        & (0.106)       \\
$\gamma_{1,1}$ & 1.401         & (0.031)       & 1.271        & (0.030)       & 1.400         & (0.005)       & 1.401        & (0.005)       \\
$\gamma_{1,2}$ & -0.801        & (0.021)       & -0.716       & (0.021)       & -0.801        & (0.017)       & -0.811       & (0.018)       \\
$\gamma_{2,1}$ & 0.700         & (0.036)       & 0.648        & (0.041)       &               &               &              &               \\
$\gamma_{2,2}$ & -0.500        & (0.028)       & -0.448       & (0.033)       &               &               &              &               \\
\midrule
\multicolumn{9}{c}{DGP 3: $\beta^l = 2.60$,   $\beta^h = 0.40$}  \\[1ex]
$\beta^l$      & 2.610         & (0.085)       &              &               & 2.609         & (0.123)       &              &               \\
$\beta^h$      & 0.395         & (0.119)       &              &               & 0.400         & (0.009)       &              &               \\
$\beta$        &               &               & 1.817        & (0.063)       &               &               & 0.327        & (0.008)       \\
$\gamma_{1,1}$ & 1.402         & (0.035)       & 1.414        & (0.033)       & 1.400         & (0.005)       & 1.392        & (0.005)       \\
$\gamma_{1,2}$ & -0.800        & (0.027)       & -0.819       & (0.024)       & -0.800        & (0.018)       & -0.739       & (0.016)       \\
$\gamma_{2,1}$ & 0.699         & (0.036)       & 0.816        & (0.046)       &               &               &              &               \\
$\gamma_{2,2}$ & -0.500        & (0.029)       & -0.581       & (0.040)       &               &               &              &               \\
\midrule
\multicolumn{9}{c}{DGP 4: $\beta^l = -0.40$,   $\beta^h = 2.60$} \\[1ex]
$\beta^l$      & -0.398        & (0.022)       &              &               & -0.384        & (0.071)       &              &               \\
$\beta^h$      & 2.590         & (0.167)       &              &               & 2.606         & (0.043)       &              &               \\
$\beta$        &               &               & -0.077       & (0.021)       &               &               & 4.152        & (0.178)       \\
$\gamma_{1,1}$ & 1.398         & (0.034)       & 0.956        & (0.031)       & 1.401         & (0.005)       & 1.402        & (0.005)       \\
$\gamma_{1,2}$ & -0.800        & (0.023)       & -0.532       & (0.022)       & -0.802        & (0.017)       & -0.815       & (0.018)       \\
$\gamma_{2,1}$ & 0.700         & (0.036)       & 0.791        & (0.048)       &               &               &              &               \\
$\gamma_{2,2}$ & -0.500        & (0.028)       & -0.500       & (0.037)       &               &               &              &               \\ \bottomrule
\end{tabular}
\begin{tablenotes}[para,flushleft]\footnotesize
Notes: The models are simulated and estimated 1,000 times. Values reported without parentheses correspond to mean estimates, whereas values in parentheses denote standard deviations. To generate the instruments, we use random forests with 5-fold cross-fitting (see Supplemental Appendix \ref{append:instrument}) and 1,000 trees. For each training step, 20\% of the sample is used to tune the remaining hyperparameters.
\end{tablenotes}
\end{threeparttable}
\end{table}

The results provide strong evidence that, under the asymmetric specification, our estimation method successfully recovers the structural parameters in finite samples for both network formation processes and across all DGPs. For DGP 1, the symmetric specification also performs well and is, unsurprisingly, more precise. Fortunately, once the asymmetric model is estimated, the null hypothesis $\beta^l = \beta^h$ can be easily tested using standard tests. If the null hypothesis is not rejected, the symmetric model can be used to improve efficiency.

For DGPs 2--4, the estimates of the symmetric peer-effect parameter lie between $\beta^l$ and $\beta^h$ when friendships form randomly. Yet, these estimates are overly simplistic and mask important features of the heterogeneous structure of peer effects. For instance, in DGP 4, where there is anticonformity toward lower-performing peers only, the symmetric specification implies anticonformity toward all peers.  For the segregated network, the estimates of $\beta$ always fall outside the range defined by $\beta^l$ and $\beta^h$, as expected. Across all DGPs, either $\beta < \beta^h < \beta^l$ or $\beta > \beta^h > \beta^l$, implying that $\beta$ assigns a negative weight to $\beta^l$ and a weight greater than one to $\beta^h$. This is consistent with $\delta_1^l < 0$ in Proposition~\ref{prop:biasSymmetric}.

Remarkably, the simulation results also reveal that the coefficients of the exogenous variables are biased under the symmetric specification. This bias is particularly pronounced in the randomly formed network, likely due to the inclusion of contextual effects in the specification.\footnote{Following a similar approach as in the proof of Proposition~\ref{prop:biasSymmetric}, we can derive expressions for the bias of $\boldsymbol \gamma_1$ and $\boldsymbol \gamma_2$. However, this bias is less intuitive, as it depends on how $\mathbf{x}_{m,i}$ and $\bar{\mathbf{x}}_{m,i}$ relate to $\check y_{m,i}$.} This result suggests that our asymmetric peer effect model is effective not only in estimating endogenous peer effects but also in estimating exogenous effects, such as treatment effects, in the presence of network interference when asymmetry matters.

\section{Empirical Application} \label{sec:empirics}
This section presents an empirical application using Wave I of the Add Health data. We analyze multiple outcomes to uncover different forms of asymmetry in peer effects that cannot be captured by the standard symmetric model.

\subsection{Add Health Data}
Wave I of the Add Health survey provides nationally representative and detailed information on students in grades 7--12 from 144 schools in the United States during the 1994--1995 school year. Approximately 90,000 students completed an in-school questionnaire covering demographics, family background, academic performance, health-related behaviors, and friendship networks. Each respondent could also nominate up to five male and five female best friends within the same school.

We analyze four outcomes: smoking, fighting, optimism, and drinking. Smoking and drinking are defined as the number of days per week that a student smokes tobacco or consumes alcohol, respectively. Fighting is measured as the number of times per year that a student engages in a physical fight. Optimism is constructed as the average of indicators that reflect whether students believe they will graduate from college and live to age 35. These indicators range from 0 to 8, with 8 indicating certainty about the outcome.

Although Add Health is one of the most comprehensive datasets for studying peer effects, it has some limitations. First, the observed degree may be censored because students can nominate at most 10 friends. Second, some nominated friends cannot be matched to valid student identifiers and are therefore excluded from the network, as is common in studies using this dataset.

Using the standard symmetric model, several studies have shown that degree censoring in this dataset induces only limited attenuation bias, while unmatched links may imply substantial bias \citep[see][]{griffith2022name, boucher2025estimating}. Although formally addressing the missing links problem is beyond the scope of this paper, it is important to note that the unmatched links may lead us to incorrectly classify an individual as isolated when none of their nominated friends can be matched. Because the outcome specifications differ between isolated and non-isolated individuals (see Equations~\eqref{eq:yred} and \eqref{eq:yrediso}), we conduct a robustness check in which we exclude individuals for whom none of their nominated friends can be matched (i.e., false isolates).\footnote{False isolates may still be nominated as friends by other individuals. We therefore do not remove them entirely from the network but exclude them only as focal observations. Consequently, this procedure does not generate additional missing links.} The results, reported in Table~\ref{tab:estimates_wofakeiso} in Supplemental Appendix~\ref{append:empirics}, are similar to those obtained using the full sample.

We control for a rich set of exogenous characteristics, including age, grade, sex, race, Hispanic ethnicity, and mother’s education and employment status. To mitigate the impact of missing links, we also include the number of unmatched friends and the number of matched friends as control variables. Additionally, we control for contextual variables defined as the averages of students' characteristics among friends.

After excluding observations with missing outcome values, our final sample consists of approximately 75{,}000 students from 140 schools. The average number of friends per student is 3.6, and 22\% of students have no friends. Summary statistics of the variables are reported in Supplemental Appendix~\ref{append:empirics}.

\subsection{Empirical Results}

\begin{table}[htbp]
\centering
\footnotesize
\begin{threeparttable}
\caption{Estimation Results}
\label{tab:SumEstimates}
\begin{tabular}{lcccc}
\toprule
\textbf{Outcome} & \multicolumn{2}{c}{\textbf{Smoking}} &
\multicolumn{2}{c}{\textbf{Fighting}} \\
Model & Asymmetric & Symmetric & Asymmetric & Symmetric \\
\midrule
$\beta^l$ & 1.219 & & -0.065 & \\
~ & (0.298) & & (0.087) & \\[1ex]
$\beta^h$ & 3.448 & & 1.115 & \\
~ & (0.449) & & (0.179) & \\[1ex]
$\beta$ & & 2.839 & & 0.232 \\
~ & & (0.458) & & (0.059) \\[1ex]
$\beta^h-\beta^l$ & 2.229 & & 1.180 & \\
~ & (0.281) & & (0.232) & \\[2.5ex]
Total Peer Effect & [0.549, 0.775] & 0.740 &
[-0.069, 0.527] & 0.188 \\[2.5ex]
\nth{1} stage F stat. ($\bar y$) & 88.899 & 171.597 &
42.730 & 64.382 \\
\nth{1} stage F stat. ($\check y$) & 62.752 & &
23.597 & \\[1ex]
KP LM test p-value & 0.000 & 0.000 &
0.002 & 0.000 \\
\end{tabular}
\begin{tabular}{lcccc}
\toprule
\textbf{Outcome} & \multicolumn{2}{c}{\textbf{Optimism}} &
\multicolumn{2}{c}{\textbf{Drinking}} \\
Model & Asymmetric & Symmetric & Asymmetric & Symmetric \\
\midrule
$\beta^l$ & 3.694 & & 0.842& \\
~ & (0.625) & & (0.184) & \\[1ex]
$\beta^h$ & -0.201 & & 1.206 & \\
~ & (0.162) & & (0.296) & \\[1ex]
$\beta$ & & 0.623 & & 0.897 \\
~ & & (0.093) & & (0.172) \\[1ex]
$\beta^h-\beta^l$ & -3.896 & & 0.364 & \\
~ & (0.737) & & (0.285) & \\[2.5ex]
Total Peer Effect & [-0.252, 0.787] & 0.384 &
[0.457, 0.547] & 0.473 \\[2.5ex]
\nth{1} stage F stat. ($\bar y$)  & 84.574 & 103.343 &
25.228 & 40.992
 \\
\nth{1} stage F stat. ($\check y$) & 22.324 & & 18.574 &  \\[1ex]
KP LM test p-value & 0.000 & 0.000 &
0.000 & 0.000 \\
\bottomrule
\end{tabular}
\begin{tablenotes}[para,flushleft]\footnotesize
Notes: Estimates are reported without parentheses, with standard errors (clustered at the network level) shown in parentheses. To generate the instruments, we use random forests with 5-fold cross-fitting (see Supplemental Appendix~\ref{append:instrument}) and 1,500 trees. For each training step, 20\% of the sample is used to tune the remaining hyperparameters. The row labeled ``Total Peer Effect'' corresponds to the range of marginal peer effects, defined by $\frac{\beta^l}{1 + \beta^l}$ and $\frac{\beta^h}{1 + \beta^h}$ (see Equation \eqref{eq:yi1}). The row labeled ``KP LM test p-value” reports the p-value of the LM test of \cite{kleibergen200}. The full table, including the coefficients on own and contextual control variables, is reported in the Supplemental Appendix (Table~\ref{tab:estimates}).
\end{tablenotes}
\end{threeparttable}
\end{table}

We estimate both the asymmetric and symmetric peer effects models for each outcome. For brevity, Table \ref{tab:SumEstimates} reports the estimates of the endogenous peer effects parameters, while the detailed estimation results, including the coefficients on the control variables, are reported in Supplemental Appendix \ref{append:empirics}. To assess instrument strength, we report the F-statistics from the first-stage regressions. We also report the Kleibergen--Paap (KP) rank LM test for underidentification \citep{kleibergen200}. The null hypothesis is that $\frac{1}{M} \sum_{m=1}^M \hat{\mathbf{Z}}_m^{\prime} \mathbf{V}_m$ is rank-deficient (see Assumption \ref{assumption:iv:fullrank}). Overall, these tests confirm that the instruments are relevant and rule out underidentification across all specifications and outcomes.

For three of the outcomes studied---namely, smoking, fighting, and optimism---peer effects exhibit different forms of asymmetry. For smoking, both types of peers exert strong conformity effects, although the influence of friends who smoke more than the individual is larger. Depending on peer composition, a one-unit increase in peers' smoking consumption implies an increase in the individual's smoking consumption ranging from 0.549 to 0.775 units.\footnote{The range of marginal peer effects is defined by $\frac{\beta^l}{1 + \beta^l}$ and $\frac{\beta^h}{1 + \beta^h}$ (see Equation \eqref{eq:yi1}).} The lower bound corresponds to a situation in which all friends smoke less than the individual, whereas the upper bound corresponds to a situation in which all friends smoke more than the individual. Exposure to friends who smoke more than the individual therefore has a larger effect than exposure to friends who smoke less, and this difference is statistically significant. In contrast, the standard symmetric model primarily captures the effect of exposure to friends who smoke more. Given the estimated standard errors, we cannot rule out the possibility that $\beta$ falls outside $[\beta^l,\,\beta^h]$.

Asymmetry is more pronounced for fighting and optimism, with stronger conformity toward higher- and lower-performing friends, respectively. In the case of fighting, the asymmetric model indicates that a one-unit increase in friends' fighting behavior induces peer effects ranging from $-0.069$ to $0.527$, where the lower bound is not statistically different from zero and corresponds to the effects from lower-performing peers. Students imitate friends who are more aggressive than themselves but show little or no response to less aggressive friends. This means that, depending on the status composition of the peer group, students may be either completely unresponsive or strongly responsive to peer behavior. The pattern reverses for optimism, where a one-unit increase in friends' optimism behavior index induces peer effects ranging from $-0.252$ to $0.787$. Here again, the lower bound is not statistically different from zero, but corresponds to the effects from higher-performing peers. Students show no response to friends who are more optimistic than themselves but conform strongly to more pessimistic friends. 

However, for both fighting and optimism, the symmetric model masks this heterogeneity by indicating that all types of peers exert the same influence. In particular, this model suggests that students strongly conform to less aggressive friends and to more optimistic friends, although these friends actually exert little or no influence. Moreover, it underestimates the effects of the friends who actually exert strong influence, namely more aggressive friends in the case of fighting and more pessimistic friends in the case of optimism.

Lastly, for drinking, the estimated effects of friends who drink more and friends who drink less are similar in magnitude and statistically indistinguishable. In this case, the standard model provides an adequate description of peer influence, and the symmetry restriction cannot be statistically rejected.

Interestingly, all socially undesirable behaviors---namely smoking, fighting, and drinking---exhibit stronger peer effects from higher-performing peers, although the asymmetry is not statistically significant in the case of drinking. These results suggest that, for such outcomes, the most active individuals are likely the most effective targets for amplifying the effects of policy interventions aimed at mitigating these behaviors through the network. This finding provides additional insight into the results of \citet{lee2021}, who show that key players in networks involving risky behaviors, such as juvenile delinquency, are not necessarily the most active individuals. Because their analysis relies on a standard symmetric model, it does not account for the possibility that the most active individuals, although not central nodes in the network, may exert stronger influence on their peers.

Furthermore, the asymmetry in peer effects has the opposite targeting implication for optimism, a socially desirable outcome. Improving optimism may be important for fostering aspirations and self-esteem, both of which are key determinants of long-term welfare \citep{bernard2026future,genicot2020aspirations}. Our results indicate that, in this case, policy interventions should target the most pessimistic students to maximize their effectiveness.

\subsection{Policy Intervention} \label{application:policy}
We study the policy implications of the empirical results. We consider an intervention in which a social planner seeks to improve or reduce the aggregate outcome, $\sum_{i = 1}^{n_m} y_{m,i}$, in some network $m$, depending on whether the outcome is socially desirable or undesirable, respectively. This is achieved by treating students, where treatment consists of increasing productivity $\alpha_{m,i}$ by one unit when the outcome is socially desirable, or decreasing $\alpha_{m,i}$ by one unit when the outcome is socially undesirable.\footnote{Our results are not driven by the specific choice of increasing or decreasing productivity by one unit. Similar results are obtained under alternative treatment intensities.} As in \cite{galeotti2020}, we assume that resources are limited, so that only $\kappa$ students can be treated. We vary $\kappa$ from 1 to $n_m$ and show that ignoring asymmetry may lead to inefficient allocations. This misallocation can result in substantial losses in aggregate outcomes and may even render the policy counterproductive.

\subsubsection{Treatment Allocation in the Asymmetric Model}
As shown in Section \ref{microfoundation::socialmultiplier}, the optimal order in which students should be treated under the symmetric model is determined by the ranking induced by the vector $\boldsymbol{\Delta}$, which admits a closed form. Consequently, for any budget $\kappa$, the optimal set of students can be identified straightforwardly under this model. By contrast, because the asymmetric model does not admit a closed-form solution for the outcome, determining the optimal allocation is considerably more challenging. It requires an exhaustive search over all possible combinations of $\kappa$ students from a school of size $n_m$. Even for moderate values of $\kappa$ and $n_m$, the number of possible allocations quickly reaches several billions, rendering such a search computationally infeasible.\footnote{For example, treating half of the students in a school of 100 yields more than $10^{29}$ possible combinations.}

To circumvent this challenge, we propose two sequential algorithms---referred to as the ``forward'' and ``backward'' algorithms---that approximate the \textit{oracle} optimal treatment allocation. For a budget $\kappa$, we denote by $\mathcal T_{m,\kappa}$ the set of treated students selected by a given approximation algorithm.

The ``forward'' algorithm first addresses the case $\kappa = 1$, which is straightforward because it involves only $n_m$ possible allocations. It then proceeds recursively. Given $\kappa < n_m$ and $\mathcal T_{m,\kappa}$, the allocation for $\kappa + 1$ is obtained under the assumption that $\mathcal T_{m,\kappa} \subset \mathcal T_{m,\kappa+1}$. This assumption reduces the search space to $n_m - \kappa$ candidate allocations, since one only needs to identify the single student to add to $\mathcal T_{m,\kappa}$ to obtain $\mathcal T_{m,\kappa+1}$. 

The ``backward'' algorithm proceeds similarly, but in reverse. It first addresses the case of $\kappa = n_m$, which involves only one possible allocation. Then, for any $\kappa > 1$, the allocation for $\kappa - 1$ is obtained under the assumption that $\mathcal T_{m,\kappa - 1} \subset \mathcal T_{m,\kappa}$, reducing the search space to $\kappa$ candidate allocations. This consists of identifying the student to remove from $\mathcal T_{m,\kappa}$ to obtain $\mathcal T_{m,\kappa - 1}$.

Implementing these algorithms for all values of $\kappa$, from 1 to $n_m$, reduces the computational complexity from $2^{n_m}$ to $n_m^2$, making the problem tractable even for large networks. However, the two approximations generally do not yield the same solution and may differ from the oracle allocation because the latter does not necessarily satisfy $\mathcal T_{m,\kappa} \subset \mathcal T_{m,\kappa+1}$. Unfortunately, it is challenging to quantify the regret of these approximations relative to the oracle allocation. Nevertheless, the degree of agreement between the two algorithms provides a useful diagnostic. When the allocations they produce are similar, the planner may have greater confidence in the robustness of the resulting recommendation. In practice, for each value of $\kappa$, the planner can select the algorithm that delivers a better-performing allocation.

\subsubsection{Intervention Results}
We consider the cases where the students to target are selected according to either the symmetric model or the forward or backward approximation of the asymmetric model. Once the selected students are treated, we evaluate the resulting aggregate outcome using the asymmetric model, which we regard as the true data-generating process because it is more general. For each outcome, we present results for two randomly selected schools (see Figure~\ref{fig:intervention}). Panel A presents the spillover generated by each targeting rule. For a socially desirable outcome, the spillover is defined as the difference between the change in the aggregate outcome and $\kappa$, expressed as a percentage of $\kappa$.\footnote{For example, for a socially desirable outcome, the spillover is 50\% when the aggregate outcome increases by 15 while 10 students are treated.} For socially undesirable outcomes, it is defined analogously with the opposite sign. Panel B reports the loss in spillover associated with the symmetric allocation relative to the asymmetric allocations.

\afterpage{\begin{landscape}
\begin{figure}
    \centering
    \includegraphics[scale = 0.58]{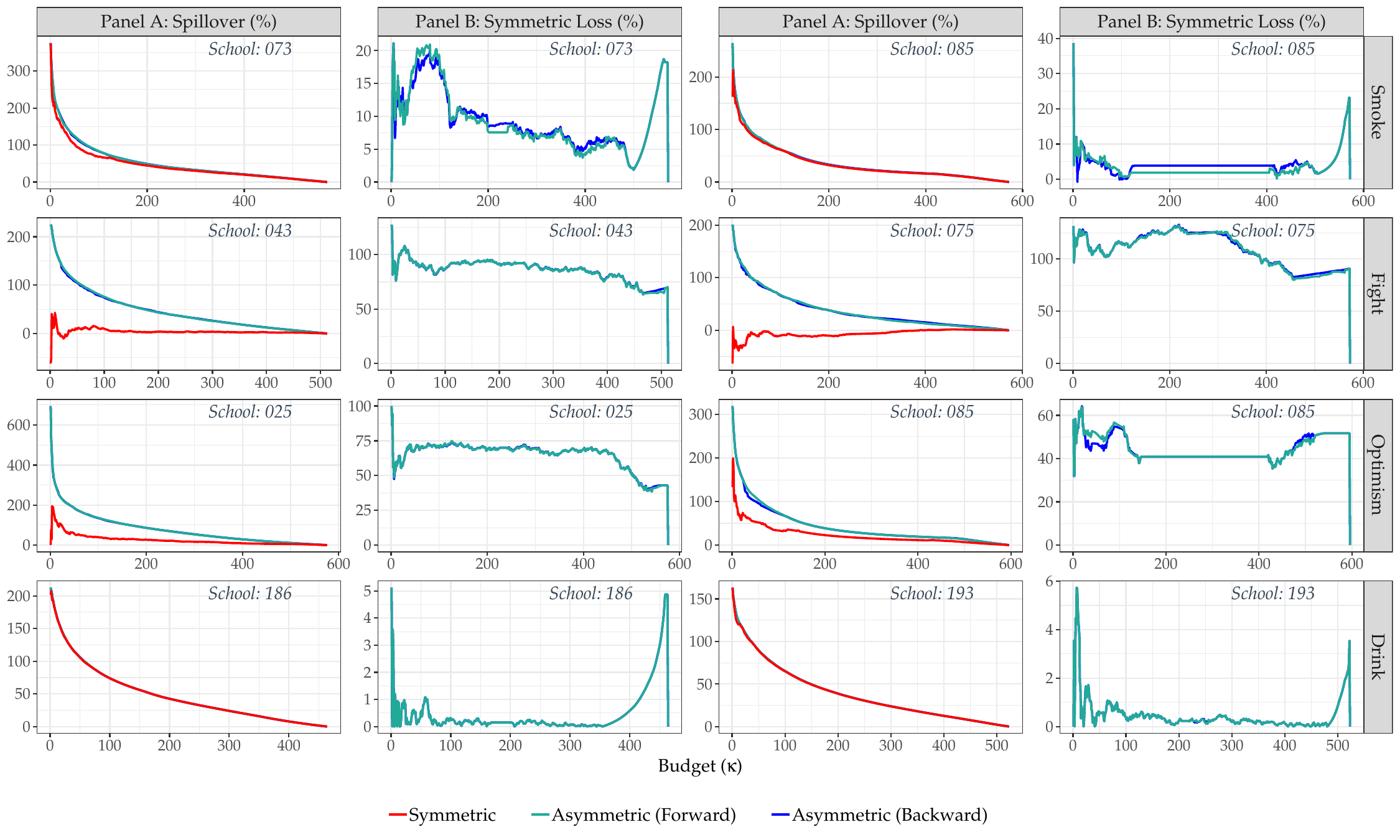}
    \caption{Spillover losses from using the symmetric model to allocate a treatment}
    \justifying
    \footnotesize{\noindent Notes: For each outcome, two schools are randomly selected from among those with $n_m \in [\bar n - 100, \, \bar n + 100]$, where $\bar n$ is the average school size. For each school, Panel A plots the spillover effect of the intervention as a function of the intervention budget $\kappa$. Panel B displays the loss in spillover associated with the symmetric allocation relative to the forward and backward asymmetric allocations.}
    \label{fig:intervention}
\end{figure}
\end{landscape}}

In general, both targeting rules of the asymmetric model yield similar results, with the forward approximation outperforming the backward approximation for certain values of $\kappa$, and vice versa. This similarity provides reassurance regarding the robustness of the approximations.

Although asymmetry matters for smoking, the difference in spillover effects under the symmetric targeting rule is not large. This is because both higher- and lower-performing peers exert strong influence, so that ignoring asymmetry generates only limited inefficiency. Using the symmetric model for treatment assignment leads to a 5\% to 15\% loss in spillovers. Yet, even such a modest loss may translate into substantial monetary costs when the policy is implemented in a large school.

For fighting and optimism, however, asymmetry is substantially more important. The results indicate that using the symmetric model may generate no spillover effects, or may even lead to negative spillovers, which is worse than what would arise in a benchmark where students do not interact. In the case of fighting, the symmetric model assigns treatment to students who are less aggressive friends, while these students exert little or no influence. This generates no spillovers or slightly negative spillovers. In this case, the loss in the aggregate outcome when the planner uses the symmetric model exceeds 70\% and may reach 100\%. The results are similar for optimism. The symmetric model assigns treatment to more optimistic friends, even though these friends exert little or no influence. As a result, the loss in the aggregate outcome generally ranges from 40\% to 75\%.

Finally, for drinking, using the symmetric model generates virtually no loss. This is expected because peer effects are approximately symmetric for this outcome. The loss is generally below 1\% and is likely not statistically significant.

\section{Conclusion} \label{sec:conclu}

This paper develops a model of asymmetric peer effects in which individuals respond differently to peers who outperform them and to those who underperform them. We propose an econometric approach to identify and estimate both effects and show that this asymmetry is empirically pervasive and consequential for policy. Standard peer effects models overlook this distinction and instead recover an aggregate effect that is not simply the average of the two asymmetric effects. In fact, this aggregate effect can lie outside the range defined by the underlying asymmetric effects. Consequently, policy interventions based on the symmetric model may be inefficient or even harmful. Our empirical application demonstrates the relevance of these asymmetric effects across several student outcomes. In targeted interventions, we find that ignoring this asymmetry may generate substantial losses in spillovers.

This paper contributes to the recent literature on heterogeneous peer effects, where heterogeneity is often introduced through the distribution of peer outcomes \citep{boucher2023}. Although this form of heterogeneity is relevant, our results show that it does not fully capture how individual decisions are formed. Decisions also depend on how individuals evaluate themselves relative to their peers. As a result, even when facing the same distribution of peer outcomes, individuals may respond differently depending on their own performance. Since overlooking this distinction may lead to counterproductive interventions, standard peer effects models should be interpreted with caution. More broadly, our results emphasize that policy-oriented research on peer effects should rely on flexible models to better understand the mechanisms through which social interactions shape behavior.
\clearpage
\newpage
\appendix
\numberwithin{equation}{section}
\numberwithin{figure}{section}
\numberwithin{table}{section}
\numberwithin{lemma}{section}
\numberwithin{assumption}{section}
\numberwithin{proposition}{section}

\section{Proofs}
\subsection{Equilibrium Uniqueness} \label{proof:uniqueness}
Let  $\mathcal{A}_i = \{y_j : g_{ij} > 0\} \cup \{-\infty, +\infty\}$ denote the set of effort levels chosen by individual $i$’s peers, augmented with $-\infty$ and $+\infty$. Let  $a_1, \dots, a_{q}$ denote the ordered elements of $\mathcal{A}_i$; that is, $a_1 < a_2< \dots < a_q$, with $a_1 = -\infty$ and $a_q = +\infty$. 

We first state and show the following lemmas.
\begin{lemma}\label{lemma:differentiable}
 The utility function \eqref{eq:utility} is continuously differentiable in $y_i$.
\end{lemma}
\begin{proof}
The utility function can be written as:
\begin{equation}\label{Eq:Utility2}
    U_i(y_i, \mathbf{y}_{-i})=\alpha_iy_i - \frac{y_i^2}{2} -\dfrac{1}{2}\sum_{j \ne i} \left(\beta^l\mathbbm{1}\{y_j\leq y_i\} + \beta^h\mathbbm{1}\{y_j> y_i\}\right)g_{ij}(y_i - y_j)^2,
\end{equation}
where $\mathbbm{1}\{\cdot\}$ denotes the indicator function. 

The terms $\mathbbm{1}\{y_j \leq y_i\}$ and $\mathbbm{1}\{y_j > y_i\}$ are continuously differentiable in $y_i$ for all $y_i \ne y_j$. Thus, the function $U_i(y_i, \mathbf{y}_{-i})$ is continuously differentiable in $y_i$ for all $y_i \in \mathbb{R} \setminus \mathcal{A}_i$. It is therefore sufficient to show that that $U_i(y_i, \mathbf{y}_{-i})$ is also continuously differentiable at any finite $y_i \in \mathcal{A}_i$; that is, $y_i = a_\ell$ for $\ell \in \{2, \,\dots, \, q-1\}$.\footnote{Such a $y_i$ exists only for individuals with at least one friend. For individuals without friends, the lemma holds because $U_i(y_i, \mathbf{y}_{-i}) = \alpha_i y_i - \frac{y_i^2}{2}$, which is continuously differentiable in $y_i$ for all $y_i \in \mathbb{R}$.}

For any $\ell \in \{2, \,\dots, \, q-1\}$, if $y_i \in (a_{\ell-1}, a_\ell)$, then $y_i$, $\mathbbm{1}\{y_j\leq y_i\} = \mathbbm{1}\{y_j< a_\ell\}$ and $\mathbbm{1}\{y_j> y_i\}=\mathbbm{1}\{y_j\geq a_\ell\}$. The utility function can thus be written as:
\begingroup
\allowdisplaybreaks
\begin{align*}
    U_i(y_i, \mathbf{y}_{-i})&=\alpha_iy_i - \frac{y_i^2}{2} -\dfrac{1}{2}\sum_{j \ne i} \left(\beta^l\mathbbm{1}\{y_j< a_{\ell}\} + \beta^h\mathbbm{1}\{y_j\geq a_\ell\}\right)g_{ij}(y_i - y_j)^2,\\
    \begin{split}
    U_i(y_i, \mathbf{y}_{-i})&= \alpha_iy_i - \frac{y_i^2}{2} -\dfrac{\beta^h}{2}n_i^{\ast}(y_i - a_\ell)^2 -\\
    &\quad\quad \dfrac{1}{2}\sum_{j \ne i} \left(\beta^l\mathbbm{1}\{y_j< a_{\ell}\} + \beta^h\mathbbm{1}\{y_j> a_\ell\}\right)g_{ij}(y_i - y_j)^2,
    \end{split}
\end{align*}
\endgroup
where $n_i^{\ast} = \sum_{j \ne i} g_{ij} \mathbbm{1}\{y_j = a_\ell\}$. The function $U_i(y_i, \mathbf{y}_{-i})$ is differentiable in $y_i$ on the interval $(a_{\ell-1}, a_\ell)$, since this intervals do not contain any point of $\mathcal{A}_i$. It is easy to verify that the limit of its derivative as $y_i$ approaches $a_\ell$ from the left is:
\begin{equation*}
    \lim_{y_i \underset{<}{\to} a_\ell} \frac{\partial U_i(y_i, \mathbf{y}_{-i})}{\partial y_i}=\alpha_i-a_\ell-\sum_{j \ne i} \left(\beta^l\mathbbm{1}\{y_j< a_{\ell}\} + \beta^h\mathbbm{1}\{y_j> a_\ell\}\right)g_{ij}(a_\ell - y_j).
\end{equation*}

Similarly, if $y_i \in (a_\ell, a_{\ell+1})$, the utility function can be expressed as:
\begingroup
\allowdisplaybreaks
\begin{align*}
    U_i(y_i, \mathbf{y}_{-i})&=\alpha_iy_i - \frac{y_i^2}{2} -\dfrac{1}{2}\sum_{j \ne i} \left(\beta^l\mathbbm{1}\{y_j\leq a_{\ell}\} + \beta^h\mathbbm{1}\{y_j > a_{\ell}\}\right)g_{ij}(y_i - y_j)^2,\\
    \begin{split}
    U_i(y_i, \mathbf{y}_{-i})&=-\dfrac{1}{2}\sum_{j \ne i} \left(\beta^l\mathbbm{1}\{y_j< a_{\ell}\} + \beta^h\mathbbm{1}\{y_j> a_\ell\}\right)g_{ij}(y_i - y_j)^2 - \\
    &\quad\quad \dfrac{\beta^l}{2}n_i^{\ast}(y_i - a_\ell)^2 + \alpha_iy_i - \frac{y_i^2}{2}.
    \end{split}\end{align*}
\endgroup
Thus the limit of $\frac{\partial U_i(y_i, \mathbf{y}_{-i})}{\partial y_i}$ as $y_i$ approaches $a_\ell$ from the right is:
\begin{equation*}
    \lim_{y_i \underset{>}{\to} a_\ell} \frac{\partial U_i(y_i, \mathbf{y}_{-i})}{\partial y_i}=\alpha_i-a_\ell-\sum_{j \ne i} \left(\beta^l\mathbbm{1}\{y_j< a_{\ell}\} + \beta^h\mathbbm{1}\{y_j> a_\ell\}\right)g_{ij}(a_\ell - y_j).
\end{equation*}

Since the left-hand and right-hand limits of $\frac{\partial U_i(y_i, \mathbf{y}_{-i})}{\partial y_i}$ at $y_i = a_\ell$ are equal, it follows that $U_i(y_i, \mathbf{y}_{-i})$ is continuously differentiable at every $y_i \in \mathcal{A}_i$. Therefore, $U_i(y_i, \mathbf{y}_{-i})$ is continuously differentiable in $y_i$ over $\mathbb{R}$.
\end{proof}

\begin{lemma}\label{lemma:concave}
Under Assumption \eqref{assumption:minbeta}, the utility function \eqref{eq:utility} is strictly concave in $y_i$.
\end{lemma}
\begin{proof}
For any interval $(a_\ell, a_{\ell+1})$ and any $y_i \in (a_{\ell}, a_{\ell+1})$, we have $\mathbbm{1}\{y_j \leq y_i\} = \mathbbm{1}\{y_j \leq a_\ell\}$ and $\mathbbm{1}\{y_j > y_i\} = \mathbbm{1}\{y_j > a_\ell\}$. Substituting these equalities into the expression of $U_i(y_i, \mathbf y_{-i})$ given in Equation~\eqref{Eq:Utility2} and differentiating with respect to $y_i$ yields:
\begin{equation}\label{eq:derivative}
    \dfrac{\partial U_i(y_i, \mathbf y_{-i})}{\partial y_i} = \alpha_i - y_i - \textstyle\sum_{j \ne i} \left(\beta^l\mathbbm{1}\{y_j \leq a_{\ell}\} + \beta^h\mathbbm{1}\{y_j > a_\ell\}\right)g_{ij}(y_i - y_j).
\end{equation}
One can notice that $\frac{\partial U_i(y_i, \mathbf{y}_{-i})}{\partial y_i}$ is linear in $y_i$ on each $(a_\ell, a_{\ell+1})$. However, its slope is likely to differ across the intervals $(a_\ell, a_{\ell+1})$. Thus, $U_i(y_i, \mathbf{y}_{-i})$ is \textit{not} twice differentiable at $y_i \in \mathcal{A}_i$. Nonetheless, since $\frac{\partial U_i(y_i, \mathbf{y}_{-i})}{\partial y_i}$ is continuous (Lemma~\ref{lemma:differentiable}), it suffices to show that $\frac{\partial^2 U_i(y_i, \mathbf{y}_{-i})}{\partial y_i^2} < 0$ for all $y_i \in (a_{\ell}, a_{\ell+1})$. This would imply that $\frac{\partial U_i(y_i, \mathbf{y}_{-i})}{\partial y_i}$ is strictly decreasing on each $(a_\ell, a_{\ell+1})$. By continuity, it is then strictly decreasing for all $y_i \in \mathbb{R}$, implying that $U_i(y_i, \mathbf{y}_{-i})$ is strictly concave in $y_i$. 

For any finite $y_i \in (a_{\ell}, a_{\ell+1})$, the second derivative is given by:
\begin{align*}
    \dfrac{\partial^2 U_i(y_i, \mathbf y_{-i})}{\partial y_i^2} = -1 - \sum_{j \ne i} \left(\beta^l\mathbbm{1}\{y_j \leq a_{\ell}\} + \beta^h\mathbbm{1}\{y_j > a_\ell\}\right)g_{ij}.
\end{align*}
If $\beta^h \geq \beta^l$, then $\beta^l\mathbbm{1}\{y_j \leq a_{\ell}\} + \beta^h\mathbbm{1}\{y_j > a_\ell\} \geq \beta^l$ because $\mathbbm{1}\{y_j \leq a_\ell\} + \mathbbm{1}\{y_j > a_\ell\} = 1$. Thus,
$\frac{\partial^2 U_i(y_i, \mathbf y_{-i})}{\partial y_i^2} \leq -1 - \beta^l \sum_{j \ne i} g_{ij}$. Similarly, if $\beta^h \leq \beta^l$, then
$\frac{\partial^2 U_i(y_i, \mathbf y_{-i})}{\partial y_i^2} \leq -1 - \beta^h \sum_{j \ne i} g_{ij}$. Since $\beta^l > -\frac{1}{2}$ and $\beta^h > -\frac{1}{2}$ by Assumption~\ref{assumption:minbeta}, and $\sum_{j \ne i} g_{ij} = 1$, it follows that $\frac{\partial^2 U_i(y_i, \mathbf y_{-i})}{\partial y_i^2} < 0$. As a result, $U_i(y_i, \mathbf y_{-i})$ is strictly concave in $y_i$.
\end{proof}

\begin{lemma}\label{lemma:BRFcontinuous}
Let $\displaystyle \bar{y}_i^l = \sum_{j \ne i} \mathbbm{1}\{y_j \le y_i\} g_{ij} y_j$, 
$\displaystyle \bar{y}_i^h = \sum_{j \ne i} \mathbbm{1}\{y_j > y_i\} g_{ij} y_j$, 
$\displaystyle g_i^l = \sum_{j \ne i} \mathbbm{1}\{y_j \le y_i\} g_{ij}$, 
and 
$\displaystyle g_i^h = \sum_{j \ne i} \mathbbm{1}\{y_j > y_i\} g_{ij}$. The following results hold under Assumption~\ref{assumption:minbeta}. \\
\begin{inparaenum}[(i)]
\item The best-response function of $i$ is $b_i(\mathbf y_{-i}) = y_i$, where $y_i$ is implicitly defined by:
$$
y_i = \frac{\alpha_i + \beta^l \bar{y}_i^l + \beta^h \bar{y}_i^h}{1 + \beta^l g_i^l + \beta^h g_i^h}.
$$
\label{lemma:BRFcontinuous:i}
\item For any $j \ne i$, $b_i(\mathbf y_{-i})$ is continuous in $y_j \in \mathbb{R}$. \label{lemma:BRFcontinuous:ii}
\end{inparaenum}
\end{lemma}
\begin{proof}
For any $y_i \in \mathbb R$, there exists $a_{\ell}, ~a_{\ell + 1} \in \mathcal{A}_i$ such that $y_i \in (a_\ell, ~ a_{\ell + 1})$ and $\frac{\partial U_i(y_i, \mathbf y_{-i})}{\partial y_i}$ is given by Equation \eqref{eq:derivative}. Since $y_i \in (a_\ell, ~ a_{\ell + 1})$ implies that $\mathbbm{1}\{y_j \leq a_{\ell}\} = \mathbbm{1}\{y_j \leq y_i\}$ and that $\mathbbm{1}\{y_j > a_\ell\} = \mathbbm{1}\{y_j > y_i\}$, we can also write $\frac{\partial U_i(y_i, \mathbf y_{-i})}{\partial y_i}$ as:
\begin{equation}\label{eq:derivative:withouta}
    \dfrac{\partial U_i(y_i, \mathbf y_{-i})}{\partial y_i} = \alpha_i- y_i - \textstyle\sum_{j \ne i} \left(\beta^l\mathbbm{1}\{y_j \leq y_i\} + \beta^h\mathbbm{1}\{y_j > y_i\}\right)g_{ij}(y_i - y_j).
\end{equation}

By continuity of $\frac{\partial U_i(y_i, \mathbf y_{-i})}{\partial y_i}$, this expression also holds on any closed (or half-closed) interval $[a_\ell, a_{\ell+1}]$, provided that both $a_\ell$ and $a_{\ell+1}$ are finite. Since the utility function is continuously differentiable and strictly concave, it admits a unique maximizer $b_i(\mathbf y_{-i}) = y_i$, which satisfies the first-order condition $\frac{\partial U_i(y_i, \mathbf y_{-i})}{\partial y_i} = 0$. From Equation \eqref{eq:derivative:withouta}, $y_i$ is implicitly given by:
$$y_i = \frac{\alpha_i + \beta^l \bar{y}_i^l + \beta^h \bar{y}_i^h}{1 + \beta^l g_i^l + \beta^h g_i^h}.
$$

Furthermore, given that $\frac{\partial U_i(y_i, \mathbf y_{-i})}{\partial y_i}$ is continuous and strictly decreasing in $y_i$ (Lemmas~\ref{lemma:differentiable} and \ref{lemma:concave} ), continuity of the best response follows once we show that $\frac{\partial U_i(y_i, \mathbf y_{-i})}{\partial y_i}$ is continuous in $y_j$, for any $j\ne i$. 

From Equation \eqref{eq:derivative:withouta}, it is clear that $\frac{\partial U_i(y_i, \mathbf{y}_{-i})}{\partial y_i}$ is continuous in $y_j$ for all $y_j \ne y_i$. To ensure continuity at $y_j = y_i$, we verify that the left and right limits of $\frac{\partial U_i(y_i, \mathbf{y}_{-i})}{\partial y_i}$ as $y_j$ approaches $y_i$ are equal. This holds because, although $\beta^l \mathbbm{1}\{y_j \leq y_i\} + \beta^h \mathbbm{1}\{y_j > y_i\}$ may be discontinuous at $y_j = y_i$, it is multiplied by $(y_i - y_j)$, which tends to zero. Therefore, the term
$\left( \beta^l \mathbbm{1}\{y_j \leq y_i\} + \beta^h \mathbbm{1}\{y_j > y_i\} \right) g_{ij} (y_i - y_j)$
tends to zero as $y_j$ approaches $y_i$ from either side. As a result, $\frac{\partial U_i(y_i, \mathbf{y}_{-i})}{\partial y_i}$ is continuous in $y_j$. This completes the proof of the lemma.\end{proof}

\subsubsection{Proof of Proposition \eqref{propo:equilibrium}} \label{proof:uniqueness:theo}
We employ the contraction mapping theorem. For any two strategy profiles $\mathbf{y} = (y_1, \dots, y_n)^{\prime}$ and $\tilde{\mathbf{y}} = (\tilde{y}_1, \dots, \tilde{y}_n)^{\prime}$ in $\mathbb{R}^n$, we show that $\lvert b_i(\mathbf{y}_{-i}) - b_i(\tilde{\mathbf{y}}_{-i})\rvert \leq \mu \lVert \mathbf{y} - \tilde{\mathbf{y}} \rVert_{\infty}$ where $\mu < 1$.\footnote{For any $\mathbf a = (a_1, ~\dots, ~ a_n)^{\prime} \in \mathbb R^n$, the infinity norm of $\mathbf a$ is defined as $\displaystyle \lVert \mathbf a \rVert_{\infty} = \max_{i} \lvert a_i \rvert$.} However, the best-response function $b_i(\mathbf{y}_{-i})$ is not always differentiable in $y_j$ for a friend $j$. Let
$$\mathcal{D} = \left\{\mathbf{y}_{-i}\in \mathbb{R}^{n-1}:\, b_i(\mathbf{y}_{-i})\text{ is not differentiable in } \mathbf{y}_{-i}\right\}.$$ Note that $b_i(\mathbf{y}_{-i})$ is not differentiable in $y_j$ only when $y_j = y_i$.\footnote{This is because of the term $\beta^l \mathbbm{1}\{y_j \leq y_i\} + \beta^h \mathbbm{1}\{y_j > y_i\}$ in the denominator of Equation~\eqref{eq:yi1}.} Since $y_i = b_i(\mathbf{y}_{-i})$, for any $\mathbf{y}_{-i} \in \mathcal{D}$ there exists a friend $j$ such that $b_i(\mathbf{y}_{-i}) = $ $y_j$.

\bigskip
\paragraph{Step 1: {\normalfont Characterization of the subset \texorpdfstring{$\mathcal{D}$}{TEXT}}}
\medskip

\noindent Since each friend’s status can be either a low or high performer, there are $2^{n_i}$ possible combinations of statuses. For a given combination, we have either $b_i(\mathbf{y}_{-i}) \leq y_j$ or $b_i(\mathbf{y}_{-i}) > y_j$, depending on whether $j$ is classified as a low or high performer, respectively. Given that $\mathbbm{1}\{y_j \leq y_i\}$ remains constant for each combination, the implicit expression of $b_i(\mathbf{y}_{-i})$ in Lemma \ref{lemma:BRFcontinuous} implies that the conditions $b_i(\mathbf{y}_{-i}) \leq y_j$ and $b_i(\mathbf{y}_{-i}) > y_j$ define linear inequalities in $\mathbf{y}_{-i}$. Taking the inequalities for all friends $j$ yields a system of linear inequalities in $\mathbf{y}_{-i}$. Consequently, each combination forms a polyhedron, which we refer to as a \textit{cell}. However, because these cells are defined in an $n_i$-dimensional space, at least $n_i + 1$ inequalities are required to form a bounded region. Since we only have $n_i$ such inequalities, each combination produces an unbounded cell. 

Considering all possible combinations, we can partition the space of $\mathbf{y}_{-i}$ into $2^{n_i}$ cells for each $i$. Every cell is delimited by $n_i$ facets, which are half-hyperplanes characterized by the equations $b_i(\mathbf{y}_{-i}) = y_j$ for each of the $n_i$ friends $j$. Within a cell, the status of every friend remains unchanged.

We provide an illustration in Figure~\ref{fig_equil} for an individual $i_1$ in a network of three individuals, where $i_2$ and $i_3$ are both friends of $i_1$. Since $n_{i_1} = 2$, the space is partitioned into four cells, which are angular sectors delimited by two facets. Each facet is a half-line characterized by the equation $b_{i_1}(\mathbf{y}_{-i_1}) = y_{i_2}$ or $b_{i_1}(\mathbf{y}_{-i_1}) = y_{i_3}$.

\bigskip
\begin{figure}[!ht]
	\centering
	\footnotesize
	\begin{tikzpicture}[scale=1]
        \begin{axis}[
            axis x line=bottom,    % put x-axis at the bottom
            axis y line=left,      % put y-axis at the left
            xlabel={$y_{i_2}$}, ylabel={$y_{i_3}$},
            xtick distance=10,      % keep x ticks 
            ytick distance=10,      % keep x ticks 
            ymin=-20, ymax=20,
            samples=2,
            width=12cm, height=8cm,
        ]
        \node[right, align=left, black] at (axis cs:-18,10) {\footnotesize Cell 1: $i_2$ is a lower \\\footnotesize   and $i_3$ is a higher};
        \node[right, align=left, black] at (axis cs:4,10) {\footnotesize Cell 2: $i_2$ and $i_3$ \\\footnotesize  are higher};
        \node[right, align=left, black] at (axis cs:-18,-14) {\footnotesize Cell 3: $i_2$ and $i_3$ \\\footnotesize  are lower};
        \node[right, align=left, black] at (axis cs:4,-14) {\footnotesize Cell 4: $i_2$ is a higher \\\footnotesize and $i_3$ is a lower};
            % Parameters
            \pgfmathsetmacro{\a}{1}       % alpha
            \pgfmathsetmacro{\bl}{2.1}    % betal
            \pgfmathsetmacro{\bh}{0.4}    % betah (unused here)
            \pgfmathsetmacro{\xmin}{-20}  %xmin
            \pgfmathsetmacro{\xmax}{15}   %xmax
            \pgfmathsetmacro{\pxo}{-10}   % x for the point w0
            \pgfmathsetmacro{\pyo}{3}     % y for the point w0
            \pgfmathsetmacro{\pxt}{1.5}   % x for the point w3
            \pgfmathsetmacro{\pyt}{-6}    % y for the point w3
            % Compute slope and intercept for line through w0 w3
            \pgfmathsetmacro{\aww}{(\pyo - \pyt) / (\pxo - \pxt)}             
            \pgfmathsetmacro{\bww}{(\pyo - \aww * \pxo)} 
            % Compute w1 = (w0x,woy) 
            \pgfmathsetmacro{\wox}{(\bww*(2 + \bl) - 2*\a) / (\bl - \aww*(2 + \bl))}
            \pgfmathsetmacro{\woy}{\aww * \wox + \bww}
             % Compute w2 = (w2x,w2y) 
            \pgfmathsetmacro{\wtx}{(2*\a + \bl * \bww)/(2 + \bl*(1 - \aww))}
            \pgfmathsetmacro{\wty}{(\aww * \wtx + \bww)}
            % yi3 < yi2 and yi1 = yi3
            \addplot[blue, line width=0.6pt, domain=\a:\xmax] {(2*\a + \bh*x)/(2+\bh)}
            node[pos=1, sloped, below left] {$b_{i_1}(\mathbf{y}_{-i_1}) = y_{i_3}$};
            % yi3 < yi2 and yi1 = yi2
            \addplot[red, line width=0.6pt, domain=\xmin:\a] {(-2*\a + (2 + \bl)*x)/(\bl)};
            % yi3 > yi2 and yi1 = yi3
            \addplot[blue, line width=0.6pt, domain=\xmin:\a] {(2*\a + \bl*x)/(2 + \bl)};
            % yi3 > yi2 and yi1 = yi2
            \addplot[red, line width=0.6pt, domain=\a:\xmax] {(-2*\a + (2 + \bh)*x)/\bh}
            node[pos=.22, sloped, above left] {$b_{i_1}(\mathbf{y}_{-i_1}) = y_{i_2}$};
            % line from y to tilde y
            \draw[thick, black!20, dashed] (axis cs:\pxo, \pyo) -- (axis cs:\pxt, \pyt);
            \node[left, align=left, gray] at (axis cs:\pxo, \pyo){\footnotesize $(y_{i_2}, y_{i_3})$};
            \node[right, align=right, gray] at (axis cs:\pxt, \pyt){\footnotesize $(\tilde y_{i_2}, \tilde y_{i_3})$};
            \node[left, gray] at (axis cs:\wox - 0.1,\woy){\footnotesize $\boldsymbol{\omega}^{(1)}$};
            \node[right, gray] at (axis cs:\wtx,\wty + 1){\footnotesize $\boldsymbol{\omega}^{(2)}$};
            \addplot[
                thick, black!60,
                mark=x,
                mark options={scale=1},
                only marks=true
            ] coordinates {
                (\pxo, \pyo)
                (\pxt, \pyt)
                (\wox, \woy)
                (\wtx, \wty)
            };
        \end{axis}

    \end{tikzpicture}
	\caption{Statuses of $i_2$ and $i_3$ and the space of their outcome}
	\label{fig_equil}
    \justifying
    \noindent \footnotesize{Note: this figure illustrates the statuses of friends $i_2$ and $i_3$ in the space of $(y_{i_2}, y_{i_3})$. We set $\alpha_{i_1} = 1$, $\beta^l = 2.1$ and $\beta^h = 0.4$.}
\end{figure}
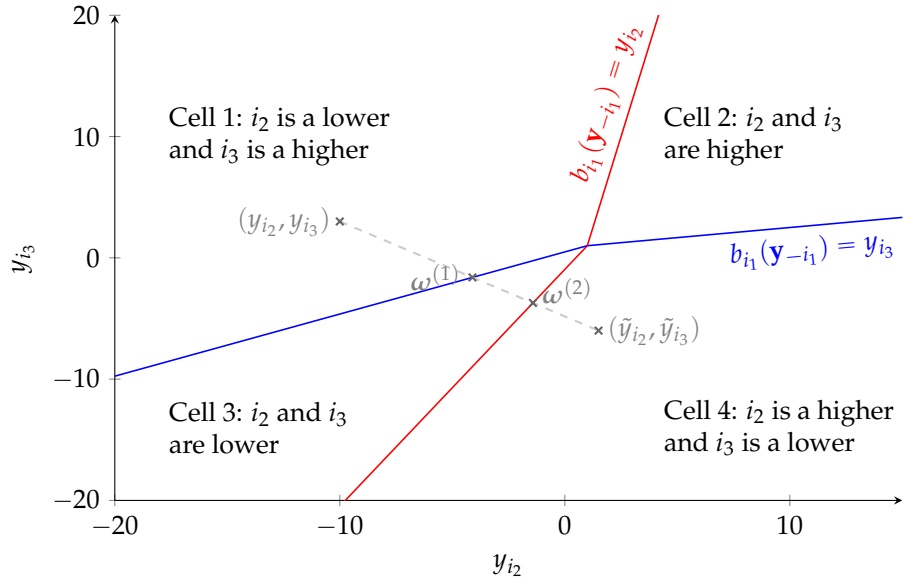

Since $b_i(\mathbf{y}_{-i})$ is not differentiable if and only if $b_i(\mathbf{y}_{-i}) = y_j$, it follows that the subset $\mathcal{D}$ is the union of the facets that separate the cells. When $n_i > 1$, $\mathcal{D}$ is unbounded as it includes at least one half-hyperplane or half-line. Therefore, $b_i(\mathbf{y}_{-i})$ is not differentiable at infinitely many points.

\paragraph{Step 2: {\normalfont We show that 
\texorpdfstring{$\lvert b_i(\mathbf{y}_{-i}) - b_i(\tilde{\mathbf{y}}_{-i})\rvert \leq \mu \lVert \mathbf{y} - \tilde{\mathbf{y}} \rVert_{\infty}$}{TEXT}  if \texorpdfstring{$\mathbf{y}_{-i}$ and $\tilde{\mathbf{y}}_{-i}$}{TEXT} belong to the same cell, for some constant $\mu \in (0,1)$.}}
\medskip

From Lemma \ref{lemma:BRFcontinuous}, we have 
$$b_i(\mathbf{y}_{-i}) = \dfrac{\alpha_i + \beta^l \sum_{j \ne i} \mathbbm{1}\{y_j \le y_i\} g_{ij} y_j + \beta^h \sum_{j \ne i} \mathbbm{1}\{y_j > y_i\} g_{ij} y_j}{1 + \beta^l \sum_{j \ne i} \mathbbm{1}\{y_j \le y_i\} g_{ij} + \beta^h \sum_{j \ne i} \mathbbm{1}\{y_j > y_i\} g_{ij}},$$
where $y_i = b_i(\mathbf{y}_{-i})$. If $\mathbf{y}_{-i}$ and $\tilde{\mathbf{y}}_{-i}$ are from the same cell, then $\mathbbm{1}\{y_j \le y_i\} g_{ij} = \mathbbm{1}\{\tilde y_j \le \tilde y_i\} g_{ij}$ for all $j \ne i$. Consequently,
\begingroup
\allowdisplaybreaks
\begin{align*}
    b_i(\mathbf{y}_{-i}) - b_i(\tilde{\mathbf{y}}_{-i}) &= \dfrac{\beta^l \sum_{j \ne i} \mathbbm{1}\{y_j \le y_i\} g_{ij} (y_j - \tilde y_j) + \beta^h \sum_{j \ne i} \mathbbm{1}\{y_j > y_i\} g_{ij} (y_j - \tilde y_j)}{1 + \beta^l \sum_{j \ne i} \mathbbm{1}\{y_j \le y_i\} g_{ij} + \beta^h \sum_{j \ne i} \mathbbm{1}\{y_j > y_i\} g_{ij}},\\
    \lvert b_i(\mathbf{y}_{-i}) - b_i(\tilde{\mathbf{y}}_{-i}) \rvert &\leq \dfrac{\lvert \beta^l \rvert \sum_{j \ne i} \mathbbm{1}\{y_j \le y_i\} g_{ij} + \lvert \beta^h \rvert \sum_{j \ne i} \mathbbm{1}\{y_j > y_i\} g_{ij}}{\lvert 1 + \beta^l \sum_{j \ne i} \mathbbm{1}\{y_j \le y_i\} g_{ij} + \beta^h \sum_{j \ne i} \mathbbm{1}\{y_j > y_i\} g_{ij} \rvert} \max_{j\ne i} \lvert y_j - \tilde y_j \rvert,\\
    \lvert b_i(\mathbf{y}_{-i}) - b_i(\tilde{\mathbf{y}}_{-i}) \rvert &\leq \dfrac{\lvert \beta^l \rvert g_i^l + \lvert \beta^h \rvert g_i^h}{\lvert 1 + \beta^l g_i^l + \beta^h g_i^h \rvert} \lVert \mathbf{y} - \tilde{\mathbf{y}} \rVert_{\infty} = \dfrac{(\lvert \beta^h \rvert - \lvert \beta^l \rvert) g_i^h + \lvert \beta^l \rvert}{\lvert 1 + (\beta^h  - \beta^l) g_i^h + \beta^l \rvert} \lVert \mathbf{y} - \tilde{\mathbf{y}} \rVert_{\infty}.
\end{align*}
\endgroup

The last equality holds because $g_i^l + g_i^h  = 1$. As $\beta^l > -\frac{1}{2}$ and $\beta^h > -\frac{1}{2}$, then $1 + \beta^l g_i^l + \beta^h g_i^h$ is positive. Moreover, $\frac{(\lvert \beta^h \rvert - \lvert \beta^l \rvert) g_i^h + \lvert \beta^l \rvert}{ 1 + (\beta^h  - \beta^l) g_i^h + \beta^l }$ is monotone in $g_i^h$ and is thus bounded above by either $\frac{\lvert \beta^h \rvert}{1 + \beta^h}$ or $\frac{\lvert \beta^l \rvert}{1 + \beta^l}$. As a result, for $\mu = \displaystyle\max_{k \in \{l,h\}} \textstyle\frac{\lvert \beta^k \rvert}{1 + \beta^k}$, we have $\lvert b_i(\mathbf{y}_{-i}) - b_i(\tilde{\mathbf{y}}_{-i}) \rvert \leq \mu \lVert \mathbf{y} - \tilde{\mathbf{y}} \rVert_{\infty}$, where $\mu < 1$ because $\beta^l > -\frac{1}{2}$ and $\beta^h > -\frac{1}{2}$.

Importantly, the condition $\lvert b_i(\mathbf{y}_{-i}) - b_i(\tilde{\mathbf{y}}_{-i}) \rvert \leq \mu \lVert \mathbf{y} - \tilde{\mathbf{y}} \rVert_{\infty}$ holds even if $\mathbf{y}_{-i}$ and $\tilde{\mathbf{y}}_{-i}$ belong to facets delimiting a cell (e.g., $\boldsymbol{\omega}^{(1)}$ and $\boldsymbol{\omega}^{(2)}$ in Figure~\ref{fig_equil}). Indeed, such  $\mathbf{y}_{-i}$ and $\tilde{\mathbf{y}}_{-i}$ can be approached as limits of sequences in the interior of the cell. As $b_i$ is continuous (Lemma \ref{lemma:BRFcontinuous}), the condition holds for the sequences and their limits.

\bigskip
\paragraph{Step 3: {\normalfont We show that \texorpdfstring{$\lvert b_i(\mathbf{y}_{-i}) - b_i(\tilde{\mathbf{y}}_{-i})\rvert \leq \mu \lVert \mathbf{y} - \tilde{\mathbf{y}} \rVert_{\infty}$}{TEXT} for any \texorpdfstring{$\mathbf{y}_{-i}$ and $\tilde{\mathbf{y}}_{-i}$}{TEXT}.}}
\medskip

Given the result of Step 2, it is sufficient to show that $\lvert b_i(\mathbf{y}_{-i}) - b_i(\tilde{\mathbf{y}}_{-i})\rvert \leq \mu \lVert \mathbf{y} - \tilde{\mathbf{y}} \rVert_{\infty}$ when $\mathbf{y}_{-i}$ and $\tilde{\mathbf{y}}_{-i}$ are from different cells. The challenge here is that $\mathbbm{1}\{y_j \le y_i\} g_{ij}$ may differ from $\mathbbm{1}\{\tilde y_j \le \tilde y_i\} g_{ij}$ and, therefore, $b_i(\mathbf{y}_{-i}) - b_i(\tilde{\mathbf{y}}_{-i})$ cannot be easily simplified.

For any $\mathbf{y}_{-i}$ and $\tilde{\mathbf{y}}_{-i}$ from different cells, let $\boldsymbol{w}^{(1)}, \dots, \boldsymbol{w}^{(R)}$ denote the intersections of the segment $[\mathbf{y}_{-i},\tilde{\mathbf{y}}_{-i}]$ with cell facets, where $R$ is the number of intersections. Note that each segment $[\boldsymbol{w}^{(r)}, \boldsymbol{w}^{(r+1)}]$, for $r \leq R$, belongs to a single cell because $\boldsymbol{w}^{(r)}$ and $\boldsymbol{w}^{(r+1)}$ are consecutive intersections and thus lie on facets delimiting the same cell (see Figure~\ref{fig_equil}). Define $\boldsymbol{w}^{(0)} = \mathbf{y}_{-i}$ and $\boldsymbol{w}^{(R+1)} = \tilde{\mathbf{y}}_{-i}$. We use the telescoping decomposition
$b_i(\mathbf{y}_{-i}) - b_i(\tilde{\mathbf{y}}_{-i}) = \sum_{r = 0}^{R} \big(b_i(\boldsymbol{w}^{(r)}) - b_i(\boldsymbol{w}^{(r + 1)})\big)$,
which implies that:
\begingroup
\allowdisplaybreaks
\begin{align}
    \lvert b_i(\mathbf{y}_{-i}) - b_i(\tilde{\mathbf{y}}_{-i}) \rvert &\leq \sum_{r = 0}^{R} \lvert b_i(\boldsymbol{w}^{(r)}) - b_i(\boldsymbol{w}^{(r + 1)})\rvert \nonumber,\\
    \lvert b_i(\mathbf{y}_{-i}) - b_i(\tilde{\mathbf{y}}_{-i}) \rvert &\leq \mu\sum_{r = 0}^{R}  \lVert \boldsymbol{w}^{(r)} - \boldsymbol{w}^{(r+1)}\rVert_{\infty}.\label{eq:telescoping}
\end{align}
\endgroup
The first inequality holds by the triangular inequality, while the second inequality holds by Step~2 because $\boldsymbol{w}^{(r)}$ and $\boldsymbol{w}^{(r+1)}$ belong to the same cell (including facets).

Furthermore, since all $\boldsymbol{w}^{(r)}$ lie on the segment $[\mathbf{y}_{-i}, \tilde{\mathbf{y}}_{-i}]$, we have $\boldsymbol{w}^{(r)} - \boldsymbol{w}^{(r+1)} = t_r (\mathbf{y}_{-i} - \tilde{\mathbf{y}}_{-i})$, with $t_r > 0$ and $\sum_{r = 0}^{R} t_r = 1$. Substituting this into \eqref{eq:telescoping} yields:
$$
    \lvert b_i(\mathbf{y}_{-i}) - b_i(\tilde{\mathbf{y}}_{-i}) \rvert \leq \mu \lVert \mathbf{y}_{-i} - \tilde{\mathbf{y}}_{-i} \rVert_{\infty} \sum_{r = 0}^{R} t_r \leq \mu \lVert \mathbf{y} - \tilde{\mathbf{y}} \rVert_{\infty}.
$$
As a result, the mapping $\mathbf{b}$ is a contraction, and the Nash equilibrium is unique.

\subsection{Proof of Proposition \ref{propo:socialmultiplier}} \label{append:proof:socialmultiplier}
The equilibrium outcome of individual $i$ is given by Equation~\eqref{eq:yi1}:
\begin{equation*}
   y_i = \dfrac{\alpha_i + \beta^l \bar{y}_i^l + \beta^h \bar{y}_i^h}{1 + \beta^l g_i^l + \beta^h g_i^h}.
\end{equation*}
Assume a uniform increase in productivity by $\bar \alpha$ such that the new productivity level is $\tilde \alpha_i=\alpha_i + \bar \alpha$ for all $i$. This increase likely affects individuals' best responses and the equilibrium. Let $\tilde y_i$, $\tilde{\bar{y}}_i^l$, $\tilde{\bar{y}}_i^h$, $\tilde g_i^l$, and $\tilde g_i^h$ denote the counterparts of $y_i$, $\bar{y}_i^l$, $\bar{y}_i^h$, $g_i^l$, and $g_i^h$ under the new equilibrium. Assume that $\tilde y_i=y_i+ \bar \alpha$ for all $i$, so that $\tilde{\bar{y}}_i^l=\bar{y}_i^l+g_i^l \bar \alpha$, $\tilde{\bar{y}}_i^h=\bar{y}_i^h+g_i^h \bar \alpha$, $\tilde g_i^l=g_i^l$, and $\tilde g_i^h=g_i^h$. It follows that:
\begin{align}
      \dfrac{\tilde \alpha_i + \beta^l \tilde{\bar{y}}_i^l + \beta^h \tilde{\bar{y}}_i^h}{1 + \beta^l \tilde g_i^l + \beta^h \tilde g_i^h}
      &= \dfrac{\alpha_i + \bar \alpha + \beta^l (\bar{y}_i^l + g_i^l \bar \alpha) + \beta^h (\bar{y}_i^h + g_i^h \bar \alpha)}{1 + \beta^l g_i^l + \beta^h g_i^h}\nonumber \\
      &= \dfrac{\alpha_i + \beta^l \bar{y}_i^l + \beta^h \bar{y}_i^h}{1 + \beta^l g_i^l + \beta^h g_i^h}
      + \dfrac{\bar \alpha(1 + \beta^l g_i^l + \beta^h g_i^h)}{1 + \beta^l g_i^l + \beta^h g_i^h} \nonumber\\
      &= y_i + \bar \alpha,\nonumber \\
      &= \tilde y_i.\label{eq:yitilde}
\end{align}
Equation~\eqref{eq:yitilde} is analogous to Equation~\eqref{eq:yi1}, but written for $\tilde \alpha$ and $\tilde y_i$. Because the contraction mapping theorem guarantees that Equation~\eqref{eq:yitilde} has a unique solution, and since the equilibrium associated with the new productivity level must satisfy Equation~\eqref{eq:yitilde}, it follows that $\tilde y_i = y_i + \bar \alpha$ is the outcome of the new equilibrium.

\subsection{Proof of Proposition \ref{prop:ident}} \label{append:proof:Ident}
From Equation~\eqref{eq:yred}, we have:
$$\mathbb E_m (\mathbf y_m) = \tilde c \mathbf 1_{n_m} +  \theta_1\mathbf G_m\mathbb E_m (\mathbf y_m) + \theta_2 \mathbb E_m (\check{\mathbf y}_m) +  \mathbf X_m\boldsymbol\theta_3 + \mathbf G_m \mathbf X_m\boldsymbol\theta_4,$$ which implies
$$(\mathbf I_{n_m} - \theta_1\mathbf G_m)\mathbb E_m (\mathbf y_m) = \tilde c \mathbf 1_{n_m} + \theta_2 \mathbb E_m (\check{\mathbf y}_m) +  (\mathbf I_{n_m} - \theta_1\mathbf G_m)\mathbf X_m\boldsymbol\theta_3 + \mathbf G_m \mathbf X_m(\theta_1\boldsymbol\theta_3 + \boldsymbol\theta_4).$$
Under Assumption \ref{assumption:minbeta}, $\lvert \theta_1 \rvert < 1$, thus $\mathbf I_{n_m} - \theta_1\mathbf G_m$ is invertible. It thus follows that:
\begin{equation}
    \mathbb E_m (\mathbf y_m) = (\mathbf I_{n_m} - \theta_1\mathbf G_m)^{-1}\left(\tilde c \mathbf 1_{n_m} + \theta_2 \mathbb E_m (\check{\mathbf y}_m) +    \mathbf G_m \mathbf X_m(\theta_1\boldsymbol\theta_3 + \boldsymbol\theta_4)\right) + \mathbf X_m\boldsymbol\theta_3.\label{eq:AppendEys}
\end{equation}

We first show that the reduced-form parameter vector
$\boldsymbol{\theta} = (\theta_1, \, \theta_2, \, \tilde c, \, \boldsymbol\theta_3^{\prime}, \, \boldsymbol\theta_4^{\prime})^{\prime}$
is identified. Suppose that two parameter vectors, $\boldsymbol{\theta}$ and
$\dot{\boldsymbol{\theta}} = (\dot\theta_1, \, \dot\theta_2, \, \dot{\tilde c}, \, \dot{\boldsymbol\theta}_3^{\prime}, \, \dot{\boldsymbol\theta}_4^{\prime})^{\prime}$,
generate the same distribution of the data, that is, the same $\mathbb E_m(\mathbf y_m)$ and $\mathbb E_m(\check{\mathbf y}_m)$. Following \citet{rothenberg1971identification}, we establish identification by showing that this implies
$\boldsymbol{\theta} = \dot{\boldsymbol{\theta}}$. 

From Equation~\eqref{eq:AppendEys}, we have:

\begin{align}
\begin{split}
    &(\mathbf I_{n_m} - \theta_1\mathbf G_m)^{-1}\left(\tilde c \mathbf 1_{n_m} + \theta_2 \mathbb E_m (\check{\mathbf y}_m) +  \mathbf G_m \mathbf X_m(\theta_1\boldsymbol\theta_3 + \boldsymbol\theta_4)\right) + \mathbf X_m\boldsymbol\theta_3 = \\
    &\quad (\mathbf I_{n_m} - \dot\theta_1\mathbf G_m)^{-1}(\dot{\tilde c} \mathbf 1_{n_m} + \dot\theta_2 \mathbb E_m (\check{\mathbf y}_m) +    \mathbf G_m \mathbf X_m(\dot\theta_1\dot{\boldsymbol\theta}_3 + \dot{\boldsymbol\theta}_4)) + \mathbf X_m\dot{\boldsymbol\theta}_3
\end{split} \label{eq:observational_equivalence}
\end{align}

By premultiplying each term of Equation \eqref{eq:observational_equivalence} by $(\mathbf I_{n_m} - \theta_1\mathbf G_m)(\mathbf I_{n_m} - \dot\theta_1\mathbf G_m)$, which is also equal to $(\mathbf I_{n_m} - \dot\theta_1\mathbf G_m)(\mathbf I_{n_m} - \theta_1\mathbf G_m)$, we obtain:
\begin{align}
\begin{split}
    &(\mathbf I_{n_m} - \dot\theta_1\mathbf G_m)\big(\tilde c \mathbf 1_{n_m} + \theta_2 \mathbb E_m (\check{\mathbf y}_m) +  \mathbf G_m \mathbf X_m(\theta_1\boldsymbol\theta_3 + \boldsymbol\theta_4) + (\mathbf I_{n_m} - \theta_1\mathbf G_m) \mathbf X_m\boldsymbol\theta_3\big) = \\
    &\quad (\mathbf I_{n_m} - \theta_1\mathbf G_m)\big(\dot{\tilde c} \mathbf 1_{n_m} + \dot\theta_2 \mathbb E_m (\check{\mathbf y}_m) +  \mathbf G_m \mathbf X_m(\dot\theta_1\dot{\boldsymbol\theta}_3 + \dot{\boldsymbol\theta}_4) + (\mathbf I_{n_m} - \dot \theta_1\mathbf G_m)\mathbf X_m\dot{\boldsymbol\theta}_3\big). 
\end{split}\nonumber 
\end{align}
This implies that:
\allowdisplaybreaks
\begin{align}
\begin{split}
    &(\mathbf I_{n_m} - \dot\theta_1\mathbf G_m)\big(\tilde c \mathbf 1_{n_m} + \theta_2 \mathbb E_m (\check{\mathbf y}_m) + \mathbf X_m\boldsymbol\theta_3 +  \mathbf G_m \mathbf X_m\boldsymbol\theta_4 \big) = \\
    &\quad (\mathbf I_{n_m} - \theta_1\mathbf G_m)\big(\dot{\tilde c} \mathbf 1_{n_m} + \dot\theta_2 \mathbb E_m (\check{\mathbf y}_m)  + \mathbf X_m\dot{\boldsymbol\theta}_3 +  \mathbf G_m \mathbf X_m\dot{\boldsymbol\theta}_4\big). 
\end{split}\nonumber 
\end{align}
As we assume that there are no isolated individuals, $\mathbf G_m^k \mathbf 1_{n_m} = \mathbf 1_{n_m}$. We thus obtain:
\begin{align}
\begin{split}
    &(\tilde c -  \dot{\tilde c} - \theta_1\dot{\tilde c} + \dot\theta_1\tilde c)\mathbf 1_{n_m} + (\theta_2 - \dot\theta_2)\mathbb E_m (\check{\mathbf y}_m) + (\theta_1\dot\theta_2 - \dot\theta_1\theta_2)\mathbf G_m \mathbb E_m (\check{\mathbf y}_m) + \\
    &\quad \mathbf X_m(\boldsymbol\theta_3  - \dot{\boldsymbol\theta}_3) +  \mathbf G_m \mathbf X_m(\theta_1 \dot{\boldsymbol\theta}_3 - \dot\theta_1 \boldsymbol\theta_3 +\boldsymbol\theta_4 - \dot{\boldsymbol\theta}_4) + \mathbf G_m^2\mathbf X_m (\theta_1 \dot{\boldsymbol\theta}_4-\dot \theta_1\boldsymbol\theta_4) = \mathbf 0. 
\end{split}\label{eq:linear_combination}
\end{align}

Since $\frac{1}{M}\sum_{m = 1}^M \mathbf A_m^{\prime}\mathbf A_m$ is a full-rank matrix (Assumption \ref{assumption:ident:fullrank}), it follows that all coefficients of the linear combination in Equation \eqref{eq:linear_combination} are equal to zero. 

Setting the coefficients of $\mathbb E_m (\check{\mathbf y}_m)$ and $\mathbf X_m$ to zero implies that $\theta_2 = \dot\theta_2$ and $\boldsymbol\theta_3 = \dot{\boldsymbol\theta}_3$. Since $\boldsymbol\theta_3 = \dot{\boldsymbol\theta}_3$, setting the coefficients of $\mathbf G_m \mathbf X_m$ and $\mathbf G_m^2 \mathbf X_m$ to zero implies:
\begin{align}
    &\theta_1 \boldsymbol\theta_3 +\boldsymbol\theta_4 = \dot\theta_1 \boldsymbol\theta_3 + \dot{\boldsymbol\theta}_4 \quad \text{and} \label{equation:cond1}\\
    &\theta_1 \dot{\boldsymbol\theta}_4 =\dot \theta_1\boldsymbol\theta_4.\label{equation:cond2}
\end{align}

By multiplying Equation \eqref{equation:cond1} by $\theta_1$, we obtain $\theta_1(\theta_1\boldsymbol \theta_3 + \boldsymbol \theta_4) = \theta_1\dot \theta_1\boldsymbol \theta_3 + \theta_1\dot{\boldsymbol \theta_4}$. Substituting $\theta_1\dot{\boldsymbol \theta_4}$ with $\dot \theta_1\boldsymbol \theta_4$ (from Equation \eqref{equation:cond2}) yields $\theta_1(\theta_1\boldsymbol \theta_3 + \boldsymbol \theta_4) = \dot \theta_1(\theta_1\boldsymbol \theta_3 + \boldsymbol \theta_4)$. Since $\theta_1\boldsymbol{\theta}_3 + \boldsymbol{\theta}_4 \ne \mathbf 0$ (Assumption \ref{assumption:ident:nonzero}), it follows that $\theta_1 = \dot \theta_1$. Substituting back into Equation \eqref{equation:cond1} implies that $\boldsymbol\theta_4 = \dot{\boldsymbol\theta}_4$. 

Furthermore, setting the coefficient on $\mathbf 1_{n_m}$ equal to zero implies that $\tilde c = \dot{\tilde c}$. As a result, $\boldsymbol{\theta} = \dot{\boldsymbol{\theta}}$, and the reduced-form parameters are point identified.

The identification of the structural parameters then follows directly. In particular, identification of $\theta_1 = \frac{\beta^l}{1 + \beta^l}$ and $\tilde c = \frac{c}{1 + \beta^l}$ implies that $\beta^l$ and $c$ are identified. The identification of $\theta_2 = \frac{\beta^h - \beta^l}{1 + \beta^l}$ together with $\beta^l$ implies that $\beta^h$ is identified. Finally, identification of $\boldsymbol{\theta}_3 = \frac{\boldsymbol{\gamma}_1}{1 + \beta^l}$ and $\boldsymbol{\theta}_4 = \frac{\boldsymbol{\gamma}_2}{1 + \beta^l}$ implies that $\boldsymbol{\gamma}_1$ and $\boldsymbol{\gamma}_2$ are identified.

\subsection{Proof of Proposition \ref{prop:consistency}}\label{append:prop:consistency}
\noindent Since the model is just identified, the IV estimator can be written as
\allowdisplaybreaks
\begin{align*}
    &\textstyle\hat{\boldsymbol\theta} = \left(\frac{1}{M} \sum_{m = 1}^M \hat{\mathbf Z}^{\prime}_m \mathbf V_m\right)^{-1} \left(\frac{1}{M} \sum_{m = 1}^M \hat{\mathbf Z}_m^{\prime} \mathbf y_m\right), \nonumber\\
    &\textstyle\sqrt{M}(\hat{\boldsymbol\theta} - \boldsymbol\theta_0) = \frac{1}{1 + \beta^l_0} \left(\frac{1}{M} \sum_{m = 1}^M \hat{\mathbf Z}^{\prime}_m \mathbf V_m\right)^{-1} \left(\frac{1}{\sqrt{M}} \sum_{m = 1}^M \hat{\mathbf Z}_m^{\prime} \boldsymbol \varepsilon_m\right),
\end{align*}
where $\beta^l_0$ is the true value of $\beta^l$. The second equation follows from substituting $\mathbf y_m$ with its expression given in Equation~\eqref{eq:yred}, evaluated at the true parameter $\boldsymbol{\theta}_0$.

Assumption~\ref{assumption:iv} implies that $\plim \frac{1}{M} \sum_{m = 1}^M \hat{\mathbf Z}^{\prime}_m \mathbf V_m$ is deterministic, where $\plim$ denotes the probability limit as $M$ goes to infinity. It is thus sufficient to show that $\frac{1}{\sqrt{M}} \sum_{m = 1}^M \hat{\mathbf Z}_m^{\prime} \boldsymbol \varepsilon_m$ converges in distribution to a centered normal distribution. 

However, the central limit theorem does not apply directly to $\frac{1}{\sqrt{M}} \sum_{m = 1}^M \hat{\mathbf Z}_m^{\prime} \boldsymbol \varepsilon_m$ because $\hat{\mathbf Z}_m^{\prime} \boldsymbol \varepsilon_m$ is not independent across $m$. To address this issue, we decompose $\frac{1}{\sqrt{M}} \sum_{m = 1}^M \hat{\mathbf Z}_m^{\prime} \boldsymbol \varepsilon_m$ as follows:
\begin{align}
    \textstyle\frac{1}{\sqrt{M}} \sum_{m = 1}^M \hat{\mathbf Z}_m^{\prime} \boldsymbol \varepsilon_m
    &= \textstyle \frac{1}{\sqrt{M}} \sum_{m = 1}^M \mathbf Z_m^{\prime} \boldsymbol \varepsilon_m 
    + \frac{1}{\sqrt{M}} \sum_{m = 1}^M \left(\hat{\mathbf Z}_m - \mathbf Z_m\right)^{\prime} \boldsymbol \varepsilon_m. \nonumber
\end{align}
The first term is a sum of i.i.d.\ variables scaled by $\sqrt{M}$. Thus, it converges to a centered normal distribution. We show that the second term vanishes asymptotically. 

We randomly split the networks into $L$ folds $\mathcal F_1, \dots, \mathcal F_L$, where $L < \infty$. Define $\mathbf u_l = \frac{1}{\sqrt M} \sum_{m \in \mathcal{F}_l} \left(\hat{\mathbf Z}_m - \mathbf Z_m\right)^{\prime} \boldsymbol \varepsilon_m$, so that $\frac{1}{\sqrt{M}} \sum_{m = 1}^M \left(\hat{\mathbf Z}_m - \mathbf Z_m\right)^{\prime} \boldsymbol \varepsilon_m = \sum_{l = 1}^L \mathbf u_l$. The use of cross-fitting ensures that, for any $m, m^{\prime} \in \mathcal F_l$, $\hat{\mathbf Z}_{m^{\prime}}$ is independent of $\boldsymbol \varepsilon_m$ conditional on $\mathbf X_m$ and $\mathbf G_m$, implying that $\mathbb E(\mathbf u_l) = \mathbf 0$. Moreover, $\left(\hat{\mathbf Z}_m - \mathbf Z_m\right)^{\prime} \boldsymbol \varepsilon_m$ is independent across $m \in \mathcal F_l$, conditional on $\mathbf X_m$ and $\mathbf G_m$. Thus,
\begin{align*}
    \textstyle
    \mathbb V(\mathbf u_l) =
    \frac{1}{M} \sum_{m \in \mathcal F_l}
    \mathbb E \Big[
    (\hat{\mathbf Z}_m - \mathbf Z_m)^{\prime}
    \mathbb E_m(\boldsymbol \varepsilon_m \boldsymbol \varepsilon_m^{\prime})
    (\hat{\mathbf Z}_m - \mathbf Z_m)
    \Big],
\end{align*}
where $\mathbb V$ denotes the variance.

Since $\hat{\mathbf Z}_m - \mathbf Z_m = o_p(1)$ (Assumption~\ref{assumption:iv:ml}), then $\mathbb V(\mathbf u_l)$ is asymptotically zero. Consequently, $\mathbf u_l = o_p(1)$ by mean square convergence. As $L < \infty$, it follows that $\frac{1}{\sqrt{M}} \sum_{m = 1}^M \left(\hat{\mathbf Z}_m - \mathbf Z_m\right)^{\prime} \boldsymbol \varepsilon_m = o_p(1)$. Thus,
$$\textstyle\sqrt{M}(\hat{\boldsymbol\theta} - \boldsymbol\theta_0) = \frac{1}{1 + \beta^l_0} \left(\frac{1}{M} \sum_{m = 1}^M \mathbf Z^{\prime}_m \mathbf V_m\right)^{-1} \left(\frac{1}{\sqrt{M}} \sum_{m = 1}^M \mathbf Z_m^{\prime} \boldsymbol \varepsilon_m\right) + o_p(1).$$

As a result, $\hat{\boldsymbol\theta}$ is a consistent estimator of $\boldsymbol\theta_0$, and $\sqrt{M}(\hat{\boldsymbol\theta} - \boldsymbol{\theta}_0) \overset{d}{\to} N(\mathbf 0, \boldsymbol\Sigma)$. The asymptotic variance is given by
$$\boldsymbol\Sigma = \frac{1}{(1 + \beta^l_0)^2} \boldsymbol\Omega \left(\plim \frac{1}{M} \sum_{m = 1}^M \mathbf Z^{\prime}_m \mathbb E_m(\boldsymbol \varepsilon_m \boldsymbol \varepsilon_m^{\prime}) \mathbf Z_m \right) \boldsymbol\Omega^{\prime},$$
where $\boldsymbol\Omega = \left(\plim \frac{1}{M} \sum_{m = 1}^M \mathbf Z^{\prime}_m \mathbf V_m\right)^{-1}$.

\subsection{Proof of Proposition \ref{prop:biasSymmetric}}\label{append:prop:biasSymmetric}
The moment condition for the estimation of the symmetric model is given by: 
\begin{equation}\label{eq:momentsym}
    \mathbb E\big(\tilde{\mathbf{Z}}_m^{\prime} (\mathbf y_{m} - \delta_1 \bar{\mathbf y}_{m} - \mathbf{R}_{m} \boldsymbol{\delta}_2)\big) = \mathbf{0},
\end{equation}
where $\mathbf R_m = (\mathbf 1_{n_m}, \, \mathbf X_m, \, \mathbf G_m \mathbf X_m)$ and $\delta_1 = \dfrac{\beta}{1 + \beta}$.

We first prove that $\dfrac{\beta}{1+\beta} = \dfrac{\beta^l}{1+\beta^l} + \delta^h_1 \dfrac{\beta^h - \beta^l}{1+\beta^l}$. The expression of the outcome in the asymmetric model (Equation \eqref{eq:yi2}) can be written as:
\begin{equation}\label{eq:yi2:matrix}
    \mathbf y_m = \dfrac{\beta^l \bar{\mathbf{y}}_m + (\beta^h - \beta^l)\check{\mathbf{y}}_m^h + \mathbf R_m \boldsymbol \gamma + \boldsymbol \varepsilon_m}{1 + \beta^l},
\end{equation}
where $\boldsymbol \gamma = (c,\, \boldsymbol{\gamma}_1^{\prime}, \, \boldsymbol{\gamma}_2^{\prime})^{\prime}$. Substituting $\mathbf y_m$ in Equation \eqref{eq:momentsym} yields:
\begin{equation}\label{eq:proof5}
    \bigg(\dfrac{\beta^l}{1+\beta^l} - \delta_1 \bigg)\mathbb E(\tilde{\mathbf{Z}}_m^{\prime} \bar{\mathbf y}_{m}) 
    + \dfrac{\beta^h - \beta^l}{1+\beta^l} \mathbb E(\tilde{\mathbf{Z}}_m^{\prime} \check{\mathbf{y}}_m^h) + \mathbb E(\tilde{\mathbf{Z}}_m^{\prime} \mathbf{R}_{m})\bigg( \dfrac{\boldsymbol \gamma}{1+\beta^l} - \boldsymbol \delta_2 \bigg) = \mathbf{0},
\end{equation}
because $\mathbb E\big(\tilde{\mathbf{Z}}_m^{\prime} \boldsymbol \varepsilon_m\big)=0$ by the exogeneity assumption in Assumption \ref{assumption:data:exogeneity}.

Next, we use the moment condition $\mathbb E \left( \tilde{\mathbf{Z}}_m^{\prime} \big( \check{\mathbf y}_m^h - \delta^h_1 \bar{\mathbf y}_m - \mathbf R_m \boldsymbol{\delta}^h_2 \big) \right) = \mathbf{0}$ of the auxiliary IV \eqref{eq:auxiliaryIV}, which implies $\mathbb E \big( \tilde{\mathbf{Z}}_m^{\prime} \check{\mathbf y}_m^h \big) = \delta^h_1 \mathbb E \big( \tilde{\mathbf{Z}}_m^{\prime} \bar{\mathbf y}_m \big) - \mathbb E \big( \tilde{\mathbf{Z}}_m^{\prime} \mathbf R_m \big)\boldsymbol{\delta}^h_2$. Substituting $\mathbb E \big( \tilde{\mathbf{Z}}_m^{\prime} \check{\mathbf y}_m^h \big)$ into Equation \eqref{eq:proof5} implies that:
\begin{equation*}
    \left(\dfrac{\beta^l}{1+\beta^l} +\dfrac{\beta^h-\beta^l}{1+\beta^l}\delta^h_1-\delta_1\right)\mathbb E(\tilde{\mathbf{Z}}_m^{\prime} \bar{\mathbf y}_{m}) + \mathbb E(\tilde{\mathbf{Z}}_m^{\prime} \mathbf{R}_{m})\left( \dfrac{\boldsymbol \gamma}{1+\beta^l}+\dfrac{\beta^h-\beta^l}{1+\beta^l}\boldsymbol \delta_2^h-\boldsymbol\delta_2\right)
    = \mathbf{0}.
\end{equation*}

As $M$ tends to infinity, Assumption~\ref{assumption:iv:fullrank} implies that $\mathbb E(\tilde{\mathbf{Z}}_m^{\prime}\mathbf V_m)$ is full rank. Thus, $\mathbb E(\tilde{\mathbf{Z}}_m^{\prime} \mathbf{R}_{m})$ and $\mathbb E(\tilde{\mathbf{Z}}_m^{\prime} \bar{\mathbf y}_{m})$ are linearly independent because $\tilde{\mathbf{Z}}_m$ is composed of columns of $\hat{\mathbf{Z}}_m$, and $\mathbf{R}_{m}$ and $\bar{\mathbf y}_{m}$ are also composed of columns of $\mathbf V_m$. As a result, $\dfrac{\beta^l}{1+\beta^l} +\dfrac{\beta^h-\beta^l}{1+\beta^l}\delta^h_1-\delta_1 = 0$, and the result follows because $\delta_1 = \dfrac{\beta}{1 + \beta}$.

To prove that $\dfrac{\beta}{1+\beta} = \dfrac{\beta^h}{1+\beta^h} + \delta^l_1 \dfrac{\beta^l - \beta^h}{1+\beta^h}$, we proceed analogously, except that we substitute $\mathbf y_m$ in Equation \eqref{eq:momentsym} with the following alternative representation:
\begin{equation*}
     \mathbf y_m = \dfrac{\beta^h \bar{\mathbf{y}}_m + (\beta^l - \beta^h)\check{\mathbf{y}}_m^l + \mathbf R_m \boldsymbol \gamma + \boldsymbol \varepsilon_m}{1 + \beta^h},
\end{equation*}
As for Equation \eqref{eq:yi2:matrix}, this expression is derived from Equation \eqref{eq:yi1} by substituting $g_{m,i}^h = 1 - g_{m,i}^l$ and $\bar{y}_{m,i}^h = \bar{y}_{m,i} - \bar{y}_{m,i}^l$. 

Furthermore, in the resulting equation, which is the counterpart of \eqref{eq:proof5}, we substitute $\mathbb E \big( \tilde{\mathbf{Z}}_m^{\prime} \check{\mathbf y}_m^l \big) = \delta^l_1 \mathbb E \big( \tilde{\mathbf{Z}}_m^{\prime} \bar{\mathbf y}_m \big) - \mathbb E \big( \tilde{\mathbf{Z}}_m^{\prime} \mathbf R_m \big)\boldsymbol{\delta}^l_2$, which is derived from the second moment condition of the auxiliary IV \eqref{eq:auxiliaryIV}.

\clearpage
{\linespread{1}
\fontsize{11}{14}\selectfont
\bibliography{references}

@article{rothenberg1971identification,
  title={Identification in parametric models},
  author={Rothenberg, Thomas J},
  journal={Econometrica: Journal of the Econometric Society},
  pages={577--591},
  year={1971},
  publisher={JSTOR}
}

@article{beaman2021can,
  title={Can network theory-based targeting increase technology adoption?},
  author={Beaman, Lori and BenYishay, Ariel and Magruder, Jeremy and Mobarak, Ahmed Mushfiq},
  journal={American Economic Review},
  volume={111},
  number={6},
  pages={1918--1943},
  year={2021},
  publisher={American Economic Association 2014 Broadway, Suite 305, Nashville, TN 37203}
}

@article{de2020two,
  title={Two-way fixed effects estimators with heterogeneous treatment effects},
  author={De Chaisemartin, Cl{\'e}ment and d’Haultfoeuille, Xavier},
  journal={American Economic Review},
  volume={110},
  number={9},
  pages={2964--2996},
  year={2020},
  publisher={American Economic Association 2014 Broadway, Suite 305, Nashville, TN 37203}
}

@article{griffith2022name,
  title={Name your friends, but only five? the importance of censoring in peer effects estimates using social network data},
  author={Griffith, Alan},
  journal={Journal of Labor Economics},
  volume={40},
  number={4},
  pages={779--805},
  year={2022},
  publisher={The University of Chicago Press Chicago, IL}
}

@article{ball2013friendship,
  title={Friendship networks and social status},
  author={Ball, Brian and Newman, Mark EJ},
  journal={Network Science},
  volume={1},
  number={1},
  pages={16--30},
  year={2013},
  publisher={Cambridge University Press}
}

@article{luttmer2005neighbors,
  title={Neighbors as negatives: Relative earnings and well-being},
  author={Luttmer, Erzo FP},
  journal={The Quarterly Journal of Economics},
  volume={120},
  number={3},
  pages={963--1002},
  year={2005},
  publisher={MIT Press}
}

@article{carrell2013natural,
  title={From natural variation to optimal policy? The importance of endogenous peer group formation},
  author={Carrell, Scott E and Sacerdote, Bruce I and West, James E},
  journal={Econometrica},
  volume={81},
  number={3},
  pages={855--882},
  year={2013},
  publisher={Wiley Online Library}
}

@article{sacerdote2014experimental,
  title={Experimental and quasi-experimental analysis of peer effects: two steps forward?},
  author={Sacerdote, Bruce},
  journal={Annu. Rev. Econ.},
  volume={6},
  number={1},
  pages={253--272},
  year={2014},
  publisher={Annual Reviews}
}

@article{viviano2025policy,
  title={Policy targeting under network interference},
  author={Viviano, Davide},
  journal={Review of Economic Studies},
  volume={92},
  number={2},
  pages={1257--1292},
  year={2025},
  publisher={Oxford University Press UK}
}

@article{kitagawa2018should,
  title={Who should be treated? empirical welfare maximization methods for treatment choice},
  author={Kitagawa, Toru and Tetenov, Aleksey},
  journal={Econometrica},
  volume={86},
  number={2},
  pages={591--616},
  year={2018},
  publisher={Wiley Online Library}
}

@article{comola2025heterogeneous,
  title={Heterogeneous peer effects and gender-based interventions for teenage obesity},
  author={Comola, Margherita and Dieye, Rokhaya and Fortin, Bernard},
  journal={Journal of Health Economics},
  volume={102},
  pages={103023},
  year={2025},
  publisher={Elsevier}
}

@article{feld2017understanding,
  title={Understanding peer effects: On the nature, estimation, and channels of peer effects},
  author={Feld, Jan and Z{\"o}litz, Ulf},
  journal={Journal of Labor Economics},
  volume={35},
  number={2},
  pages={387--428},
  year={2017},
  publisher={University of Chicago Press Chicago, IL}
}

@article{banerjee2019using,
  title={Using gossips to spread information: Theory and evidence from two randomized controlled trials},
  author={Banerjee, Abhijit and Chandrasekhar, Arun G and Duflo, Esther and Jackson, Matthew O},
  journal={The Review of Economic Studies},
  volume={86},
  number={6},
  pages={2453--2490},
  year={2019},
  publisher={Oxford University Press}
}

@article{akerlof1997,
	title = {Social Distance and Social Decisions},
	volume = {65},
	doi = {10.2307/2171877},
	language = {en},
	number = {95},
		journal = {Econometrica},
	author = {Akerlof, George},
	year = {1997},
	pages = {1005},
}

@article{Bramoulle2009,
	title = {Identification of peer effects through social networks},
	volume = {150},
	journal = {journal of Econometrics},
	author = {Bramoullé, Y. and Djebbari, H. and Fortin, B.},
	year = {2009},
	pages = {41--55},
}

@article{ballester2006,
  title={Who's who in networks. Wanted: The key player},
  author={Ballester, Coralio and Calv{\'o}-Armengol, Antoni and Zenou, Yves},
  journal={Econometrica},
  volume={74},
  number={5},
  pages={1403--1417},
  year={2006},
  publisher={Wiley Online Library}
}

@article{Manski1993,
	title = {Identification of endogenous social effects: The reflection problem},
	volume = {60},
	journal = {Review of Economic Studies},
	author = {Manski, F. C.},
	year = {1993},
	pages = {531--542},
}

@inproceedings{mackey2018orthogonal,
  title={Orthogonal machine learning: Power and limitations},
  author={Mackey, Lester and Syrgkanis, Vasilis and Zadik, Ilias},
  booktitle={International Conference on Machine Learning},
  pages={3375--3383},
  year={2018},
  organization={PMLR}
}

@article{houndetoungan2026count,
    journal={Journal of Applied Econometrics},
    author={Aristide Houndetoungan},
    title={Count Data Models With Heterogeneous Peer Effects Under Rational Expectations},
    year={2026},
    month={April},
    pages={295-309},
    volume={41},
    number={3},
    doi={10.1002/jae.70045},
    url={https://ideas.repec.org/a/wly/japmet/v41y2026i3p295-309.html},
}

@inproceedings{de2017econometrics,
  title={Econometrics of network models},
  author={De Paula, Aureo},
  booktitle={Advances in economics and econometrics: Theory and applications, eleventh world congress},
  pages={268--323},
  year={2017},
  organization={Cambridge University Press Cambridge}
}

@article{Chernozhukov2018Double,
    author = {Chernozhukov, Victor and Chetverikov, Denis and Demirer, Mert and Duflo, Esther and Hansen, Christian and Newey, Whitney and Robins, James},
    title = {Double/debiased machine learning for treatment and structural parameters},
    journal = {The Econometrics Journal},
    volume = {21},
    number = {1},
    pages = {C1-C68},
    year = {2018},
    month = {02},
    issn = {1368-4221},
    doi = {10.1111/ectj.12097}
}

@article{bonhomme2026higher,
  title={Higher-Order Neyman Orthogonality in Moment-Condition Models},
  author={Bonhomme, St{\'e}phane and Jochmans, Koen and Newey, Whitney K and Weidner, Martin},
  journal={arXiv preprint arXiv:2605.10842},
  year={2026}
}

@article{carrell2010sex,
  title={Sex and Science: How Professor Gender Perpetuates the Gender Gap},
  author={Carrell, Scott E and Page, Marianne E and West, James E},
  journal={Quarterly Journal of Economics},
  volume={125},
  number={3},
  pages={1101--1144},
  year={2010},
  publisher={Oxford University Press (OUP)}
}

@article{wills1981downward,
  title={Downward comparison principles in social psychology.},
  author={Wills, Thomas A},
  journal={Psychological bulletin},
  volume={90},
  number={2},
  pages={245},
  year={1981},
  publisher={American Psychological Association}
}

@article{brock2007identification,
  title={Identification of binary choice models with social interactions},
  author={Brock, William A and Durlauf, Steven N},
  journal={journal of Econometrics},
  volume={140},
  number={1},
  pages={52--75},
  year={2007},
  publisher={Elsevier}
}

@article{card2012inequality,
  title={Inequality at work: The effect of peer salaries on job satisfaction},
  author={Card, David and Mas, Alexandre and Moretti, Enrico and Saez, Emmanuel},
  journal={American Economic Review},
  volume={102},
  number={6},
  pages={2981--3003},
  year={2012},
  publisher={American Economic Association}
}

@article{genicot2020aspirations,
  title={Aspirations and economic behavior},
  author={Genicot, Garance and Ray, Debraj},
  journal={Annual Review of Economics},
  volume={12},
  number={1},
  pages={715--746},
  year={2020},
  publisher={Annual Reviews}
}

@article{bernard2026future,
  title={The future in mind: Aspirations and long-term outcomes in rural Ethiopia},
  author={Bernard, Tanguy and Dercon, Stefan and Orkin, Kate and Schinaia, Giulio and Taffesse, Alemayehu Seyoum},
  journal={The Quarterly Journal of Economics},
  pages={qjag002},
  year={2026},
  publisher={Oxford University Press}
}

@article{lee2021,
  title={Who is the key player? A network analysis of juvenile delinquency},
  author={Lee, Lung-Fei and Liu, Xiaodong and Patacchini, Eleonora and Zenou, Yves},
  journal={Journal of Business \& Economic Statistics},
  volume={39},
  number={3},
  pages={849--857},
  year={2021},
  publisher={Taylor \& Francis}
}

@article{chamberlain1987,
  title={Asymptotic efficiency in estimation with conditional moment restrictions},
  author={Chamberlain, Gary},
  journal={journal of Econometrics},
  volume={34},
  number={3},
  pages={305--334},
  year={1987},
  publisher={Elsevier}
}

@article{boucher2025estimating,
  title={Estimating peer effects using partial network data},
  author={Boucher, Vincent and Houndetoungan, Aristide},
  journal={Review of Economics and Statistics},
  pages={1--48},
  year={2025},
  publisher={MIT Press 255 Main Street, 9th Floor, Cambridge, Massachusetts 02142, USA~…}
}

@article{lambotte2025,
  title={Heterogeneity in peer effects for binary outcomes},
  author={Lambotte, Mathieu},
  journal={arXiv preprint arXiv:2511.15891},
  year={2025}
}

@article{galeotti2020,
  title={Targeting interventions in networks},
  author={Galeotti, Andrea and Golub, Benjamin and Goyal, Sanjeev},
  journal={Econometrica},
  volume={88},
  number={6},
  pages={2445--2471},
  year={2020},
  publisher={Wiley Online Library}
}

@article{kleibergen200,
  title={Generalized reduced rank tests using the singular value decomposition},
  author={Kleibergen, Frank and Paap, Richard},
  journal={journal of Econometrics},
  volume={133},
  number={1},
  pages={97--126},
  year={2006},
  publisher={Elsevier}
}

@article{boucher2023,
  title={Toward a general theory of peer effects},
  author={Boucher, Vincent and Rendall, Michelle and Ushchev, Philip and Zenou, Yves},
  year={2024},
  journal={Econometrica},
volume={92},
issue={2},
  pages={543--565},
}

@article{blume_linear_2015,
	title = {Linear {Social} {Interactions} {Models}},
	volume = {123},
	issn = {0022-3808, 1537-534X},
	doi = {10.1086/679496},
	language = {en},
	number = {2},
	urldate = {2022-01-25},
	journal = {Journal of Political Economy},
	author = {Blume, Lawrence E. and Brock, William A. and Durlauf, Steven N. and Jayaraman, Rajshri},
	month = apr,
	year = {2015},
	pages = {444--496},
	
}

@article{boucherfortin2015,
  title={Some challenges in the empirics of the effects of networks},
  author={Boucher, Vincent and Fortin, Bernard},
  journal={Oxford Handbook on the Economics of Networks, Oxford: Oxford University Press, forthcoming},
  year={2015}
}

@article{ushchev2020,
  title={Social norms in networks},
  author={Ushchev, Philip and Zenou, Yves},
  journal={Journal of Economic Theory},
  volume={185},
  pages={104969},
  year={2020},
  publisher={Elsevier}
}

@article{lee_binary_2014,
	title = {Binary {Choice} {Models} with {Social} {Network} under {Heterogeneous} {Rational} {Expectations}},
	volume = {96},
		number = {3},
	urldate = {2022-01-25},
	journal = {Review of Economics and Statistics},
	author = {Lee, Lung-fei and Li, Ji and Lin, Xu},
		year = {2014},
	pages = {402--417},
	}

@article{festinger1954theory,
  title={A theory of social comparison processes},
  author={Festinger, Leon},
  journal={Human Relations},
  volume={7},
  number={2},
  pages={117--140},
  year={1954},
  publisher={Sage Publications Sage CA: Thousand Oaks, CA}
}

@article{zenou2025peer,
  title={Peer vs. network effects: Microfoundations, identification, and beyond},
  author={Zenou, Yves},
  journal={Network Effects: Microfoundations, Identification, and Beyond*(November 05, 2025)},
  year={2025}
}

@article{bramoulle2020peer,
  title={Peer effects in networks: A survey},
  author={Bramoull{\'e}, Yann and Djebbari, Habiba and Fortin, Bernard},
  journal={Annual Review of Economics},
  volume={12},
  number={1},
  pages={603--629},
  year={2020},
  publisher={Annual Reviews}
}

@misc{herstad2026identification,
      title={Identification of Heterogeneous Peer Effects}, 
      author={Eyo I. Herstad and Aristide Houndetoungan and Myungkou Shin},
      year={2026},
      eprint={2410.14317},
      archivePrefix={arXiv},
      primaryClass={econ.EM},
      url={https://arxiv.org/abs/2410.14317}, 
}

@article{akerlof2000economics,
  title={Economics and identity},
  author={Akerlof, George A and Kranton, Rachel E},
  journal={The Quarterly Journal of Economics},
  volume={115},
  number={3},
  pages={715--753},
  year={2000},
  publisher={MIT Press}
}

@article{bernheim1994theory,
  title={A theory of conformity},
  author={Bernheim, B Douglas},
  journal={Journal of Political Economy},
  volume={102},
  number={5},
  pages={841--877},
  year={1994},
  publisher={The University of Chicago Press}
}

@article{benabou2006incentives,
  title={Incentives and prosocial behavior},
  author={B{\'e}nabou, Roland and Tirole, Jean},
  journal={American Economic Review},
  volume={96},
  number={5},
  pages={1652--1678},
  year={2006},
  publisher={American Economic Association}
}
\bibliographystyle{ecta}}

\clearpage
\setcounter{page}{1}
\begin{mytitlepage}
\title{Supplemental Appendix to "\TITLE"}
\maketitle

\end{mytitlepage}

\section{Construction of the Instruments} \label{append:instrument}

In this section, we detail how we construct instruments for $\check{y}_{m,i}$ and $\bar{y}_{m,i}$. These instruments are predictions of $\mathbb{E}_m(\check{y}_{m,i})$ and $\mathbb{E}_m(\bar{y}_{m,i})$ obtained using flexible supervised machine learning (ML) methods.

We define some notation for convenience. Let $\mathcal{N}_m$ denote the set of individuals in network $m$. We split the sample across networks into $L \geq 2$ folds labeled $\mathcal{F}_1, \dots, \mathcal{F}_L$, where $L < \infty$. For any fold $\mathcal{F}_l$, we define
\begin{align*}
    \check{\mathcal{S}}_l 
    &= \Big\{(i, j) \in \mathcal{F}_l \times \mathcal{F}_l : \exists~m \ \text{such that}\ \mathcal{N}_m \subset \mathcal{F}_l, \ i, j \in \mathcal{N}_m, \ \text{and} \ g_{m,ij} > 0 \Big\}, \\
    \bar{\mathcal{S}}_l 
    &= \Big\{ i \in \mathcal{F}_l : \exists~m \ \text{such that}\ \mathcal{N}_m  \subset \mathcal{F}_l, \ i \in \mathcal{N}_m, \ \text{and} \ n_{m,i} > 0 \Big\}.
\end{align*}
The set $\check{\mathcal{S}}_l$ includes all pairs of individuals $(i, j)$ from the same network $\mathcal{N}_m$ in fold $\mathcal{F}_l$ such that $g_{m,ij} > 0$, whereas $\bar{\mathcal{S}}_l$ includes all individuals $i$ from some network $\mathcal{N}_m$ who have at least one friend.

We also define
$$
\check{\mathcal{S}}_{-l} = \bigcup_{l' \ne l}^L \check{\mathcal{S}}_{l'} 
\quad \text{and} \quad 
\bar{\mathcal{S}}_{-l} = \bigcup_{l' \ne l}^L \bar{\mathcal{S}}_{l'}.
$$
The set $\check{\mathcal{S}}_{-l}$ contains all pairs of connected individuals not in $\check{\mathcal{S}}_l$, and $\bar{\mathcal{S}}_{-l}$ includes all non-isolated individuals not in $\bar{\mathcal{S}}_l$.

\subsection{Instrument for \texorpdfstring{$\check{y}_{m,i}$}{check y}}\label{append:instrument:check}

For any $i \in \mathcal{N}_m$ with $\mathcal{N}_m \subset \mathcal{F}_l$, we can write $\mathbb{E}_m(\check{y}_{m,i})$ from Equation \eqref{eq:Eycheck} as
\begin{equation*}
    \mathbb{E}_m(\check{y}_{m,i}) 
    = \sum_{(i, j) \in \check{\mathcal{S}}_l} g_{m,ij}\,
    \underbrace{\mathbb{P}_m\{\Delta^y_{m,ji} > 0\}}_{\text{Extensive margin}} \,
    \underbrace{\mathbb{E}_m(\Delta^y_{m,ji} \mid \Delta^y_{m,ji} > 0)}_{\text{Intensive margin}}.
\end{equation*}

Consequently, we focus on predicting extensive and intensive margins for pairs $(i, j) \in \check{\mathcal{S}}_l$. Since $\check{\mathbf{w}}_{m,i} = \big(\mathbf{x}_{m,i}^{\prime}, \, \mathbf{g}_{m,i}\mathbf{X}_m,\, \mathbf{g}_{m,i}\mathbf{G}_m\mathbf{X}_m\big)^{\prime}$, where $\mathbf{g}_{m,i}$ is the $i$-th row of $\mathbf{G}_m$, we define $\check{\mathbf{w}}_{m,ji} = \check{\mathbf{w}}_{m,j} - \check{\mathbf{w}}_{m,i}$ as the predictor for both margins.

For any $(i, j) \in \check{\mathcal{S}}_l$, with $i, j \in \mathcal{N}_m$, we have
\begin{align*}
    \hat{\mathbb{P}}_m\{\Delta^y_{m,ji} > 0\} &= \check{h}_l^{\text{ext}}(\check{\mathbf{w}}_{m,ji}), \\
    \hat{\mathbb{E}}_m(\Delta^y_{m,ji} \mid \Delta^y_{m,ji} > 0) &= \check{h}_l^{\text{int}}(\check{\mathbf{w}}_{m,ji}),
\end{align*}
where $\hat{\mathbb{P}}_m$ and $\hat{\mathbb{E}}_m$ denote the predicted probability and expectation, respectively. The mappings $\check{h}_l^{\text{ext}}$ and $\check{h}_l^{\text{int}}$ are estimated using supervised ML methods. We rely on random forests, although other methods could also be used.

Specifically, to obtain $\check{h}_l^{\text{ext}}$, we train random forests where the target variable is $\mathbbm{1}\{\Delta^y_{m,ji} > 0\}$ and the predictor is $\check{\mathbf{w}}_{m,ji}$. This training step is performed using data from $\check{\mathcal{S}}_{-l}$ only, ensuring that $\check{h}_l^{\text{ext}}$ is independent of the data in $\check{\mathcal{S}}_l$. Consequently, $\check{h}_l^{\text{ext}}(\check{\mathbf{w}}_{m,ji})$ is independent of $\boldsymbol{\varepsilon}_m$.

Similarly, to obtain $\check{h}_l^{\text{int}}$, we train random forests where the target variable is $\Delta^y_{m,ji}$ and the predictor is $\check{\mathbf{w}}_{m,ji}$. The training is performed using data from the sample
$$\check{\mathcal S}_{-l}^{int} = \Big\{(i, j) \in \check{\mathcal S}_{-l}: ~ i, \, j \in \mathcal N_{m^\prime}, \, \text{and} \, \Delta^y_{m^\prime,ji} > 0\Big\}.$$ 
This ensures that $\check{h}_l^{\text{int}}(\check{\mathbf{w}}_{m,ji})$ is independent of $\boldsymbol{\varepsilon}_m$, while being interpretable as a prediction of a conditional expectation given $\Delta^y_{m,ji} > 0$.

The rationale for the cross-fitting approach is that $\check{h}_l^{\text{ext}}$ and $\check{h}_l^{\text{int}}$ do not depend on $l$ asymptotically. Therefore, the training can be performed excluding data from $\check{\mathcal{S}}_l$.

Finally, we construct the instrument $\hat{\mathbb{E}}_m(\check{y}_{m,i})$ as
\begin{equation*}
    \hat{\mathbb{E}}_m(\check{y}_{m,i}) 
    = \sum_{j \ne i} g_{m,ij}\,
    \hat{\mathbb{P}}_m\{\Delta^y_{m,ji} > 0\} \,
    \hat{\mathbb{E}}_m(\Delta^y_{m,ji} \mid \Delta^y_{m,ji} > 0).
\end{equation*}

\subsection{Instrument for \texorpdfstring{$\bar{y}_{m,i}$}{bar y}}
As discussed in Section \ref{sec:estim}, given the prediction $\hat{\mathbb{E}}_m(\check{y}_{m,i})$ obtained above,
$$
\bar{\mathbf{w}}_{m,i} 
= \big(
\mathbf{g}_{m,i}\mathbf{X}_m,\,
\mathbf{g}_{m,i}\mathbf{G}_m\mathbf{X}_m,\,
\mathbf{g}_{m,i} \hat{\mathbb{E}}_m(\check{\mathbf{y}}_m),\,
\mathbf{g}_{m,i}\mathbf{G}_m \hat{\mathbb{E}}_m(\check{\mathbf{y}}_m)
\big)^{\prime}
$$
is a natural predetermined predictor of $\mathbb{E}_m(\bar{y}_{m,i})$.

For any $i \in \bar{\mathcal{S}}_l$, with $i \in \mathcal{N}_m$, we thus have
\begin{equation*}
   \hat{\mathbb{E}}_m(\bar{y}_{m,i}) = \bar{h}_l(\bar{\mathbf{w}}_{m,i}).
\end{equation*}

To obtain $\bar{h}_l$, we train random forests where the target variable is $\bar{y}_{m,i}$ and the predictor is $\bar{\mathbf{w}}_{m,i}$. As in Section \ref{append:instrument:check}, this training is performed using data from $\bar{\mathcal{S}}_{-l}$, ensuring that $\bar{h}_l(\bar{\mathbf{w}}_{m,i})$ is independent of $\boldsymbol{\varepsilon}_m$.

\section{Summary Statistics and Detailed Empirical Results}\label{append:empirics}

\begin{landscape}
\begin{table}[htbp]
\footnotesize
\centering
\caption{Summary Statistics}
\label{tab:sumstats}
\begin{threeparttable}
\begin{tabular}{l cccc cccc cccc cccc}
\toprule
& \multicolumn{4}{c}{Smoke} & \multicolumn{4}{c}{Fight} & \multicolumn{4}{c}{Optimism} & \multicolumn{4}{c}{Drink} \\
\cmidrule(lr){2-5} \cmidrule(lr){6-9} \cmidrule(lr){10-13} \cmidrule(lr){14-17}
 & \multicolumn{1}{c}{Mean} & \multicolumn{1}{c}{SD} & \multicolumn{1}{c}{Min} & \multicolumn{1}{c}{Max} 
 & \multicolumn{1}{c}{Mean} & \multicolumn{1}{c}{SD} & \multicolumn{1}{c}{Min} & \multicolumn{1}{c}{Max} 
 & \multicolumn{1}{c}{Mean} & \multicolumn{1}{c}{SD} & \multicolumn{1}{c}{Min} & \multicolumn{1}{c}{Max} 
 & \multicolumn{1}{c}{Mean} & \multicolumn{1}{c}{SD} & \multicolumn{1}{c}{Min} & \multicolumn{1}{c}{Max}  \\
\midrule
Outcome            & 0.925  & 2.189 & 0   & 7   & 1.336  & 2.159 & 0   & 8   & 6.519   & 1.77   & 0   & 8   & 0.402  & 1.114 & 0   & 7   \\
Age                & 15.047 & 1.695 & 10  & 19  & 15.083 & 1.683 & 10  & 19  & 15.035  & 1.694  & 10  & 19  & 15.047 & 1.695 & 10  & 19  \\
Female             & 0.509  & 0.5   & 0   & 1   & 0.511  & 0.5   & 0   & 1   & 0.51    & 0.5    & 0   & 1   & 0.509  & 0.5   & 0   & 1   \\
Grade              & 9.665  & 1.596 & 6   & 12  & 9.71   & 1.587 & 6   & 12  & 9.655   & 1.599  & 6   & 12  & 9.665  & 1.597 & 6   & 12  \\
Hispanic           & 0.168  & 0.374 & 0   & 1   & 0.161  & 0.367 & 0   & 1   & 0.167   & 0.373  & 0   & 1   & 0.168  & 0.374 & 0   & 1   \\
\multicolumn{17}{l}{\textit{Race}} \\
\quad White        & 0.635  & 0.481 & 0   & 1   & 0.645  & 0.479 & 0   & 1   & 0.636   & 0.481  & 0   & 1   & 0.635  & 0.481 & 0   & 1   \\
\quad Black        & 0.174  & 0.379 & 0   & 1   & 0.17   & 0.376 & 0   & 1   & 0.173   & 0.378  & 0   & 1   & 0.174  & 0.379 & 0   & 1   \\
\quad Asian        & 0.069  & 0.253 & 0   & 1   & 0.069  & 0.253 & 0   & 1   & 0.068   & 0.252  & 0   & 1   & 0.069  & 0.253 & 0   & 1   \\
\multicolumn{17}{l}{\textit{Mother education}} \\
\quad $<$ High     & 0.17   & 0.375 & 0   & 1   & 0.164  & 0.371 & 0   & 1   & 0.171   & 0.376  & 0   & 1   & 0.17   & 0.375 & 0   & 1   \\
\quad $>$ High     & 0.405  & 0.491 & 0   & 1   & 0.412  & 0.492 & 0   & 1   & 0.407   & 0.491  & 0   & 1   & 0.405  & 0.491 & 0   & 1   \\
\quad Missing      & 0.127  & 0.333 & 0   & 1   & 0.124  & 0.33  & 0   & 1   & 0.124   & 0.329  & 0   & 1   & 0.127  & 0.333 & 0   & 1   \\
\multicolumn{17}{l}{\textit{Mother job}} \\
\quad Professional & 0.197  & 0.398 & 0   & 1   & 0.201  & 0.4   & 0   & 1   & 0.198   & 0.399  & 0   & 1   & 0.197  & 0.398 & 0   & 1   \\
\quad Other        & 0.422  & 0.494 & 0   & 1   & 0.423  & 0.494 & 0   & 1   & 0.424   & 0.494  & 0   & 1   & 0.422  & 0.494 & 0   & 1   \\
\quad Missing      & 0.179  & 0.383 & 0   & 1   & 0.175  & 0.38  & 0   & 1   & 0.176   & 0.381  & 0   & 1   & 0.179  & 0.383 & 0   & 1   \\
\multicolumn{17}{l}{\textit{Degree}} \\
\quad Matched      & 3.604  & 2.858 & 0   & 10  & 3.519  & 2.825 & 0   & 10  & 3.656   & 2.864  & 0   & 10  & 3.598  & 2.854 & 0   & 10  \\
\quad Unmatched    & 1.726  & 1.871 & 0   & 10  & 1.811  & 1.916 & 0   & 10  & 1.716   & 1.861  & 0   & 10  & 1.732  & 1.873 & 0   & 10  \\
\midrule
$n$ (students)     & \multicolumn{4}{c}{74{,}840}  & \multicolumn{4}{c}{71{,}544}  & \multicolumn{4}{c}{75{,}302}    & \multicolumn{4}{c}{74{,}677}  \\
$M$ (schools)      & \multicolumn{4}{c}{140}    & \multicolumn{4}{c}{140}    & \multicolumn{4}{c}{140}      & \multicolumn{4}{c}{140}   \\
\bottomrule
\end{tabular}
\begin{tablenotes}
\footnotesize
\item Notes: The sample sizes differ across outcomes because of missing outcome values. 
\end{tablenotes}
\end{threeparttable}
\end{table}
\end{landscape}

\begin{scriptsize}
\setlength{\tabcolsep}{10.18pt}
\renewcommand{\arraystretch}{0.8}
\begin{longtable}{p{2.4cm}cccccccc}
\caption{Estimation Results}\label{tab:estimates}\\ 
    \toprule
    \textbf{Outcome} & \multicolumn{2}{c}{\textbf{Smoking}} &\multicolumn{2}{c}{\textbf{Fighting}} & \multicolumn{2}{c}{\textbf{Optimism}} & \multicolumn{2}{c}{\textbf{Drinking}} \\
     Model & Asym. & Sym.  & Asym. & Sym. & Asym. & Sym. & Asym. & Sym. \\
    \midrule
    \endfirsthead
 \multicolumn{9}{r}{\tablename\ \thetable{} -- continued from previous page} \\
\midrule
\textbf{Outcome} & \multicolumn{2}{c}{\textbf{Smoking}} &\multicolumn{2}{c}{\textbf{Fighting}} & \multicolumn{2}{c}{\textbf{Optimism}} & \multicolumn{2}{c}{\textbf{Drinking}}  \\
     Model & Asym. & Sym.  & Asym. & Sym. & Asym. & Sym. & Asym. & Sym. \\
\midrule
\endhead
\midrule
    \multicolumn{9}{r}{\scriptsize \textit{Continued on next page}} \\
\endfoot
\caption*{\justifying \scriptsize Notes: Estimates are reported without parentheses, with standard errors (clustered at the network level) shown in parentheses. To generate the instruments, we use random forests with 5-fold cross-fitting (see Supplemental Appendix~\ref{append:instrument}) and 1,500 trees. For each training step, 20\% of the sample is used to tune the remaining hyperparameters. The row labeled ``KP LM test p-value” reports the p-value of the rank LM test for underidentification.}
\endlastfoot
          \multicolumn{9}{l}{\textbf{Endogenous Peer effects}}\\
$\beta^l$ & 1.219 & & -0.065 & & 3.694 & & 0.842 & \\ 
~ & (0.298) & & (0.087) & & (0.625) & & (0.184) & \\ 
 
$\beta^h$ & 3.448 & & 1.115 & & -0.201 & & 1.206 & \\ 
~ & (0.449) & & (0.179) & & (0.162) & & (0.296) & \\ 
 
$\beta$ & & 2.839 & & 0.232 & & 0.623 & & 0.897 \\ 
~ & & (0.458) & & (0.059) & & (0.093) & & (0.172) \\ 

 Test Symmetry & 2.229 & & 1.180 & & -3.896 & & 0.364 & \\ 
~ & (0.281) & & (0.232) & & (0.737) & & (0.285) & \\

          \multicolumn{9}{l}{\textbf{Own Effects}}\\
 Age & 0.278 & 0.289 & 0.161 & 0.188 & -0.278 & -0.309 & 0.144 & 0.147 \\ 
~ & (0.034) & (0.036) & (0.021) & (0.021) & (0.024) & (0.020) & (0.024) & (0.024) \\ 
 
Female & 0.013 & 0.013 & -0.950 & -0.931 & 0.309 & 0.358 & -0.334 & -0.334 \\ 
~ & (0.041) & (0.040) & (0.043) & (0.043) & (0.034) & (0.029) & (0.031) & (0.031) \\ 
 
Grade & -0.175 & -0.183 & -0.349 & -0.381 & 0.294 & 0.320 & -0.078 & -0.080 \\ 
~ & (0.041) & (0.044) & (0.027) & (0.025) & (0.029) & (0.027) & (0.029) & (0.029) \\ 
 
Hispanic & 0.013 & 0.021 & 0.150 & 0.197 & -0.444 & -0.374 & 0.199 & 0.201 \\ 
~ & (0.077) & (0.087) & (0.053) & (0.055) & (0.064) & (0.048) & (0.037) & (0.037) \\ 
 
\multicolumn{9}{l}{\textit{Race}} \\
\quad White & 0.401 & 0.450 & -0.070 & -0.054 & 0.032 & 0.051 & 0.061 & 0.065 \\ 
~ & (0.070) & (0.078) & (0.046) & (0.048) & (0.053) & (0.043) & (0.030) & (0.030) \\ 
 
\quad Black & -0.507 & -0.524 & 0.129 & 0.148 & -0.032 & -0.018 & -0.005 & -0.002 \\ 
~ & (0.075) & (0.081) & (0.066) & (0.065) & (0.065) & (0.051) & (0.048) & (0.048) \\ 
 
\quad Asian & 0.200 & 0.168 & 0.115 & 0.068 & -0.066 & 0.012 & 0.205 & 0.200 \\ 
~ & (0.077) & (0.088) & (0.056) & (0.057) & (0.082) & (0.062) & (0.056) & (0.056) \\ 
 
\multicolumn{9}{l}{Mother education} \\
\quad < High & 0.069 & 0.056 & 0.020 & 0.039 & -0.278 & -0.272 & 0.021 & 0.022 \\ 
~ & (0.051) & (0.058) & (0.037) & (0.038) & (0.051) & (0.036) & (0.027) & (0.027) \\ 
 
\quad > High & -0.064 & -0.165 & -0.042 & -0.098 & 0.402 & 0.404 & -0.009 & -0.013 \\ 
~ & (0.038) & (0.043) & (0.027) & (0.025) & (0.048) & (0.030) & (0.016) & (0.016) \\ 
 
\quad Missing & 0.317 & 0.267 & 0.183 & 0.198 & -0.287 & -0.230 & 0.175 & 0.176 \\ 
~ & (0.056) & (0.060) & (0.043) & (0.041) & (0.059) & (0.044) & (0.040) & (0.040) \\ 
 
\multicolumn{9}{l}{\textit{Mother job}} \\
\quad Professional & 0.104 & 0.119 & 0.099 & 0.097 & 0.198 & 0.185 & 0.071 & 0.073 \\ 
~ & (0.053) & (0.056) & (0.032) & (0.031) & (0.051) & (0.030) & (0.023) & (0.023) \\ 
 
\quad Other & 0.238 & 0.280 & 0.094 & 0.109 & 0.071 & 0.047 & 0.065 & 0.068 \\ 
~ & (0.054) & (0.061) & (0.028) & (0.026) & (0.037) & (0.023) & (0.020) & (0.020) \\ 
 
\quad Missing & 0.156 & 0.222 & 0.119 & 0.147 & -0.069 & -0.131 & 0.036 & 0.039 \\ 
~ & (0.058) & (0.062) & (0.044) & (0.039) & (0.052) & (0.034) & (0.024) & (0.023) \\ 
\multicolumn{9}{l}{\textit{Degree}} \\
\quad Matched & -0.106 & -0.149 & -0.037 & -0.034 & 0.079 & 0.085 & -0.011 & -0.010 \\ 
~ & (0.017) & (0.022) & (0.005) & (0.005) & (0.009) & (0.006) & (0.004) & (0.004) \\ 
\quad Unmatched & 0.054 & 0.054 & 0.048 & 0.052 & 0.009 & -0.006 & 0.035 & 0.035 \\ 
        ~ & (0.012) & (0.012) & (0.008) & (0.010) & (0.007) & (0.007) & (0.006) & (0.006) \\ 
         \multicolumn{9}{l}{\textbf{Contextual Effects}}\\
Age & -0.170 & 0.172 & -0.039 & 0.098 & 0.350 & -0.134 & 0.005 & 0.041 \\ 
~ & (0.108) & (0.097) & (0.049) & (0.033) & (0.124) & (0.036) & (0.041) & (0.028) \\ 
 
Female & 0.320 & 0.343 & 0.619 & 0.217 & -0.464 & -0.153 & 0.286 & 0.237 \\ 
~ & (0.109) & (0.117) & (0.088) & (0.055) & (0.104) & (0.041) & (0.056) & (0.043) \\ 
 
Grade & 0.109 & -0.047 & 0.173 & -0.064 & -0.398 & 0.102 & -0.000 & -0.024 \\ 
~ & (0.105) & (0.106) & (0.066) & (0.038) & (0.128) & (0.043) & (0.041) & (0.034) \\ 
 
Hispanic & -0.176 & -0.191 & -0.096 & 0.055 & 0.588 & -0.060 & -0.086 & -0.053 \\ 
~ & (0.175) & (0.223) & (0.086) & (0.075) & (0.198) & (0.083) & (0.069) & (0.060) \\ 
 
\multicolumn{9}{l}{\textit{Race}} \\
\quad White & -0.244 & 0.272 & 0.017 & 0.080 & -0.343 & -0.021 & 0.068 & 0.086 \\ 
~ & (0.167) & (0.168) & (0.079) & (0.064) & (0.166) & (0.077) & (0.056) & (0.052) \\ 
 
\quad Black & 0.681 & -0.022 & -0.113 & 0.029 & -0.070 & 0.008 & 0.032 & 0.039 \\ 
~ & (0.178) & (0.188) & (0.101) & (0.091) & (0.152) & (0.081) & (0.065) & (0.061) \\ 
 
\quad Asian & -0.240 & -0.483 & -0.317 & -0.376 & -0.096 & 0.155 & -0.224 & -0.243 \\ 
~ & (0.195) & (0.227) & (0.094) & (0.088) & (0.182) & (0.091) & (0.068) & (0.065) \\ 
 
\multicolumn{9}{l}{\textit{Mother education}} \\
\quad < High & 0.160 & 0.249 & 0.047 & 0.089 & 0.320 & -0.075 & -0.068 & -0.069 \\ 
~ & (0.179) & (0.215) & (0.068) & (0.063) & (0.157) & (0.071) & (0.046) & (0.044) \\ 
 
\quad > High & 0.078 & -0.223 & -0.111 & -0.213 & -0.525 & 0.147 & -0.051 & -0.060 \\ 
~ & (0.129) & (0.150) & (0.060) & (0.049) & (0.153) & (0.049) & (0.041) & (0.040) \\ 
 
\quad Missing & -0.032 & 0.028 & -0.015 & 0.106 & 0.209 & 0.015 & -0.018 & -0.009 \\ 
~ & (0.194) & (0.234) & (0.098) & (0.080) & (0.205) & (0.089) & (0.072) & (0.068) \\ 
 
\multicolumn{9}{l}{\textit{Mother job}} \\
\quad Professional & 0.017 & 0.126 & 0.068 & 0.013 & -0.277 & -0.079 & -0.015 & -0.001 \\ 
~ & (0.166) & (0.205) & (0.066) & (0.058) & (0.129) & (0.066) & (0.055) & (0.051) \\ 
 
\quad Other & 0.078 & 0.386 & 0.086 & 0.106 & -0.122 & -0.129 & 0.045 & 0.067 \\ 
~ & (0.134) & (0.160) & (0.064) & (0.054) & (0.107) & (0.056) & (0.049) & (0.043) \\ 
 
\quad Missing & -0.096 & 0.255 & 0.073 & 0.176 & 0.314 & -0.137 & 0.044 & 0.076 \\ 
~ & (0.207) & (0.240) & (0.094) & (0.077) & (0.181) & (0.069) & (0.071) & (0.061) \\
\multicolumn{9}{l}{\textit{Degree}} \\
\quad Matched  & 0.035 & -0.051 & -0.001 & -0.039 & -0.139 & 0.028 & 0.004 & -0.004 \\ 
        ~ & (0.025) & (0.026) & (0.009) & (0.006) & (0.034) & (0.009) & (0.009) & (0.007) \\ 
\quad Unmatched & -0.094 & 0.025 & -0.003 & 0.038 & 0.008 & -0.026 & -0.001 & 0.009 \\ 
        ~ & (0.038) & (0.035) & (0.015) & (0.012) & (0.028) & (0.013) & (0.011) & (0.008) \\  
        \midrule
F statistic ($\bar y$)	&88.899	&171.597	&42.730	&64.382	&84.574	&103.343	&25.228	&40.992 \\
F statistic 
 ($\check y$)	&62.752	& &23.597	&	&22.324&	&18.574	& \\
        KP LM test p-value & 0.000 & 0.000& 0.002 & 0.000 & 0.000 & 0.000& 0.000& 0.000\\ 
        \bottomrule
\end{longtable}
\end{scriptsize}

\clearpage
\begin{scriptsize}
\setlength{\tabcolsep}{10.18pt}
\renewcommand{\arraystretch}{0.8}
\begin{longtable}{p{2.4cm}cccccccc}
\caption{Estimation Results Without Fake Isolated}\label{tab:estimates_wofakeiso}\\ 
    \toprule
    \textbf{Outcome} & \multicolumn{2}{c}{\textbf{Smoking}} &\multicolumn{2}{c}{\textbf{Fighting}} & \multicolumn{2}{c}{\textbf{Optimism}} & \multicolumn{2}{c}{\textbf{Drinking}} \\
     Model & Asym. & Sym.  & Asym. & Sym. & Asym. & Sym. & Asym. & Sym. \\
    \midrule
    \endfirsthead
 \multicolumn{9}{r}{\tablename\ \thetable{} -- continued from previous page} \\
\midrule
\textbf{Outcome} & \multicolumn{2}{c}{\textbf{Smoking}} &\multicolumn{2}{c}{\textbf{Fighting}} & \multicolumn{2}{c}{\textbf{Optimism}} & \multicolumn{2}{c}{\textbf{Drinking}}  \\
     Model & Asym. & Sym.  & Asym. & Sym. & Asym. & Sym. & Asym. & Sym. \\
\midrule
\endhead
\midrule
    \multicolumn{9}{r}{\scriptsize \textit{Continued on next page}} \\
\endfoot
\caption*{\justifying \scriptsize Notes: Estimates are reported without parentheses, with standard errors (clustered at the network level) shown in parentheses. To generate the instruments, we use random forests with 5-fold cross-fitting (see Supplemental Appendix~\ref{append:instrument}) and 1,500 trees. For each training step, 20\% of the sample is used to tune the remaining hyperparameters. The row labeled ``KP LM test p-value” reports the p-value of the rank LM test for underidentification.}
\endlastfoot
          \multicolumn{9}{l}{\textbf{Endogenous Peer effects}}\\
$\beta^l$ & 1.290 & & -0.164 & & 3.003 & & 0.928 & \\ 
~ & (0.360) & & (0.085) & & (0.536) & & (0.202) & \\ 
 
$\beta^h$ & 3.452 & & 1.274 & & 0.003 & & 1.080 & \\ 
~ & (0.476) & & (0.191) & & (0.150) & & (0.291) & \\ 
 
$\beta$ & & 2.995 & & 0.227 & & 0.654 & & 0.951 \\ 
~ & & (0.525) & & (0.058) & & (0.110) & & (0.181) \\ 

 $\beta^h - \beta^l$ & 2.162 & & 1.437 & & -3.000 & & 0.152 & \\ 
~ & (0.281) & & (0.241) & & (0.617) & & (0.313) & \\

          \multicolumn{9}{l}{\textbf{Own Effects}}\\
 Age & 0.281 & 0.294 & 0.145 & 0.184 & -0.303 & -0.329 & 0.145 & 0.146 \\ 
~ & (0.038) & (0.041) & (0.021) & (0.021) & (0.029) & (0.025) & (0.027) & (0.026) \\ 
 
Female & -0.016 & -0.030 & -0.938 & -0.934 & 0.360 & 0.389 & -0.340 & -0.340 \\ 
~ & (0.049) & (0.051) & (0.042) & (0.041) & (0.042) & (0.037) & (0.032) & (0.032) \\ 
 
Grade & -0.179 & -0.189 & -0.318 & -0.364 & 0.301 & 0.326 & -0.074 & -0.076 \\ 
~ & (0.042) & (0.046) & (0.028) & (0.026) & (0.033) & (0.030) & (0.031) & (0.031) \\ 
 
Hispanic & -0.008 & -0.001 & 0.119 & 0.180 & -0.418 & -0.354 & 0.184 & 0.185 \\ 
~ & (0.071) & (0.078) & (0.053) & (0.054) & (0.066) & (0.049) & (0.035) & (0.034) \\ 
 
\multicolumn{9}{l}{\textit{Race}} \\
\quad White & 0.399 & 0.458 & -0.062 & -0.035 & 0.004 & 0.039 & 0.071 & 0.073 \\ 
~ & (0.074) & (0.081) & (0.046) & (0.047) & (0.057) & (0.046) & (0.033) & (0.032) \\ 
 
\quad Black & -0.487 & -0.502 & 0.124 & 0.154 & -0.078 & -0.041 & 0.012 & 0.013 \\ 
~ & (0.073) & (0.082) & (0.069) & (0.065) & (0.063) & (0.052) & (0.050) & (0.050) \\ 
 
\quad Asian & 0.257 & 0.236 & 0.129 & 0.064 & -0.066 & 0.019 & 0.227 & 0.225 \\ 
~ & (0.084) & (0.096) & (0.056) & (0.052) & (0.088) & (0.064) & (0.060) & (0.061) \\ 
 
\multicolumn{9}{l}{Mother education} \\
\quad < High & 0.089 & 0.082 & 0.017 & 0.045 & -0.250 & -0.260 & 0.022 & 0.023 \\ 
~ & (0.053) & (0.059) & (0.038) & (0.038) & (0.057) & (0.040) & (0.025) & (0.025) \\ 
 
\quad > High & -0.025 & -0.133 & -0.014 & -0.080 & 0.420 & 0.415 & -0.000 & -0.002 \\ 
~ & (0.045) & (0.052) & (0.027) & (0.025) & (0.056) & (0.035) & (0.018) & (0.018) \\ 
 
\quad Missing & 0.331 & 0.277 & 0.166 & 0.189 & -0.237 & -0.201 & 0.168 & 0.168 \\ 
~ & (0.059) & (0.065) & (0.045) & (0.043) & (0.066) & (0.046) & (0.037) & (0.036) \\ 
 
\multicolumn{9}{l}{\textit{Mother job}} \\
\quad Professional & 0.098 & 0.113 & 0.092 & 0.088 & 0.205 & 0.187 & 0.074 & 0.075 \\ 
~ & (0.062) & (0.068) & (0.033) & (0.031) & (0.057) & (0.033) & (0.027) & (0.027) \\ 
 
\quad Other & 0.214 & 0.254 & 0.085 & 0.104 & 0.051 & 0.036 & 0.063 & 0.064 \\ 
~ & (0.050) & (0.055) & (0.029) & (0.027) & (0.043) & (0.025) & (0.022) & (0.022) \\ 
 
\quad Missing & 0.114 & 0.177 & 0.103 & 0.139 & -0.071 & -0.138 & 0.033 & 0.035 \\ 
~ & (0.062) & (0.068) & (0.047) & (0.042) & (0.055) & (0.035) & (0.027) & (0.027) \\ 
\multicolumn{9}{l}{\textit{Degree}} \\
\quad Matched & -0.107 & -0.150 & -0.037 & -0.034 & 0.080 & 0.085 & -0.010 & -0.010 \\ 
~ & (0.018) & (0.024) & (0.005) & (0.005) & (0.009) & (0.007) & (0.004) & (0.004) \\ 
\quad Unmatched & 0.114 & 0.217 & 0.053 & 0.069 & -0.002 & -0.026 & 0.059 & 0.060 \\ 
        ~ & (0.027) & (0.040) & (0.009) & (0.009) & (0.012) & (0.009) & (0.009) & (0.008) \\ 
         \multicolumn{9}{l}{\textbf{Contextual Effects}}\\
Age & -0.161 & 0.159 & -0.063 & 0.100 & 0.243 & -0.127 & 0.023 & 0.039 \\ 
~ & (0.109) & (0.102) & (0.050) & (0.033) & (0.104) & (0.036) & (0.041) & (0.030) \\ 
 
Female & 0.342 & 0.380 & 0.685 & 0.217 & -0.420 & -0.176 & 0.266 & 0.246 \\ 
~ & (0.110) & (0.125) & (0.096) & (0.055) & (0.091) & (0.045) & (0.056) & (0.044) \\ 
 
Grade & 0.107 & -0.033 & 0.204 & -0.078 & -0.270 & 0.109 & -0.016 & -0.026 \\ 
~ & (0.106) & (0.112) & (0.069) & (0.039) & (0.108) & (0.043) & (0.041) & (0.036) \\ 
 
Hispanic & -0.163 & -0.184 & -0.113 & 0.065 & 0.425 & -0.067 & -0.059 & -0.046 \\ 
~ & (0.177) & (0.230) & (0.087) & (0.074) & (0.172) & (0.085) & (0.067) & (0.060) \\ 
 
\multicolumn{9}{l}{\textit{Race}} \\
\quad White & -0.222 & 0.278 & -0.004 & 0.071 & -0.260 & -0.015 & 0.075 & 0.082 \\ 
~ & (0.169) & (0.174) & (0.081) & (0.063) & (0.146) & (0.079) & (0.056) & (0.053) \\ 
 
\quad Black & 0.635 & -0.058 & -0.140 & 0.026 & -0.029 & 0.029 & 0.022 & 0.024 \\ 
~ & (0.178) & (0.191) & (0.105) & (0.091) & (0.132) & (0.082) & (0.066) & (0.065) \\ 
 
\quad Asian & -0.282 & -0.520 & -0.301 & -0.370 & -0.037 & 0.151 & -0.250 & -0.258 \\ 
~ & (0.198) & (0.233) & (0.099) & (0.087) & (0.161) & (0.094) & (0.071) & (0.069) \\ 
 
\multicolumn{9}{l}{\textit{Mother education}} \\
\quad < High & 0.171 & 0.272 & 0.040 & 0.092 & 0.224 & -0.077 & -0.065 & -0.066 \\ 
~ & (0.182) & (0.225) & (0.070) & (0.063) & (0.136) & (0.072) & (0.045) & (0.044) \\ 
 
\quad > High & 0.066 & -0.215 & -0.093 & -0.215 & -0.375 & 0.140 & -0.058 & -0.062 \\ 
~ & (0.131) & (0.155) & (0.061) & (0.048) & (0.139) & (0.051) & (0.042) & (0.040) \\ 
 
\quad Missing & -0.029 & 0.032 & -0.041 & 0.106 & 0.162 & 0.015 & -0.013 & -0.009 \\ 
~ & (0.196) & (0.244) & (0.101) & (0.080) & (0.178) & (0.090) & (0.072) & (0.070) \\ 
 
\multicolumn{9}{l}{\textit{Mother job}} \\
\quad Professional & 0.019 & 0.126 & 0.077 & 0.012 & -0.234 & -0.084 & -0.010 & -0.005 \\ 
~ & (0.169) & (0.214) & (0.069) & (0.057) & (0.111) & (0.069) & (0.056) & (0.052) \\ 
 
\quad Other & 0.086 & 0.379 & 0.081 & 0.106 & -0.123 & -0.130 & 0.055 & 0.064 \\ 
~ & (0.137) & (0.165) & (0.066) & (0.054) & (0.094) & (0.057) & (0.051) & (0.044) \\ 
 
\quad Missing & -0.091 & 0.230 & 0.053 & 0.175 & 0.212 & -0.132 & 0.059 & 0.072 \\ 
~ & (0.210) & (0.250) & (0.097) & (0.077) & (0.153) & (0.071) & (0.071) & (0.063) \\
\multicolumn{9}{l}{\textit{Degree}} \\
\quad Matched  & 0.036 & -0.039 & 0.006 & -0.039 & -0.101 & 0.026 & 0.002 & -0.002 \\ 
        ~ & (0.026) & (0.027) & (0.010) & (0.006) & (0.028) & (0.009) & (0.009) & (0.008) \\ 
\quad Unmatched & -0.110 & -0.032 & -0.013 & 0.033 & 0.003 & -0.019 & -0.004 & 0.000 \\ 
        ~ & (0.040) & (0.042) & (0.015) & (0.011) & (0.024) & (0.014) & (0.011) & (0.009) \\  
        \midrule
F statistic ($\bar y$)	&88.899	&171.597	&42.730	&64.382	&84.574	&103.343	&25.228	&40.992 \\
F statistic 
 ($\check y$)	&62.752	& &23.597	&	&22.324&	&18.574	& \\
        KP LM test p-value & 0.000 & 0.000& 0.002 & 0.000 & 0.000 & 0.000& 0.000& 0.000\\ 
        \bottomrule
\end{longtable}
\end{scriptsize}

\end{document}